\pdfoutput=1
\RequirePackage{ifpdf}
\ifpdf
\documentclass[pdftex]{sigma}
\else
\documentclass{sigma}
\fi

\numberwithin{equation}{section}

\newtheorem{Theorem}{Theorem}[section]
\newtheorem*{Theorem*}{Theorem}
\newtheorem{Corollary}[Theorem]{Corollary}

\newtheorem{Proposition}[Theorem]{Proposition}

\theoremstyle{definition}

\newtheorem{Remark}[Theorem]{Remark}

\begin{document}

\allowdisplaybreaks

\newcommand{\arXivNumber}{2605.07659}

\renewcommand{\PaperNumber}{118}

\FirstPageHeading

\ShortArticleName{Superintegrability in the Interaction of Two Particles with Spin}
\ArticleName{Superintegrability in the Interaction of Two Particles\\
with Spin: First-Order Pseudo-Scalar Integrals\\
of Motion}

\Author{Fatih T\"{U}RKKAN~$^{\rm a}$, O.~O\u{g}ulcan TUNCER~$^{\rm b}$
 and \.{I}smet YURDU\c{S}EN~$^{\rm b}$}
\AuthorNameForHeading{F.~T\"urkkan, O.O.~Tuncer and \.{I}.~Yurdu\c{s}en}

\Address{$^{\rm a)}$~Graduate School of Science and Engineering, Hacettepe University, \\
\hphantom{$^{\rm a)}$}~06800 Beytepe, Ankara, Turkey}
\EmailD{\mail{fatihturkkan@hacettepe.edu.tr}}

\Address{$^{\rm b)}$~Department of Mathematics, Hacettepe University, 06800 Beytepe, Ankara, Turkey}
\EmailD{\mail{otuncer@hacettepe.edu.tr}, \mail{yurdusen@hacettepe.edu.tr}}

\ArticleDates{Received May 12, 2026, in final form November 06, 2026; Published online November 20, 2026}

\Abstract{In a~recent work, we initiated a~research program aimed at the systematic investigation of quantum superintegrable systems describing the interaction of two non-relativistic spin-$\frac{1}{2}$ particles in three-dimensional Euclidean space. In that study, we classified all such superintegrable systems admitting additional first-order scalar integrals of motion. In the present paper, we continue this program by focusing on systems that admit additional pseudo-scalar integrals of motion. Starting from the most general rotationally invariant Hamiltonian for two interacting spin-$\frac{1}{2}$ particles, we construct the most general pseudo-scalar operator of differential order at most one in the form of a~matrix polynomial in the momenta. Imposing the commutativity of this operator with the Hamiltonian leads to a~system of determining equations. By solving these equations, we obtain a~complete classification of Hamiltonians in the considered class admitting additional pseudo-scalar integrals of differential order at most one, and determine these integrals explicitly. The resulting classification provides new families of spin-dependent Hamiltonians with additional pseudo-scalar symmetries. These systems enrich the class of integrable models relevant to nucleon-nucleon interactions and contribute to the broader program of classifying superintegrable quantum systems with spin. For selected cases, we further construct the associated polynomial symmetry algebras generated by the integrals of motion, providing additional insight into the algebraic structure of the systems.}


\Classification{81Q80; 81R12; 70H06}

\section{Introduction}

This paper is part of a~research program devoted to the systematic study of integrability and superintegrability in quantum systems describing the interaction of two nonrelativistic particles with spin. More specifically, we consider two particles with spin $s=\frac12$ moving in the three-dimensional Euclidean space $E^3$. Such systems provide a~natural framework for modeling interactions between particles such as nucleons and exhibit a~rich structure due to the coupling between spin and orbital degrees of freedom.

Recall that, in classical mechanics, a~Hamiltonian system with $n$ degrees of freedom is called \textit{integrable} if it admits $n$ functionally independent integrals of motion
in involution, including the Hamiltonian. An integrable system is called \textit{superintegrable} if it possesses additional functionally independent integrals of motion. More precisely, a~superintegrable system may admit up to~${n-1}$ further integrals, so that the total number of functionally independent integrals is at most~${2n-1}$. If there is only one additional integral, the system is called \textit{minimally superintegrable}, if the total number of independent integrals is $2n-1$, it is called \textit{maximally
superintegrable}. The additional integrals are not required to be in involution,
either with each other or with the original commuting set, except of course that
they commute with the Hamiltonian. In the quantum case, the definitions are
analogous, but now integrals of motion are linear operators commuting with the
Hamiltonian, and functional independence is replaced by algebraic independence. Here algebraic independence means that no nonzero completely symmetrized
(Jordan) polynomial with scalar constant coefficients in the chosen
integrals vanishes identically~\cite{MSW}.\looseness=-1

In a~recent work~\cite{TuncerYurdusen2025}, we initiated this program by classifying all superintegrable systems of two spin-$\frac12$ particles admitting additional first-order scalar integrals of motion. In that work, the most general rotationally invariant Hamiltonian for the interaction of two spin-$\frac12$ particles was constructed, and the determining equations for scalar integrals were derived and solved. The results revealed a~wide class of admissible spin-dependent interaction potentials and demonstrated the richness
of the symmetry structure of such systems.

The present paper continues this analysis. Here we classify superintegrable systems admitting additional \emph{pseudo-scalar} integrals of motion of differential order at most one. Pseudo-scalar operators arise naturally from combinations of spin, position and momentum variables that are invariant under rotations but change sign under spatial reflections. Their inclusion is therefore essential for a~more complete understanding of the symmetry structure of rotationally invariant spin-dependent Hamiltonians.

The Hamiltonian considered throughout the paper is
\begin{align}
H={}&
-\frac{\hbar^2}{2}\Delta
+V_0(r)
+\frac12 V_1(r)\bigl(\vec{\sigma}_1+\vec{\sigma}_2,\vec L\bigr)
+V_2(r)(\vec{\sigma}_1,\vec{\sigma}_2)
+V_3(r)(\vec x,\vec{\sigma}_1)(\vec x,\vec{\sigma}_2)
 \notag\\
&
+\frac12\bigl\{V_4(r),(\vec{\sigma}_1,\vec p)(\vec{\sigma}_2,\vec p)\bigr\}+\frac12 V_5(r)
\bigl[
\bigl(\vec{\sigma}_1,\vec L\bigr)\bigl(\vec{\sigma}_2,\vec L\bigr)
+
\bigl(\vec{\sigma}_2,\vec L\bigr)\bigl(\vec{\sigma}_1,\vec L\bigr)
\bigr],\label{eq:intro.H}
\end{align}
where $r=|\vec x|$. Here $\vec x=(x_1,x_2,x_3)$ is the relative position,
\[
\vec p=(p_1,p_2,p_3),\qquad
p_k=-{\rm i}\hbar\frac{\partial}{\partial x_k},
\]
is the relative momentum, and
$\vec L=\vec x\times\vec p$ is the angular momentum. We use $(\cdot,\cdot)$ for the standard
Euclidean inner product. The Pauli vectors corresponding to the two particles
are~${\vec{\sigma}_1=\vec{\sigma}\otimes I_2}$, $\vec{\sigma}_2=I_2\otimes\vec{\sigma}$,
so that the Hamiltonian acts on four-component spinors. The coefficient functions~$V_i(r)$,~${i=0,\ldots,5}$, are real functions of $r$. The braces denote an anticommutator, which is required in the $V_4$ term when $V_4$ is nonconstant. All Hamiltonians and integrals are treated as formally Hermitian differential expressions; self-adjoint realizations and boundary conditions are not considered here. The form~\eqref{eq:intro.H} is obtained by imposing translational, Galilean, permutation, rotational, reflection, time-reversal and Hermiticity conditions on the two-spin interaction~\cite{OM, TuncerYurdusen2025}. First-order vector integrals of motion for the two-spin Hamiltonian~\eqref{eq:intro.H} with~${V_4=0}$ were classified in~\cite{TuncerYurdusenVector2026}.

The total angular momentum is
\[
\vec J=\vec L+\frac{\hbar}{2}(\vec{\sigma}_1+\vec{\sigma}_2).
\]
Since the Hamiltonian~\eqref{eq:intro.H} is rotationally invariant, it commutes
with $\vec{J}^{\,2}$ and $J_3$. Thus the system is integrable by construction, with
$H$, $\vec{J}^{\,2}$, $J_3$ forming a~commuting set. The existence of any further algebraically
independent integral of motion makes the system superintegrable.

For natural Hamiltonians of the form
\begin{equation}
\label{eq:intro.spinless}
H=-\frac{\hbar^2}{2}\Delta+V_0(\vec{x}),
\end{equation}
the systematic study of classical and quantum superintegrable systems with
integrals that are polynomial in the momenta goes back to the 1960s
\cite{Fris,Makarov}. First-order integrals are associated with geometrical
symmetries of the potential, while second-order integrals are closely related
to separation of variables in the Schr\"odinger equation, or in the
Hamilton--Jacobi equation in the classical case
\cite{Evans.b,Evans.a,Fris,Makarov}. In the spherically symmetric
case, the most familiar examples are the Kepler--Coulomb potential and the
isotropic harmonic oscillator. These systems are also distinguished by
Bertrand's theorem, since they are the only spherically symmetric potentials
for which all bounded classical trajectories are closed~\cite{Bertrand}.
The study of superintegrability has since been extended in many directions,
including systems in spaces of constant and nonconstant curvature
\cite{Grosche2,Kalnins.d,Kalnins.f}, higher-dimensional systems
\cite{Kalnins.i,Kalnins.h,Rodriguez}, and systems with higher-order integrals
of motion~\cite{EWY,Gravel.a,Marquette.b,MSW,
PostWinternitz:2015,Tremblay.b,Tremblay.c}. We refer to~\cite{MillerPostWinternitz:2013}
for a~comprehensive review.\looseness=1

The rational Calogero--Moser model provides an important $N$-particle example of superintegrability with integrals polynomial in the momenta~\cite{Hakobyan2025,Wojciechowski1983}.
Its superintegrable extensions with oscillator or Coulomb potentials and their hidden symmetries have also been investigated~\cite{HakobyanLechtenfeldNersessian2014,HakobyanNersessian2015}.
Integrable spin extensions, including models of spin-$\frac{1}{2}$ particles, have been constructed using the exchange-operator formalism~\cite{MinahanPolychronakos1993}.

Another important generalization of~\eqref{eq:intro.spinless} is obtained by
allowing velocity-dependent or momentum-dependent interaction terms. This
includes, for instance, integrable and superintegrable systems describing a~particle moving in a~magnetic field~\cite{Berube,Dorizzi,EM2026,
Libor6,Libor5,Libor4,Libor2}. Spin-dependent Hamiltonians provide another
natural extension of the spinless theory. The systematic classification of
superintegrable systems involving two nonrelativistic particles, one with
spin $s=\frac12$ and the other with spin $s=0$, was initiated in
\cite{Winternitz.c} for motion in the Euclidean plane~$E^2$. The same
problem was later studied in $E^3$, first for first-order integrals of
motion~\cite{wy3}, and then for second-order integrals~\cite{DWY,YTW}. In
those works the Hamiltonian
\[
H=-\frac{\hbar^2}{2}\Delta+V_0(\vec{x})
+\frac12\bigl\{V_1(\vec{x}),\bigl(\vec{\sigma},\vec L\bigr)\bigr\}
\]
was considered, where the last term represents a~spin-orbit interaction.
For spherically symmetric potentials, the first- and second-order scalar,
pseudo-scalar, vector, axial-vector, tensor and pseudo-tensor integrals of
motion were classified. Related studies of superintegrability with spin also
include systems in which a~spin particle interacts with an external field,
such as an electromagnetic field~\cite{Nikitin.e,Nikitin.f,Pronko.b},
as well as systems involving a~spin-$\frac12$ particle interacting with a~dyon~\cite{DHoker} or with self-dual monopoles~\cite{Feher}. An algebraic approach to a~$d$-dimensional spin-$\frac{1}{2}$ matrix Hamiltonian
with $\mathfrak{so}(d+1)$ symmetry and a~spin-extended
Laplace--Runge--Lenz vector was developed in~\cite{Quesne2025}.

In this paper, our goal is to determine all choices of the potentials $V_i(r)$ in
\eqref{eq:intro.H} for which the Hamiltonian admits an additional
pseudo-scalar integral of motion of differential order at most one. To this
end, we construct the most general formally Hermitian pseudo-scalar
operator $X$, at most linear in the momenta, and impose
$
[H,X]=0$.
This commutation condition gives an overdetermined system of differential
equations for the radial functions appearing both in the Hamiltonian and in the
operator $X$. Solving this system yields the admissible potentials and the
corresponding pseudo-scalar integrals of motion.

As in the scalar case~\cite{TuncerYurdusen2025}, one must also take into account
potentials that are induced by gauge transformations of a~natural Hamiltonian~\eqref{eq:intro.spinless}.
Such potentials do not represent genuinely new spin-dependent interactions,
because they are obtained from spin-independent systems by unitary
transformations. They should therefore be excluded from the classification analysis. Nevertheless, the gauge-induced Hamiltonians have their
own symmetry structures, and these structures are useful for understanding the
role of the gauge transformation.

For the two-spin Hamiltonian~\eqref{eq:intro.H}, the gauge-induced
potentials are
\[
 V_1=\frac{2\hbar}{r^2},\qquad
 V_2=\frac{\hbar^2}{r^2},\qquad
 V_3=-\frac{\hbar^2}{r^4},\qquad
 V_4=V_5=0,
\]
with $V_0=V_0(r)$ arbitrary. These potentials are obtained from a~spherically symmetric spin-independent Hamiltonian of the form
\eqref{eq:intro.spinless}, multiplied by $I_4$, by a~unitary gauge
transformation. To write this transformation explicitly, let
$\vec x=r(\sin\theta\cos\phi,\sin\theta\sin\phi,\cos\theta)$, $r>0$,
and introduce the unitary matrix
\[
 u=
 \begin{pmatrix}
 {\rm e}^{{\rm i}\phi}\sin\theta & -\cos\theta\\
 \cos\theta & {\rm e}^{-{\rm i}\phi}\sin\theta
 \end{pmatrix}
 =\sigma_x\frac{(\vec\sigma,\vec x)}r,
 \qquad
 \sigma_x=
 \begin{pmatrix}
 0&1\\
 1&0
 \end{pmatrix}.
\]
The transformation on the two-spin space is
\[
 U=u\otimes u
 =(\sigma_x\otimes\sigma_x)
 \left[
 \frac{(\vec\sigma,\vec x)}r\otimes
 \frac{(\vec\sigma,\vec x)}r
 \right],
 \qquad U^{-1}=U^\dagger.
\]
Expanding $u\otimes u$ gives precisely the following matrix~\cite{TuncerYurdusen2025}:
\[
U=
\begin{pmatrix}
 {\rm e}^{2{\rm i}\phi}\sin^2\theta
 &-{\rm e}^{{\rm i}\phi}\cos\theta\sin\theta
 &-{\rm e}^{{\rm i}\phi}\cos\theta\sin\theta
 &\cos^2\theta\\
 {\rm e}^{{\rm i}\phi}\cos\theta\sin\theta
 &\sin^2\theta
 &-\cos^2\theta
 &-{\rm e}^{-{\rm i}\phi}\cos\theta\sin\theta\\
 {\rm e}^{{\rm i}\phi}\cos\theta\sin\theta
 &-\cos^2\theta
 &\sin^2\theta
 &-{\rm e}^{-{\rm i}\phi}\cos\theta\sin\theta\\
 \cos^2\theta
 &{\rm e}^{-{\rm i}\phi}\cos\theta\sin\theta
 &{\rm e}^{-{\rm i}\phi}\cos\theta\sin\theta
 &{\rm e}^{-2{\rm i}\phi}\sin^2\theta
\end{pmatrix}.
\]
The corresponding integrals of motion
are obtained by applying the same transformation to the integrals
of the spin-independent Hamiltonian.

Writing
$
 \vec S=\frac12(\vec\sigma_1+\vec\sigma_2)$,
so that $\hbar\vec S$ is the total spin, we choose the generators
\[
 J_i=L_i+\hbar S_i,\qquad
 \mathcal S_i=-\hbar S_i
 +\frac{2\hbar}{r^2}x_i\bigl(\vec S,\vec x\bigr).
\]
Their relation to the original spin components follows from
\[
 u^\dagger\sigma_i u
 =
 \frac{(\vec\sigma,\vec x)}r
 (\sigma_x\sigma_i\sigma_x)
 \frac{(\vec\sigma,\vec x)}r.
\]
Indeed, $
 \sigma_x\sigma_x\sigma_x=\sigma_x,
 \sigma_x\sigma_y\sigma_x=-\sigma_y,
 \sigma_x\sigma_z\sigma_x=-\sigma_z$,
while
\[
 \frac{(\vec\sigma,\vec x)}r
 \sigma_{\rm i}
 \frac{(\vec\sigma,\vec x)}r
 =
 -\sigma_i+\frac{2x_i}{r^2}(\vec\sigma,\vec x).
\]
Consequently,
\[
 U^{-1}(\hbar S_1)U=\mathcal S_1,\qquad
 U^{-1}(\hbar S_2)U=-\mathcal S_2,\qquad
 U^{-1}(\hbar S_3)U=-\mathcal S_3.
\]
The relative signs therefore come from conjugation by the
constant factor $\sigma_x$ in $u$. For the orbital components,
the same transformation gives
$
 U^{-1}L_iU=J_i-\mathcal S_i$.
Thus $\mathcal S_i$ are the gauge-transformed total-spin components
up to the constant change of spin basis
$(S_1,S_2,S_3)\mapsto(S_1,-S_2,-S_3)$.
The commutation relations follow from
\[
 [L_i,L_j]={\rm i}\hbar\varepsilon_{ijk}L_k,\qquad
 [S_i,S_j]={\rm i}\varepsilon_{ijk}S_k,\qquad
 [L_i,S_j]=0,
\]
using $\bigl[U^{-1}AU,U^{-1}BU\bigr]=U^{-1}[A,B]U$ and the above constant
change of basis. In particular,{\samepage
\[
[J_i,J_j]={\rm i}\hbar\varepsilon_{ijk}J_k,\qquad
[\mathcal S_i,\mathcal S_j]={\rm i}\hbar\varepsilon_{ijk}\mathcal S_k,\qquad
[J_i,\mathcal S_j]={\rm i}\hbar\varepsilon_{ijk}\mathcal S_k.
\]
Equivalently, the resulting six-dimensional Lie algebra is isomorphic to
$
\mathfrak{o}(3)\oplus\mathfrak{o}(3)$.}

In the special gauge-induced case corresponding to the free spin-independent Hamiltonian, the additional integrals are
\[
 \mathcal P_i=U^{-1}p_iU
 =p_i-\frac{2\hbar}{r^2}\varepsilon_{ikl}x_kS_l.
\]
The commutation relations follow by applying the same conjugation
to those of $L_i$, $p_i$ and $\hbar S_i$, together with the
constant change of spin basis stated above
\[
[J_i-\mathcal S_i,J_j-\mathcal S_j]
=
{\rm i}\hbar\varepsilon_{ijk}(J_k-\mathcal S_k),
\qquad
[J_i-\mathcal S_i,\mathcal P_j]
=
{\rm i}\hbar\varepsilon_{ijk}\mathcal P_k,
\qquad
[\mathcal P_i,\mathcal P_j]=0,
\]
while
\[
[J_i-\mathcal S_i,\mathcal S_j]=0,
\qquad
[\mathcal P_i,\mathcal S_j]=0.
\]
Thus the corresponding nine-dimensional Lie algebra is
$
\mathfrak e(3)\oplus\mathfrak o(3)$.

In addition to the classification itself, we also investigate the symmetry
algebras generated by the pseudo-scalar integrals for selected systems. These
algebras provide a~finer description of the symmetry structure. In some
cases the pseudo-scalar integrals generate finitely generated polynomial
algebras together with the spin scalar
$
K=(\vec{\sigma}_1,\vec{\sigma}_2)$,
which is a~trivial scalar integral of motion for all potentials considered here~\cite{TuncerYurdusen2025},
and the rotational invariants $\vec{J}^{\, 2}$, $J_3$. This operator $K$ allows us to construct the operator $(1+K)/2$ that interchanges the two spin factors. Note that the Hamiltonian~\eqref{eq:intro.H} is invariant under this interchange. In other cases the algebra generated by
the displayed commuting operators is abelian with a~quadratic constraint,
or the displayed algebra reduces to a~Lie-type algebra after fixing the
central elements.

The paper is organized as follows. In Section~\ref{section2}, we construct the most general
first-order pseudo-scalar operator and derive the determining equations coming
from the commutation condition $[H,Y_p]=0$. Solving these equations gives the
complete classification of superintegrable systems admitting additional pseudo-scalar integrals of motion of differential order at most one. In Section~\ref{sec3}, we study the symmetry algebras generated by
the pseudo-scalar integrals for several representative cases. Section~\ref{sec4} contains
concluding remarks and possible directions for future work.

\section{The pseudo-scalar integrals of motion}\label{section2}
We begin by constructing pseudo-scalar operators in the direct product space generated by the vector $(\vec{x}, \vec{p}, \vec{\sigma}_1, \vec{\sigma}_2)$.
From these basic quantities, one can form ten linearly independent vectorial directions, namely
	\[
	\bigl(\vec{x}, \vec{p}, \vec{L}, \vec{\sigma}_1, \vec{\sigma}_2, \vec{\sigma}_1\times \vec{x}, \vec{\sigma}_2\times \vec{x}, \vec{\sigma}_1\times \vec{p}, \vec{\sigma}_2\times \vec{p}, \vec{\sigma}_1\times \vec{\sigma}_2\bigr).
	\]
We note that first-order pseudo-scalar integrals of motion may involve arbitrary functions of~$\vec{x}$, while remaining at most linear in~$\vec{p}$ and in each of the spin operators $\vec{\sigma}_1$, $\vec{\sigma}_2$. Under these constraints, the complete set of pseudo-scalar structures can be written as
		\begin{align*}
	P_1={}&(\vec{x},\vec{\sigma}_1), \qquad P_2=(\vec{x},\vec{\sigma}_2), \qquad P_3=(\vec{\sigma}_1,\vec{p}), \qquad P_4=(\vec{\sigma}_2,\vec{p}), \qquad P_5=(\vec{x}, \vec{p})(\vec{x},\vec{\sigma}_1), \\
 P_6={}&(\vec{x},\vec{p})(\vec{x},\vec{\sigma}_2), \qquad P_7=(\vec{x},\vec{\sigma}_1\times\vec{\sigma}_2), \quad
	P_8=(\vec{\sigma}_1\times\vec{\sigma}_2,\vec{p}), \qquad P_9=(\vec{x},\vec{p})(\vec{x},\vec{\sigma}_1\times\vec{\sigma}_2), \\
 P_{10}={}&(\vec{x},\vec{\sigma}_1)(\vec{x},\vec{\sigma}_2\times\vec{p}), \qquad P_{11}=(\vec{x},\vec{\sigma}_2)(\vec{x},\vec{\sigma}_1\times\vec{p}).
	\end{align*}
These structures satisfy the operator identity
\begin{equation}
\label{eq:pseudo_structure_identity}
 r^2P_8=P_9+P_{10}-P_{11}+{\rm i}\hbar P_7.
\end{equation}
We retain all eleven structures and take this relation into account throughout the calculation.
The most general pseudo-scalar operator is then obtained as a~linear combination of these basic structures with scalar coefficient functions depending on $r = \|\vec{x}\|$
	\[
	Y_{p}=\sum_{j=1}^{11}f_{j}(r)P_j,
	\]	
where each $f_j(r)$ is assumed to be a~real-valued function. This expression is not yet Hermitian.

In order to analyze the commutation relations, it is necessary to express $Y_p$ in its fully symmetrized form. This is achieved by replacing each individual term by
\[
 f_j(r)P_j\longmapsto
 \frac12\bigl[f_j(r)P_j+P_j^\dagger f_j(r)\bigr]
\]
and subsequently summing the results. The momentum operators are
moved to the right using
\[
 [p_i,x_j]=-{\rm i}\hbar\delta_{ij},\qquad
 [p_i,f_j(r)]=-{\rm i}\hbar\frac{x_i}{r}f_j'(r).
\]
For example,
\[
 \frac12\bigl[
 f_3(r)(\vec\sigma_1,\vec p)
 +(\vec\sigma_1,\vec p)f_3(r)
 \bigr]
 =
 f_3(r)(\vec\sigma_1,\vec p)
 -\frac{{\rm i}\hbar f_3'(r)}{2r}(\vec\sigma_1,\vec x).
\]
 The symmetrization procedure is implemented using a~Mathematica routine introduced in~\cite{YT}, which generates the completely symmetrized expressions in index notation by systematically accounting for all operator permutations.
\begin{align*}
 Y_{p} ={}& \left( f_{1}(r) - {\rm i}\hbar \bigg[ \frac{f'_{3}(r)}{2r} + \left( 2f_{5}(r) + \frac{r f'_{5}(r)}{2} \right) \bigg] \right) (\vec{\sigma}_1, \vec{x}) + f_{3}(r) (\vec{\sigma}_1, \vec{p}) \\
 & + f_{5}(r) ( \vec{\sigma}_1, \vec{x}) (\vec{x}, \vec{p})+ \left( f_{2}(r) - {\rm i}\hbar \bigg[ \frac{f'_{4}(r)}{2r} + \left( 2f_{6}(r) + \frac{r f'_{6}(r)}{2} \right) \bigg] \right) (\vec{\sigma}_2, \vec{x}) \\
 & + f_{4}(r) (\vec{\sigma}_2, \vec{p}) + f_{6}(r) (\vec{\sigma}_2, \vec{x}) (\vec{x}, \vec{p})\\
 &+ \left( f_{7}(r) - {\rm i}\hbar \bigg[ \frac{f'_{8}(r)}{2r} + \left( 2f_{9}(r) + \frac{r f'_{9}(r)}{2} \right) - \frac{f_{10}(r)}{2} + \frac{f_{11}(r)}{2} \bigg] \right) (\vec{x}, \vec{\sigma}_1 \times \vec{\sigma}_2) \\
 &+ f_{8}(r) (\vec{\sigma}_1 \times \vec{\sigma}_2, \vec{p}) + f_{9}(r) (\vec{x}, \vec{\sigma}_1 \times \vec{\sigma}_2) (\vec{x}, \vec{p}) \\
 &- f_{10}(r) (\vec{\sigma}_1, \vec{x})\bigl(\vec{\sigma}_2, \vec{L}\bigr) - f_{11}(r) (\vec{\sigma}_2, \vec{x})\bigl(\vec{\sigma}_1, \vec{L}\bigr).
\end{align*}
In particular, for the choice
\[
 f_1=\cdots=f_7=0,\qquad f_8=-r^2f_9,\qquad
 f_{10}=f_9,\qquad f_{11}=-f_9
\]
equation~\eqref{eq:pseudo_structure_identity} gives $Y_p=0$ for arbitrary $f_9(r)$. Such coefficient solutions do not represent additional integrals of motion.

We now proceed to determine the conditions under which the most general symmetrized pseudo-scalar operator commutes with the Hamiltonian. More precisely, we seek those choices of the coefficient functions $f_j(r)$, $j=1,\dots,11$, and the potentials $V_i(r)$, $i=0,\dots,5$, for which~$Y_p$ becomes an integral of motion.

To achieve this, we impose the commutation relation $[H, Y_p] = 0$, which leads to a~system of determining equations. Solving this system enables us to identify both the admissible forms of the potentials $V_i$ and the corresponding functions $f_j$.

To provide an overview of the classification before entering the detailed analysis, we first state the main result of this paper in the following theorem, where the admissible potentials and the corresponding integrals of motion are given explicitly. The determining equations and their solution are presented in
Section~\ref{subsec:derivation}.

\begin{Theorem}\label{thm:main}
Assume that all constants appearing below are real, and let $\epsilon^2=1$. The functions are real and sufficiently smooth on a~connected radial interval with $r>0$, where the displayed denominators are nonzero and $1+\beta r^2>0$ whenever the square root occurs. With the gauge-induced potentials specified in the introduction excluded, the only spherically symmetric Hamiltonians~\eqref{eq:intro.H} admitting non-trivial formally Hermitian pseudo-scalar integrals of differential order at most one are the following.

The entries are not mutually exclusive. The corresponding integrals of motion, their relations and special subfamilies are given below.
\begin{enumerate}\itemsep=0pt
\item[$(1)$] $V_{4}(r)=0$, $V_{5}(r)=0$,
\begin{align*}
V_0=V_0(r), \qquad V_{3}(r)=0, \qquad V_{1}(r)=\frac{\hbar}{r^2}, \qquad V_{2}(r)=\frac{\hbar^2}{4r^2}.
\end{align*}
Seven non-trivial pseudo-scalar integrals of motion are given by
\begin{align}
	\label{eq:2.41}
	&Y_{P}^{1} = -\frac{1}{r} (\vec{\sigma}_1, \vec{x})\bigl(\vec{\sigma}_2, \vec{L}\bigr) + \frac{{\rm i}\hbar}{2r} (\vec{x}, \vec{\sigma}_1 \times \vec{\sigma}_2),\\
	&Y_{P}^{2} = \frac{1}{r} (\vec{\sigma}_2, \vec{x})(\vec{x}, \vec{p}) - r (\vec{\sigma}_2, \vec{p}) - \frac{{\rm i}\hbar}{r} (\vec{\sigma}_2, \vec{x}),\\
	&Y_{P}^{3} = r (\vec{\sigma}_1, \vec{p}) - \frac{1}{r} (\vec{\sigma}_1, \vec{x})(\vec{x}, \vec{p}) + \frac{{\rm i}\hbar}{r} (\vec{\sigma}_1, \vec{x}),\\
	&Y_{P}^{4} = -\frac{1}{r} \bigl[ (\vec{\sigma}_2, \vec{x})\bigl(\vec{\sigma}_1, \vec{L}\bigr) - (\vec{\sigma}_1, \vec{x})\bigl(\vec{\sigma}_2, \vec{L}\bigr) + {\rm i}\hbar (\vec{x}, \vec{\sigma}_1 \times \vec{\sigma}_2) \bigr],\\
	\label{eq:2.45}
	&Y_{P}^{5} = \frac{(\vec{x}, \vec{\sigma}_1 \times \vec{\sigma}_2)}{r},\\
	\label{eq:2.46}
	&Y_{P}^{6} = \frac{(\vec{x}, \vec{\sigma}_1)}{r},\\
	\label{eq:2.47}
	&Y_{P}^{7} = \frac{(\vec{x}, \vec{\sigma}_2 - \vec{\sigma}_1)}{r}.
\end{align}

\item[$(2)$] $V_{4}(r)=0$, $V_{5}(r)=0$,
\begin{align*}
V_0=V_0(r), \qquad V_2=V_2(r), \qquad V_{3}(r)=0, \qquad V_{1}(r)=\frac{\hbar}{r^2}.
\end{align*}
Three non-trivial pseudo-scalar integrals of motion are given by
\begin{align}
	\label{eq:2.49}
	&Y_{P}^{8} = \frac{1}{r} (\vec{x}, \vec{\sigma}_1 + \vec{\sigma}_2), \\
	&Y_{P}^{9} = -\frac{1}{r} \bigl[ (\vec{\sigma}_1, \vec{x})\bigl(\vec{\sigma}_2, \vec{L}\bigr) + (\vec{\sigma}_2, \vec{x})\bigl(\vec{\sigma}_1, \vec{L}\bigr) \bigr], \\
	\label{eq:2.51}
	&Y_{P}^{10} = \frac{1}{r} (\vec{x}, \vec{\sigma}_1 + \vec{\sigma}_2)(\vec{x}, \vec{p}) - r (\vec{\sigma}_1 + \vec{\sigma}_2, \vec{p}) - \frac{{\rm i}\hbar}{r} (\vec{x}, \vec{\sigma}_1 + \vec{\sigma}_2).
\end{align}

\item[$(3)$] $V_{4}(r)=0$, $V_{5}(r)=0$,
\begin{align*}
V_0=V_0(r), \qquad V_3=V_3(r), \qquad V_{1}(r)=\frac{\hbar}{r^2}, \qquad V_2(r)=\frac{\hbar^2}{4r^2}.
\end{align*}
Three non-trivial pseudo-scalar integrals of motion are $Y_P^5$, $Y_P^6$ and $Y_P^7$, given by~\eqref{eq:2.45}--\eqref{eq:2.47}.

\item[$(4)$] $V_{4}(r)=0$, $V_{5}(r)=0$,
\begin{align*}
 V_0=V_0(r), \qquad V_2=V_2(r), \qquad V_3=V_3(r), \qquad V_{1}(r)=\frac{\hbar}{r^2}.
\end{align*} A~non-trivial pseudo-scalar integral of motion is $Y_P^8$, given by~\eqref{eq:2.49}.

\item[$(5)$] $V_{4}(r)=0$, $V_{5}(r)=0$,
\begin{align*}
 V_0=V_0(r), \qquad V_2=V_2(r), \qquad V_{1}(r)=\frac{\hbar}{r^2}, \qquad V_{3}(r)=\frac{\hbar^2-4r^2V_2}{2r^4}.
\end{align*}
Three non-trivial pseudo-scalar integrals of motion are $Y_P^8$, given by~\eqref{eq:2.49}, and
\begin{align}
	\label{eq:2.58}
	&Y_{P}^{11} = \frac{1}{r} \bigl[
	r^2 (\vec{\sigma}_1 - \vec{\sigma}_2, \vec{p}) - (\vec{x}, \vec{\sigma}_1 - \vec{\sigma}_2)(\vec{x}, \vec{p}) + {\rm i}\hbar (\vec{x}, \vec{\sigma}_1 - \vec{\sigma}_2) + \hbar (\vec{x}, \vec{\sigma}_1 \times \vec{\sigma}_2)
	\bigr],\\
	\label{eq:2.59}
	&Y_{P}^{12} = \frac{1}{r} \bigl[
	2\hbar (\vec{\sigma}_1, \vec{x}) - {\rm i}\hbar (\vec{x}, \vec{\sigma}_1 \times \vec{\sigma}_2)
	+ (\vec{\sigma}_1,\vec{x})\bigl(\vec{\sigma}_2, \vec{L}\bigr)
	- (\vec{\sigma}_2,\vec{x})\bigl(\vec{\sigma}_1, \vec{L}\bigr)
	\bigr].
\end{align}

\item[$(6)$] $V_{4}(r)=0$, $V_{5}(r)=0$,
\begin{align*}
 &V_1(r)=\frac{\hbar}{r^2}\left(1+\frac{\epsilon}{\sqrt{1+\beta r^2}}\right),\\
 &V_2(r)
= \frac{\hbar^2}{4 r^2}
\left[
\frac{2+5\beta r^2+\beta^2 r^4}{\bigl(1+\beta r^2\bigr)^2}
+ \frac{2\epsilon}{\sqrt{1+\beta r^2}}
\right], \\
 &V_3(r)=-\frac{\hbar^2}{4r^4}\left[\frac{2\epsilon}{\sqrt{1+\beta r^2}}+
\frac{2+3\beta r^2}{\bigl(1+\beta r^2\bigr)^2}
\right], \\
 & V_0(r) = k_5 + \frac{\hbar^2}{\epsilon r^2 \sqrt{1+\beta r^2}} + \frac{\hbar^2 \bigl(4+3 \beta r^2 \bigl(2+\beta r^2\bigr)\bigr)}{4 \bigl(r+\beta r^3\bigr)^2}.
\end{align*}
For $\beta\neq0$, three non-trivial pseudo-scalar integrals of motion are given by
\begin{align}
	\label{eq:2.86}
	Y_P^{13} ={}& \sqrt{1+\beta r^2} (\vec{\sigma}_1, \vec{p}) - \frac{\epsilon+\sqrt{1+\beta r^2}}{r^2} (\vec{x},\vec{\sigma}_1)(\vec{x}, \vec{p}) + \frac{{\rm i}\hbar\bigl(\epsilon+\sqrt{1+\beta r^2}\bigr)}{r^2} (\vec{x},\vec{\sigma}_1) \notag\\
	& + \frac{\hbar \bigl(\epsilon + \bigl(1+\beta r^2\bigr)^{3/2}\bigr)}{2 r^2 \bigl(1+\beta r^2\bigr)} (\vec{x}, \vec{\sigma}_1 \times \vec{\sigma}_2), \\
	\label{eq:2.87}
	Y_P^{14} ={}& \sqrt{1+\beta r^2} (\vec{\sigma}_2, \vec{p}) - \frac{\epsilon+\sqrt{1+\beta r^2}}{r^2} (\vec{x},\vec{\sigma}_2)(\vec{x}, \vec{p}) + \frac{{\rm i}\hbar\bigl(\epsilon+\sqrt{1+\beta r^2}\bigr)}{r^2} (\vec{x},\vec{\sigma}_2) \notag\\
	& - \frac{\hbar \bigl(\epsilon + \bigl(1+\beta r^2\bigr)^{3/2}\bigr)}{2 r^2 \bigl(1+\beta r^2\bigr)} (\vec{x}, \vec{\sigma}_1 \times \vec{\sigma}_2), \\
	\label{eq:2.88}
	Y_P^{15} ={}& \frac{\epsilon + \bigl(1+\beta r^2\bigr)^{3/2}}{\beta r^2 \bigl(1+\beta r^2\bigr)} (\vec{x}, \vec{\sigma}_1 - \vec{\sigma}_2) - \frac{{\rm i}\bigl(\sqrt{1+\beta r^2}+\epsilon\bigr)}{\beta r^2} (\vec{x}, \vec{\sigma}_1 \times \vec{\sigma}_2) \notag\\
	& + \frac{\epsilon}{\beta \hbar r^2} (\vec{x}, \vec{\sigma}_1 \times \vec{\sigma}_2)(\vec{x}, \vec{p}) + \frac{\sqrt{1+\beta r^2}}{\beta \hbar r^2} \bigl[ (\vec{\sigma}_1, \vec{x})\bigl(\vec{\sigma}_2, \vec{L}\bigr) - (\vec{\sigma}_2, \vec{x})\bigl(\vec{\sigma}_1, \vec{L}\bigr) \bigr].
\end{align}
For $\beta=0$, the potentials reduce to the $\alpha_1=0$ specialization of system~$(10)$ when $\epsilon=-1$, and to the excluded gauge-induced family when $\epsilon=1$.

\item[$(7)$] $V_{4}(r)=0$, $V_{5}(r)=0$,
\begin{align*}
 &V_1(r)=\frac{\hbar}{r^2}\left(1+\frac{\epsilon}{\sqrt{1+\beta r^2}}\right), \\
 &V_2(r)
=
-\frac{\epsilon \kappa_1 \hbar}{4\bigl(1+\beta r^2\bigr)^{3/2}}
+\frac{\hbar^2}{4r^2}
\left[
\frac{2+5\beta r^2+\beta^2 r^4}{\bigl(1+\beta r^2\bigr)^2}
+
\frac{\epsilon\bigl(2+3\beta r^2\bigr)}{\bigl(1+\beta r^2\bigr)^{3/2}}
\right],\\
 &V_3(r)
=
-\frac{\hbar^2}{4r^4}
\left[
\frac{2\epsilon}{\sqrt{1+\beta r^2}}
+
\frac{2+3\beta r^2}{\bigl(1+\beta r^2\bigr)^2}
\right], \\
 &V_0(r)=k_8 + \frac{\hbar \kappa_1}{1+\beta r^2} - \frac{\kappa_1^2}{2\beta \bigl(1+\beta r^2\bigr)} + \frac{\hbar \bigl( 4\hbar + (\beta\hbar + 3\kappa_1)r^2 \bigr)}{4\epsilon r^2 \bigl(1+\beta r^2\bigr)^{3/2}} + \frac{\hbar^2 \bigl(2+\beta r^2\bigr)^2}{4 r^2 \bigl(1+\beta r^2\bigr)^2}.
\end{align*}
Here $\beta\neq0$. A~non-trivial pseudo-scalar integral of motion is given by
\begin{align}
	\label{eq:2.103}
	Y_P^{16} ={}& \sqrt{1+\beta r^2} (\vec{\sigma}_1 - \vec{\sigma}_2, \vec{p}) - \frac{\epsilon + \sqrt{1+\beta r^2}}{r^2} (\vec{x}, \vec{\sigma}_1 - \vec{\sigma}_2)(\vec{x}, \vec{p}) \notag \\
	& + \frac{{\rm i}\hbar\bigl(\epsilon + \sqrt{1+\beta r^2}\bigr)}{r^2} (\vec{x}, \vec{\sigma}_1 - \vec{\sigma}_2) \notag \\
	& + \frac{1}{\sqrt{1+\beta r^2}} \left[ \kappa_1 + \frac{\epsilon \hbar}{r^2} \left( \epsilon + \frac{1}{\sqrt{1+\beta r^2}} \right) \right] (\vec{x}, \vec{\sigma}_1 \times \vec{\sigma}_2).
\end{align}
For $\kappa_1=\beta\hbar$ and $k_5=k_8$, the potentials coincide with those in system~$(6)$ and $Y_P^{16}=Y_P^{13}-Y_P^{14}$. The $\beta=0$ solution is obtained by replacing only the expression for $V_0$ above~by
\[
 V_0=k_8+\frac{\hbar^2(1+\epsilon)}{r^2}+\frac{\kappa_1^2r^2}{2}
\]
and setting $\beta=0$ in $V_1$, $V_2$, $V_3$ and $Y_P^{16}$. For $\kappa_1=0$ and $\epsilon=-1$, this gives the $\alpha_1=0$ specialization of system~$(10)$; the $\kappa_1=0$, $\epsilon=1$ member is excluded as gauge-induced.

\item[$(8)$] $V_{4}(r)=0$, $V_{5}(r)=0$,
\begin{align*}
 &V_0=V_0(r), \qquad V_1(r)=\frac{\hbar}{r^2}\left(1+\frac{\epsilon}{\sqrt{1+\beta r^2}}\right), \\
 &V_2(r)=k_9
+\frac{3\hbar^2}{2\epsilon r^2\sqrt{1+\beta r^2}}
+\frac{\hbar^2}{4r^2}
\frac{6+11\beta r^2+4\beta^2 r^4}{\bigl(1+\beta r^2\bigr)^2}
-
V_0(r),\\
 &V_3(r)=-\frac{\hbar^2}{4r^4}\left[
\frac{2\epsilon}{\sqrt{1+\beta r^2}}
+
\frac{2+3\beta r^2}{\bigl(1+\beta r^2\bigr)^2}
\right].
\end{align*} A~non-trivial pseudo-scalar integral of motion is given by
\begin{align}
	\label{eq:2.106}
	Y_{P}^{17} ={}& \sqrt{1 + \beta r^2}  (\vec{\sigma}_1 + \vec{\sigma}_2, \vec{p}) - \frac{\epsilon + \sqrt{1 + \beta r^2}}{r^2} (\vec{x}, \vec{\sigma}_1 + \vec{\sigma}_2)(\vec{x}, \vec{p}) \notag \\
	& + \frac{{\rm i} \hbar \bigl(\epsilon + \sqrt{1 + \beta r^2}\bigr)}{r^2} (\vec{x}, \vec{\sigma}_1 + \vec{\sigma}_2).
\end{align} For
\[
 V_0=k_9+\frac{\hbar^2}{4r^2}
 \left[
 3+\frac{4\epsilon}{\sqrt{1+\beta r^2}}
 +\frac{1}{\bigl(1+\beta r^2\bigr)^2}
 \right],
\]
the potentials coincide with those in system~$(6)$, with $k_5=k_9$,
and $Y_P^{17}=Y_P^{13}+Y_P^{14}$.

\item[$(9)$] $V_{4}(r)=0$, $V_{5}(r)=0$,
\begin{gather*}
 V_1(r)=\frac{\hbar}{r^2}\left(1+\frac{\epsilon}{\sqrt{1+\beta r^2}}\right), \qquad V_2(r)
=
\frac{\hbar^2}{4r^2}
\left(
1+\frac{\epsilon}{\sqrt{1+\beta r^2}}
\right)
+\frac{\hbar}{4k_3\epsilon}\bigl(f_1+r f_1'\bigr), \\
 V_3(r)
=
\frac{\hbar^2}{4r^4}
\left[
\frac{2\epsilon}{\sqrt{1+\beta r^2}}
+
\frac{2+3\beta r^2}{\bigl(1+\beta r^2\bigr)^2}
\right]
-
\frac{\hbar f_1'}{2k_3\epsilon r}
-
\frac{\beta\hbar f_1}{2k_3\epsilon \bigl(1+\beta r^2\bigr)},\\
 V_0(r)
= k_{10}
+ \frac{\hbar^2}{4r^2}
\left[
\frac{7+13\beta r^2+5\beta^2 r^4}{\bigl(1+\beta r^2\bigr)^2}
+ \frac{7\epsilon}{\sqrt{1+\beta r^2}}
\right]
\\
\phantom{V_0(r)
= }{} + \frac{\hbar f_1}{4k_3}
\left[
\frac{\epsilon\bigl(1-\beta r^2\bigr)}{1+\beta r^2}
+ 4\sqrt{1+\beta r^2}
\right]+ \frac{r^2 f_1^2}{2k_3^2}
- \frac{\epsilon \hbar r f_1'}{4k_3}.
\end{gather*}
Here $k_3\neq0$, and $f_1/k_3$ is arbitrary. The two pseudo-scalar integrals of motion are given~by
\begin{align}
	\label{eq:2.115}
	Y_{P}^{18} ={}&
	\sqrt{1+\beta r^2} (\vec{\sigma}_1-\vec{\sigma}_2,\vec{p}) - \frac{\epsilon+\sqrt{1+\beta r^2}}{r^2} (\vec{x},\vec{\sigma}_1-\vec{\sigma}_2)(\vec{x},\vec{p}) \notag \\
	&+ \frac{{\rm i}\hbar\bigl(\epsilon+\sqrt{1+\beta r^2}\bigr)}{r^2} (\vec{x},\vec{\sigma}_1-\vec{\sigma}_2)
	- \frac{f_1}{k_3} (\vec{x},\vec{\sigma}_1 \times \vec{\sigma}_2),
 \\
 Y_{P}^{19} ={}& f_1 (\vec{\sigma}_1-\vec{\sigma}_2, \vec{x}) \notag\\
 & - \frac{k_3}{r^2} \bigl[ \epsilon (\vec{x}, \vec{\sigma}_1 \times \vec{\sigma}_2)(\vec{x}, \vec{p}) + \sqrt{1+\beta r^2} \bigl( (\vec{\sigma}_1, \vec{x})\bigl(\vec{\sigma}_2, \vec{L}\bigr) - (\vec{\sigma}_2, \vec{x})\bigl(\vec{\sigma}_1, \vec{L}\bigr) \bigr) \bigr] \notag \\
 & + \frac{{\rm i}\hbar k_3 \bigl(\epsilon + \sqrt{1+\beta r^2}\bigr)}{r^2} (\vec{x}, \vec{\sigma}_1 \times \vec{\sigma}_2). \label{eq:2.116}
\end{align} In particular,
\[
 \big[(\vec\sigma_1,\vec\sigma_2),Y_P^{18}\big]
 =\frac{4{\rm i}}{k_3}Y_P^{19}.
\]
The choice
\[
 \frac{f_1}{k_3}=-\frac{\hbar}{r^2}
 \left[\sqrt{1+\beta r^2}+\frac{\epsilon}{1+\beta r^2}\right]
\]
gives system~$(6)$, with $k_{10}=k_5+\beta\hbar^2/2$. For $\beta\neq0$, the choice
\[
 \frac{f_1}{k_3}
 =-\frac{\hbar}{r^2\sqrt{1+\beta r^2}}
 \left(1+\frac{\epsilon}{\sqrt{1+\beta r^2}}\right)
 -\frac{\kappa_1}{\sqrt{1+\beta r^2}},
 \qquad
 k_{10}=k_8+\hbar\kappa_1-\frac{\kappa_1^2}{2\beta}
\]
gives the potentials of system~$(7)$, with $Y_P^{18}=Y_P^{16}$.

\item[$(10)$] $V_{5}(r)=0$,
\begin{align*}
 V_{1}(r)=0, \qquad V_{2}(r)=0, \qquad V_{3}(r)=0, \qquad V_0(r)=\alpha_2, \qquad V_{4}(r)=\alpha_1.
\end{align*}
Three non-trivial pseudo-scalar integrals of motion are
\begin{align}
	\label{eq:2.140}
	Y_P^{20} ={}& \frac{1}{r^2} \bigl[ (\vec{x},\vec{p})(\vec{x},\vec{\sigma}_1 \times \vec{\sigma}_2) - (\vec{x},\vec{\sigma}_1) \bigl(\vec{\sigma}_2, \vec{L}\bigr) + (\vec{x},\vec{\sigma}_2) \bigl(\vec{\sigma}_1, \vec{L}\bigr) + {\rm i}\hbar(\vec{x},\vec{\sigma}_1\times\vec{\sigma}_2) \bigr]\notag \\
 ={}&(\vec{p},\vec{\sigma}_1 \times \vec{\sigma}_2),\\
 \label{eq:2.158}
	Y_P^{21} ={}& -2 \alpha_1 \bigl(\vec{\sigma}_1, \vec{p} \bigr) + \bigl(\vec{\sigma}_2, \vec{p} \bigr), \\
 \label{eq:2.159}
	Y_P^{22} ={}& \bigl(\vec{\sigma}_1, \vec{p} \bigr).
\end{align}
 At $\alpha_1=1/2$, the arbitrary-function integral family in system~$(12)$ with $V_1=0$ also applies.

\item[$(11)$] 
$V_4(r)=\alpha_1$, $V_5(r)=0$,
\begin{align*}
 V_1={}&V_1(r),\quad
 V_0=d_6+\frac{r^2V_1^2}{8\alpha_1^2}
 +\frac{(1+2\alpha_1)\hbar}{8\alpha_1}
 \bigl(3V_1+rV_1'\bigr),\\
 V_2={}&\frac{\hbar}{8\alpha_1}
 \bigl[(6\alpha_1-1)V_1+(2\alpha_1-1)rV_1'\bigr],\quad
 V_3=\frac{V_1^2}{4\alpha_1}
 +\frac{\hbar V_1'}{4\alpha_1r}.
\end{align*}
Here $V_1(r)$ is arbitrary, and $\alpha_1$, $d_6$ are real
constants with $\alpha_1\neq0,1/2$.
The two pseudo-scalar integrals of motion are given by
\begin{align}
\label{eq:2.170}
 Y_P^{23}
 ={}&
 \frac{1}{1-2\alpha_1}
 \left[
 (\vec\sigma_1-\vec\sigma_2,\vec p)
 -\frac{V_1}{2\alpha_1}
 (\vec x,\vec\sigma_1\times\vec\sigma_2)
 \right],\\
\label{eq:2.171}
 Y_P^{24}
 ={}&
 (\vec\sigma_1\times\vec\sigma_2,\vec p)
 +\frac{V_1}{2\alpha_1}
 (\vec\sigma_1-\vec\sigma_2,\vec x).
\end{align}
They satisfy
\[
 \bigl[(\vec\sigma_1,\vec\sigma_2),Y_P^{23}\bigr]
 =
 \frac{4{\rm i}}{(1-2\alpha_1)}Y_P^{24}.
\]

\item[$(12)$] $V_4(r)=\frac12$, $V_5(r)=0$,
\[
 V_1=V_1(r),\qquad
 V_2=\frac{\hbar V_1}{2},\qquad
 V_3=\frac{V_1^2}{2}+\frac{\hbar V_1'}{2r},\qquad
 V_0=d_8+\frac{r^2V_1^2}{2}
 +\frac{\hbar}{2}(rV_1'+3V_1).
\]
Here $V_1(r)$ is arbitrary and $d_8$ is a~real constant. A~non-trivial pseudo-scalar integral of motion is given by
\begin{align}
Y_P^{25}={}&
 -f_4(\vec\sigma_1-\vec\sigma_2,\vec p)
 +\frac{{\rm i}\hbar f_4'}{2r}
 (\vec\sigma_1-\vec\sigma_2,\vec x)
 +\left(V_1f_4-\frac{\hbar f_4'}{2r}\right)
 (\vec x,\vec\sigma_1\times\vec\sigma_2)
 \notag\\
&+g(\vec\sigma_1\times\vec\sigma_2,\vec p)
 -\frac{{\rm i}\hbar g'}{2r}
 (\vec x,\vec\sigma_1\times\vec\sigma_2)
 +\left(V_1g-\frac{\hbar g'}{2r}\right)
 (\vec\sigma_1-\vec\sigma_2,\vec x).\label{eq:2.191}
\end{align}
For each fixed $V_1(r)$ and $d_8$, the functions $f_4(r)$ and
$g(r)$ are arbitrary real smooth functions, not both identically zero.

\item[$(13)$] $V_{5}(r)=0$,
\begin{align*}
 	V_0=V_0(r), \qquad V_{1}(r)=0, \qquad V_{3}(r)=0, \qquad V_{2}(r)=d_9-V_0(r), \qquad V_{4}(r)=\alpha_1.
\end{align*}
Here $\alpha_1\neq0$ is an arbitrary constant. A~non-trivial pseudo-scalar integral of motion, valid for every such $\alpha_1$, is
\begin{equation}
	\label{eq:2.199}
	Y_P^{26} = \frac{1}{2} \bigl( \vec{\sigma}_1 + \vec{\sigma}_2, \vec{p} \bigr).
\end{equation}

\item[$(14)$] 
$V_4(r)=-\frac12$, $V_5(r)=0$,
\begin{align*}
 V_1=V_1(r),\qquad
 V_0=d_{11}+\frac{r^2V_1^2}{2},\qquad
 V_2=\frac{\hbar}{2}(2V_1+rV_1'),\qquad
 V_3=-\frac{V_1^2}{2}-\frac{\hbar V_1'}{2r}.
\end{align*}
Here $V_1(r)$ is arbitrary and $d_{11}$ is a~real constant.
The following pseudo-scalar integrals of motion are admitted
\begin{align}
\label{eq:2.210}
Y_P^{27}={}&(\vec\sigma_1\times\vec\sigma_2,\vec p)
 -V_1(\vec\sigma_1-\vec\sigma_2,\vec x),\\
\label{eq:2.211}
Y_P^{28}={}&d_1(\vec\sigma_1,\vec p)
 +f_4(\vec\sigma_1+\vec\sigma_2,\vec p)
 -\frac{{\rm i}\hbar f_4'}{2r}(\vec\sigma_1+\vec\sigma_2,\vec x)
 +\frac{d_1V_1}{2}(\vec x,\vec\sigma_1\times\vec\sigma_2),
\end{align}
where $d_1$ is a~real constant and $f_4(r)$ satisfies
$
 \hbar f_4'=rV_1(d_1+2f_4)$.
They satisfy
$\bigl[(\vec\sigma_1,\vec\sigma_2),\allowbreak Y_P^{28}\bigr]=2{\rm i}d_1Y_P^{27}$.

\item[$(15)$] $V_5(r)=0$, $V_4(r)\neq0$,
\begin{align*}
 V_4={}&\frac{d_1(f_3+f_4)^2-1}{2},\qquad
 V_1=-\frac{\hbar\bigl[d_1(f_3+f_4)^2-1\bigr](f_3'+f_4')}
 {2r(f_3+f_4)},\\
 V_3={}&\frac{\hbar^2}{4r^2}\Bigg[
 \bigl[2d_1(f_3+f_4)^2-1\bigr]
 \left(\frac{f_3''+f_4''}{f_3+f_4}
 -\frac{f_3'+f_4'}{r(f_3+f_4)}\right)\\
 & +\frac{1+5d_1(f_3+f_4)^2}{2}
 \frac{(f_3'+f_4')^2}{(f_3+f_4)^2}\Bigg],\\
 V_2={}&\frac{\hbar^2}{4}\Bigg[
 \frac{f_3'+f_4'}{r(f_3+f_4)}
 -\bigl[d_1(f_3+f_4)^2-1\bigr]
 \frac{f_3''+f_4''}{f_3+f_4} +\frac{2(f_3'+f_4')^2}
 {\bigl[d_1(f_3+f_4)^2-2\bigr](f_3+f_4)^2}\Bigg],\\
 V_0={}&d_4-\frac{\hbar^2}{8}
 \frac{2d_1^2(f_3+f_4)^4-3d_1(f_3+f_4)^2+2}
 {d_1(f_3+f_4)^2-2}
 \left(\frac{f_3'+f_4'}{f_3+f_4}\right)^2.
\end{align*}
Here $f_3+f_4\neq0$,
$f_3-f_4\neq0$, $d_1d_2\neq0$, and $d_1(f_3+f_4)^2-d_2(f_3-f_4)^2=2$.
Two pseudo-scalar integrals of motion are given by
\begin{align}
Y_P^{29}={}&
 (f_3-f_4)(\vec\sigma_1\times\vec\sigma_2,\vec p)
 -\frac{{\rm i}\hbar(f_3'-f_4')}{2r}
 (\vec x,\vec\sigma_1\times\vec\sigma_2)
 \notag\\
&-\frac{\hbar(f_3f_3'-f_4f_4')}{r(f_3+f_4)}
 (\vec\sigma_1-\vec\sigma_2,\vec x),\label{eq:c3.step.IA.Y35}\\
 Y_P^{30}={}&
 f_3(\vec\sigma_1,\vec p)+f_4(\vec\sigma_2,\vec p)
 -\frac{{\rm i}\hbar f_3'}{2r}(\vec\sigma_1,\vec x)
 -\frac{{\rm i}\hbar f_4'}{2r}(\vec\sigma_2,\vec x)
 \notag\\
&+\frac{\hbar(f_3f_3'-f_4f_4')}{2r(f_3+f_4)}
 (\vec x,\vec\sigma_1\times\vec\sigma_2).\label{eq:c3.step.IA.Y36}
\end{align}

\item[$(16)$] $V_5(r)=0$, $V_4(r)\neq0$,
\begin{align*}
 V_4(r)={}&\frac{4d_1f_3^2-1}{2},\qquad V_2=V_2(r),\qquad
 V_1(r)=-\frac{\hbar\bigl(4d_1f_3^2-1\bigr)f_3'}{2rf_3},\\
 V_3(r)={}&\frac{\hbar^2}{4r^2}\left[
 \bigl(8d_1f_3^2-1\bigr)
 \left(\frac{f_3''}{f_3}-\frac{f_3'}{rf_3}\right)
 +\frac{1+20d_1f_3^2}{2}\frac{(f_3')^2}{f_3^2}
 \right],\\
 V_0(r)={}&d_4-V_2-\frac{\hbar^2}{4}\left[
 \bigl(4d_1f_3^2-1\bigr)\frac{f_3''}{f_3}
 +\frac{1+8d_1f_3^2}{2}\frac{(f_3')^2}{f_3^2}
 -\frac{f_3'}{rf_3}
 \right].
\end{align*}
Here $f_3(r)\neq0$ and $V_2(r)$ are arbitrary, subject to
$4d_1f_3^2\neq1$. A~non-trivial pseudo-scalar integral of motion is given by
\begin{equation}
\label{eq:c3.step.IBa.Y37}
 Y_P^{31}=f_3(\vec\sigma_1+\vec\sigma_2,\vec p)
 -\frac{{\rm i}\hbar f_3'}{2r}(\vec\sigma_1+\vec\sigma_2,\vec x).
\end{equation}

\item[$(17)$] $V_5(r)=0$, $V_4(r)\neq0$,
\begin{align*}
 V_4(r)={}&\frac{1+d_1f_3^2}{2},\qquad V_1=V_1(r),\\
 V_3(r)={}&
 \frac{V_1^2}{2\bigl(1+d_1f_3^2\bigr)}
 +\frac{\hbar V_1'}{2r\bigl(1+d_1f_3^2\bigr)}
 -\frac{d_1\hbar f_3f_3'V_1}{r\bigl(1+d_1f_3^2\bigr)^2}\\
 &
 +\frac{d_1\hbar^2}{2r^2}
 \left[f_3f_3''+(f_3')^2-\frac{f_3f_3'}r\right],\\
 V_2(r)={}&
 \frac{\hbar}{4\bigl(1+d_1f_3^2\bigr)}
 \bigl[d_1rf_3^2V_1'+\bigl(2+3d_1f_3^2\bigr)V_1\bigr]\\
 &
 +\frac{d_1\hbar rf_3f_3'\bigl(1-d_1f_3^2\bigr)V_1}
 {4\bigl(1+d_1f_3^2\bigr)^2}
 -\frac{d_1\hbar^2}{8}f_3
 \left(f_3''-\frac{2f_3'}r\right),\\
 V_0(r)={}&
 d_3+\frac{r^2V_1^2}{2\bigl(1+d_1f_3^2\bigr)^2}
 +\frac{\hbar\bigl(2+d_1f_3^2\bigr)}{4\bigl(1+d_1f_3^2\bigr)}
 (rV_1'+3V_1)\\
 &
 +\frac{d_1\hbar rf_3f_3'\bigl(d_1f_3^2-1\bigr)V_1}
 {4\bigl(1+d_1f_3^2\bigr)^2}
 +\frac{d_1\hbar^2}{8}f_3
 \left(f_3''+\frac{2f_3'}r\right).
\end{align*}
Here $f_3(r)$ and $V_1(r)$ are arbitrary, subject to
$1+d_1f_3^2\neq0$ and $f_3\not\equiv0$. Two pseudo-scalar integrals of motion are given by
\begin{align}
\label{eq:c3.step.IIA.Y38}
Y_P^{32}={}&f_3(\vec\sigma_1\times\vec\sigma_2,\vec p)
 -\frac{{\rm i}\hbar f_3'}{2r}(\vec x,\vec\sigma_1\times\vec\sigma_2)
 +\left(\frac{V_1f_3}{1+d_1f_3^2}-\frac{\hbar f_3'}{2r}\right)
 (\vec\sigma_1-\vec\sigma_2,\vec x), \\
\label{eq:c3.step.IIA.Y39}
Y_P^{33}={}&f_3(\vec\sigma_1-\vec\sigma_2,\vec p)
 -\frac{{\rm i}\hbar f_3'}{2r}(\vec\sigma_1-\vec\sigma_2,\vec x)
 +\left(\frac{\hbar f_3'}{2r}-\frac{V_1f_3}{1+d_1f_3^2}\right)
 (\vec x,\vec\sigma_1\times\vec\sigma_2).
\end{align}
They satisfy
$
\bigl[(\vec\sigma_1,\vec\sigma_2),Y_P^{33}\bigr]=4{\rm i}Y_P^{32}$.
\end{enumerate}
\end{Theorem}

We also examine gauge equivalence between the Hamiltonian families
listed in Theorem~\ref{thm:main}. In particular, the families in
systems~(2) and~(5) are related by a~position-dependent unitary transformation,
as shown below.

Let
\smash{$
 H=\frac{\vec p^{\, 2}}2+V_0(r)
 +\frac{\hbar}{2r^2}\bigl(\vec\sigma_1+\vec\sigma_2,\vec L\bigr)
 +V_2(r)(\vec\sigma_1,\vec\sigma_2)
$}
be a~Hamiltonian in system~(2), and consider
\smash{$
 U_a=\frac{(\vec\sigma_1,\vec x)}r$},
\smash{$U_a^\dagger=U_a^{-1}=U_a$}, $r>0$.
Since $U_a$ depends on the coordinates, the momentum operators
transform as
\[
 U_a^{-1}p_jU_a
 =p_j-{\rm i}\hbar U_a^{-1}(\partial_jU_a)
 =p_j-\frac{\hbar}{r^2}(\vec x\times\vec\sigma_1)_j.
\]
Using the Pauli identities, we also have
\smash{$
 U_a^{-1}\vec\sigma_1U_a
 =\frac{2\vec x}{r^2}(\vec\sigma_1,\vec x)-\vec\sigma_1$},
$
 U_a^{-1}\vec\sigma_2U_a=\vec\sigma_2$.
These relations yield
\begin{gather*}
 U_a^{-1}\frac{\vec p^{\, 2}}2U_a
 =\frac{\vec p^{\, 2}}2
 +\frac{\hbar}{r^2}\bigl(\vec\sigma_1,\vec L\bigr)
 +\frac{\hbar^2}{r^2}I_4,\qquad
 U_a^{-1}\bigl(\vec\sigma_1,\vec L\bigr)U_a
 =-\bigl(\vec\sigma_1,\vec L\bigr)-2\hbar I_4,\\
 U_a^{-1}\bigl(\vec\sigma_2,\vec L\bigr)U_a
 =\bigl(\vec\sigma_2,\vec L\bigr)
 +\hbar(\vec\sigma_1,\vec\sigma_2)
 -\frac{\hbar}{r^2}
 (\vec\sigma_1,\vec x)(\vec\sigma_2,\vec x),\\
 U_a^{-1}(\vec\sigma_1,\vec\sigma_2)U_a
 =-(\vec\sigma_1,\vec\sigma_2)
 +\frac{2}{r^2}
 (\vec\sigma_1,\vec x)(\vec\sigma_2,\vec x).
\end{gather*}
The scalar contributions from the kinetic and spin--orbit terms
cancel. Substituting the above identities into
$\widetilde H=U_a^{-1}HU_a$ and collecting terms, we obtain
\begin{align*}
 \widetilde H={}&\frac{\vec p^{\, 2}}2+V_0
 +\frac{\hbar}{2r^2}\bigl(\vec\sigma_1+\vec\sigma_2,\vec L\bigr)+\left(\frac{\hbar^2}{2r^2}-V_2\right)
 (\vec\sigma_1,\vec\sigma_2)
 +\left(\frac{2V_2}{r^2}-\frac{\hbar^2}{2r^4}\right)
 (\vec\sigma_1,\vec x)(\vec\sigma_2,\vec x).
\end{align*}
Thus the transformed Hamiltonian has the form~\eqref{eq:intro.H},
with
\begin{gather*}
 \widetilde V_0=V_0,\qquad
 \widetilde V_1=\frac{\hbar}{r^2},\qquad
 \widetilde V_2=\frac{\hbar^2}{2r^2}-V_2,\qquad
 \widetilde V_3=\frac{2V_2}{r^2}-\frac{\hbar^2}{2r^4}
 =\frac{\hbar^2-4r^2\widetilde V_2}{2r^4},
 \\
 \widetilde V_4=\widetilde V_5=0.
\end{gather*}
These are precisely the potentials in system~(5). The corresponding
integrals are related by
\begin{align*}
 U_a^{-1}Y_P^8U_a=Y_P^8,\qquad
 U_a^{-1}Y_P^9U_a=2\hbar Y_P^8-Y_P^{12},\qquad
 U_a^{-1}Y_P^{10}U_a=Y_P^{11}.
\end{align*}
For self-adjoint realizations with
$D\bigl(\widetilde H\bigr)=U_a^{-1}D(H)$, the spectra and eigenvalue
multiplicities coincide.

\subsection[Derivation of Theorem 2.1]{Derivation of Theorem~\ref{thm:main}}
\label{subsec:derivation}

We recall that $H$ and $Y_p$ are $4\times 4$ matrix differential operators acting on a~four-component spinor space. Hence $[H,Y_p]$ is also a~$4\times 4$ matrix differential operator, and the condition $[H,Y_p]=0$ is equivalent to the vanishing of its sixteen entries. This gives a~large overdetermined system of determining equations, many of which are dependent or redundant. We organize these equations according to the order of derivatives appearing in the commutator. The highest-order equations are displayed in the main text, since they determine the main branches of the analysis, while the reduced lower-order systems are collected in Appendix~\ref{appendix-determining-equations}.

It can be verified that the commutator $[H,Y_p]$ contains derivatives up to third order. Setting the coefficients of all third-order derivative terms equal to zero gives the first subset of determining equations
\begin{align*}
 &(f_{10}+f_{11})V_{4}=0, \qquad (f_{9}+f_{11})V_{4}=0, \qquad (f_{10}-f_{9})V_{4}=0, \qquad (f_{5}+f_{6})V_{4}=0, \\
 &(f_{5}-f_{6})V_{4}=0, \qquad
 (f_{3}+f_{4})V_{5}=0, \qquad (f_{3}-f_{4})V_{5}=0, \qquad (f_{5}+f_{6})V_{5}=0, \nonumber \\ & (f_{5}-f_{6})V_{5}=0, \qquad (f_{10}+f_{11})V_{5}=0, \qquad
 (f_{10}-f_{9})V_{5}=0,\qquad(f_{11}+f_{9})V_{5}=0, \\
 &\bigl(f_{8}+r^{2}f_{9}\bigr)V_{5}=0, \qquad \bigl(f_{8}+r^{2}f_{10}\bigr)V_{5}=0, \qquad \bigl(f_{3}+r^{2}f_{6}\bigr)V_{5}=0, \\
 & \bigl(f_{8}-r^{2}f_{11}\bigr)V_{5}=0, \qquad\bigl(r^{2}f_{5}-f_{3}\bigr)V_{5}=0, \qquad \bigl(f_{4}+r^{2}f_{5}\bigr)V_{5}=0, \qquad \bigl(f_{4}-r^{2}f_{6}\bigr)V_{5}=0.
\end{align*}
These equations suggest the following four cases:
\begin{alignat*}{3}
 &\textbf{Case 1.}\ V_{4}=0, \ V_{5}=0, \qquad&&
\textbf{Case 2.}\ V_{4}=0, \ V_{5}\neq 0,&\\
& \textbf{Case 3.}\ V_{4}\neq 0, \ V_{5}=0, \qquad&&
\textbf{Case 4.}\ V_{4}\neq 0, \ V_{5}\neq 0.&
\end{alignat*}
We analyze these cases separately on connected radial intervals with $r>0$, on which the coefficients are smooth and the displayed denominators do not vanish. The conditions $V_4\neq0$ and $V_5\neq0$ refer to intervals where the corresponding functions are nonzero. In each case, the main text gives the solution of the determining equations and the resulting admissible potentials and pseudo-scalar integrals. The reduced lower-order determining equations used in the analysis are listed in Appendix~\ref{appendix-determining-equations}.

We first consider Case~1.

{\bf Case 1. $\boldsymbol{ V_{4}=0}$, $\boldsymbol{V_{5}=0}$.}
For $V_4=V_5=0$, all determining equations obtained from the third-order
derivative terms are identically satisfied. We therefore proceed to the
determining equations coming from the second-order derivative terms. The second-order system is given by
the $24$ equations listed in Appendix~\ref{appendix-determining-equations}.
This list includes dependent equations.
These equations impose further restrictions on the potentials $V_i$ and on
the coefficient functions $f_j$. Once the second-order system is solved, the
resulting expressions are substituted into the remaining first- and zeroth-order
determining equations, which are then solved in each subcase.

From the second-order determining equation~\eqref{eq:2.17}, we obtain
$
V_1=\frac{\hbar}{r^2}$ or $ f_{10}=-f_{11}$.
We first examine the branch $V_1=\frac{\hbar}{r^2}$.

{\bf Subcase 1.} $V_1=\frac{\hbar}{r^2}$.
We shall now start to solve these determining equations from~\eqref{eq:2.2} to~\eqref{eq:2.25} by introducing this condition. First, from equations~\eqref{eq:2.3}, \eqref{eq:2.10} and~\eqref{eq:2.15} we get, respectively,
\begin{equation}
 \label{eq:2.26}
 f_{10}=\frac{c_1}{r}-f_{11}, \qquad f_6=\frac{c_2}{r}, \qquad f_{9}=-\frac{f_8}{r^2},
\end{equation}
where $c_1$ and $c_2$ are real constants. In the following steps, $c_i$ $(i=1,2,3,\ldots)$ will be considered as real constants.

Equation~\eqref{eq:2.6} reduces to
$r f_6 + f_4' = 0$,
which implies
$f_4' = -r f_6$.
Using this together with~\eqref{eq:2.26}, we obtain
$f_4 = -r c_2 + c_3$.
Equation~\eqref{eq:2.24} gives{\samepage
\[
	f_3' - \frac{f_3}{r} = 0 \ \Rightarrow\ f_3 = r c_4,
\]
and from~\eqref{eq:2.14} it follows that
$f_5 = -\frac{f_3}{r^2} = -\frac{c_4}{r}$.}

Substituting the expressions for $f_3$, $f_4$, $f_5$, $f_6$, $f_9$ and $f_{10}$ into~\eqref{eq:2.2}--\eqref{eq:2.25} and removing dependent equations yields $c_3=0$ and
\[
	r^2 f_{11} + f_8 + r^3 f_{11}' - r f_8' = 0,
\]
from which
$f_{11} = \frac{c_5}{r} + \frac{f_8}{r^2}$.

Thus, all second-order determining equations are satisfied. The remaining first- and zeroth-order determining equations are treated by substitution and elimination of redundant relations. From equation~\eqref{eq:feq_93}, it follows that
\[
	f_7 + r f_7' = 0 \;\;\Rightarrow\;\; f_7 = \frac{c_6}{r}.
\]
Equation~\eqref{eq:feq_53} reads
$
	f_1 + f_2 + r(f_1' + f_2') = 0$,
which gives
\begin{equation}
 \label{eq:2.28}
	f_1 = \frac{c_7}{r} - f_2.
\end{equation}
Substituting~\eqref{eq:2.28} into equation~\eqref{eq:feq_29}, we obtain
$f_2 + r f_2' = 0$,
from which it follows that
$f_2 = \frac{c_8}{r}$.
Equation~\eqref{eq:feq_4} reduces to
$(c_2 - c_4)V_3 = 0$.
This relation implies either~${V_3 = 0}$ or $c_2 = c_4$. We first examine the case $V_3 = 0$.

{\bf S1.} $V_3 = 0$. From equations~\eqref{eq:feq_8}, \eqref{eq:feq_65}, and~\eqref{eq:feq_95}, we obtain respectively
\[
	(c_2 + c_4)\bigl(\hbar^2 - 4r^2 V_2\bigr) = 0, \qquad
	(c_1 - 2c_5)\bigl(\hbar^2 - 4r^2 V_2\bigr) = 0, \qquad
	c_6\bigl(\hbar^2 - 4r^2 V_2\bigr) = 0.
\]
This implies that either \smash{$V_2=\frac{\hbar^2}{4r^2}$} or $c_4=-c_2$, $c_1=2c_5$, $c_6=0$. First we will examine case~$\smash{V_2=\frac{\hbar^2}{4r^2}}$.

{\bf I.} $V_2 = \frac{\hbar^2}{4r^2}$.
 In this case, all determining equations arising from the coefficients of the first- and zeroth-order terms are automatically satisfied. Hence, the full set of determining equations is satisfied, and the scalar potential $V_0$ remains arbitrary. We thus obtain seven constants of motion corresponding to the parameters $c_1$, $c_2$, $c_4$, $c_5$, $c_6$, $c_7$, and $c_8$, given by
\eqref{eq:2.41}--\eqref{eq:2.47}.

{\bf II.} $c_4 = -c_2$, $c_1 = 2c_5$, $c_6 = 0$, \smash{$\bigl( V_2 \neq \frac{\hbar^2}{4r^2} \bigr)$}.
 Equation~\eqref{eq:feq_11} reduces to
\begin{equation}	\label{eq:2.48}
	(c_7 - 2c_8)\bigl(\hbar^2 - 4r^2 V_2\bigr) = 0.
\end{equation}
Since $V_2 \neq \frac{\hbar^2}{4r^2}$, it follows from~\eqref{eq:2.48} that $c_7 = 2c_8$. In this case, all determining equations arising from the coefficients of the first- and zeroth-order terms are satisfied. Both potentials $V_0$ and $V_2$ remain arbitrary. We are thus left with three constants of motion corresponding to the parameters $c_8$, $c_5$, and $c_2$, given by
\eqref{eq:2.49}--\eqref{eq:2.51}.

{\bf S2.} $c_4 = c_2$, $(V_3 \neq 0)$.
From~\eqref{eq:feq_8}, we obtain $c_2\bigl(\hbar^2 - 4r^2 V_2 - 2r^4 V_3\bigr) = 0$. The relation implies either $c_2 = 0$, or
\begin{equation}\label{eq:2.52}
	V_3 = \frac{\hbar^2 - 4r^2 V_2}{2r^4}.
\end{equation}
We first consider the case $c_2 = 0$. Since $c_4 = c_2$, it follows that $c_4 = 0$.

{\bf I.} $c_2 = 0$, $c_4 = 0$, $\bigl( V_3 \neq \frac{\hbar^2 - 4r^2 V_2}{2r^4} \bigr)$.
 Substituting these conditions into the determining equations yields $c_1 = 0$. Equation~\eqref{eq:feq_12} reads{\samepage
\begin{equation}\label{eq:2.53}
	c_6\left( \frac{\hbar^2}{r^2} - 4V_2 \right) = 0.
\end{equation}
From~\eqref{eq:2.53}, we obtain either $c_6 = 0$ or $V_2 = \frac{\hbar^2}{4r^2}$. We first consider the case $V_2 = \frac{\hbar^2}{4r^2}$.}

{\bf A.} $V_2 = \frac{\hbar^2}{4r^2}$.
 Substituting this condition into the determining equations yields $c_5 = 0$. Then all determining equations arising from the coefficients of the first- and zeroth-order terms are satisfied. The potentials $V_0$ and $V_3$ remain arbitrary, and we are left with three constants~$c_6$,~$c_7$, and $c_8$. The corresponding integrals of motion coincide with previously obtained ones: the integral associated with $c_6$ is given by~\eqref{eq:2.45}, that associated with $c_7$ by~\eqref{eq:2.46}, and that associated with $c_8$ by~\eqref{eq:2.47}.

{\bf B.} $c_6 = 0$, $\big( V_2 \neq \frac{\hbar^2}{4r^2} \big)$.
 Substituting these conditions into the determining equations yields $c_5 = 0$. Equation~\eqref{eq:feq_11} reduces to
\begin{equation}\label{eq:2.54}
(c_7 - 2c_8)\bigl(\hbar^2 - 4r^2 V_2\bigr) = 0.
\end{equation}
Since $V_2 \neq \frac{\hbar^2}{4r^2}$, it follows from~\eqref{eq:2.54} that $c_7 = 2c_8$. Thus, all determining equations are satisfied. The potentials $V_0$, $V_2$, and $V_3$ remain arbitrary, and only one constant $c_8$ remains. The corresponding integral of motion coincides with~\eqref{eq:2.49}.

{\bf II.} $V_3 = \frac{\hbar^2 - 4r^2 V_2}{2r^4}$.
From~\eqref{eq:feq_98} and~\eqref{eq:feq_12}, we obtain
\begin{gather}\label{eq:2.55}
	c_1\bigl(\hbar^2 - 4r^2 V_2\bigr) = 0,
\\
	\label{eq:2.56}
	(c_6 - c_2 \hbar)\bigl(\hbar^2 - 4r^2 V_2\bigr) = 0.
\end{gather}
If $V_2 = \frac{\hbar^2}{4r^2}$, then $V_3 = 0$, which contradicts the present assumption $V_3 \neq 0$. Therefore, from~\eqref{eq:2.55} and~\eqref{eq:2.56} we obtain $c_1 = 0$, $c_6 = c_2 \hbar$. Equation~\eqref{eq:feq_11} reads
\begin{equation}
	\label{eq:2.57}
	 (c_7 - 2 (c_8 + c_5 \hbar ) )\bigl(\hbar^2 - 4r^2 V_2\bigr) = 0.
\end{equation}
Since $V_2 \neq \frac{\hbar^2}{4r^2}$, it follows from~\eqref{eq:2.57} that $c_7 = 2\bigl(c_8 + c_5 \hbar\bigr)$. In this case, all determining equations arising from the coefficients of the first- and zeroth-order terms are satisfied. The potentials $V_0$ and $V_2$ remain arbitrary, and we are left with three constants. The corresponding integrals of motion for $c_2$ and $c_5$ are given, respectively, by
$Y_P^{11}$ and $Y_P^{12}$ in~\eqref{eq:2.58} and~\eqref{eq:2.59}.
The integral of motion corresponding to $c_8$ coincides with~\eqref{eq:2.49}.

{\bf Subcase 2.} $f_{11} = -f_{10} $.
 From~\eqref{eq:2.7} and~\eqref{eq:2.14}, we obtain $f_6$ and $f_5$ as
\begin{equation}\label{eq:2.60}
	f_6 = -\frac{f_4 V_1}{\hbar}, \qquad
	f_5 = -\frac{f_3 V_1}{\hbar}.
\end{equation}

Equation~\eqref{eq:2.24} reduces to
$
	r f_3 V_1 \bigl(2\hbar - r^2 V_1\bigr) - \hbar^2 f_3' = 0$,
from which we obtain
\begin{equation}
	\label{eq:2.62}
	f_3' = \frac{r f_3 V_1 \bigl(2\hbar - r^2 V_1\bigr)}{\hbar^2}.
\end{equation}

Substituting~\eqref{eq:2.60} and~\eqref{eq:2.62} into~\eqref{eq:2.16}, we obtain
\begin{equation}
	\label{eq:2.63}
	f_3 \bigl( 3\hbar r V_1^2 - r^3 V_1^3 + \hbar^2 V_1' \bigr) = 0.
\end{equation}

From equation~\eqref{eq:2.63}, we obtain either $f_3 = 0$ or \smash{$V_1(r) = \frac{\hbar}{r^2} \bigl( 1 + \frac{\epsilon}{\sqrt{1 + \beta r^2}}\bigr)$}, where $\epsilon^2=1$, together with the already considered solution $V_1=\hbar/r^2$. Note that the special cases $V_1=2\hbar/r^2$ and $V_1=0$ correspond to $(\epsilon,\beta)=(1,0)$ and $(-1,0)$, respectively. The solution $V_1=\hbar/r^2$ is also obtained as a~limit of this family. The square-root expressions below are considered on intervals where $1+\beta r^2>0$. First we will consider case \smash{$V_1=\frac{\hbar}{r^2}\bigl(1+\frac{\epsilon}{\sqrt{1+\beta r^2}}\bigr)$}.

{\bf S1.} \smash{$V_1=\frac{\hbar}{r^2}\bigl(1+\frac{\epsilon}{\sqrt{1+\beta r^2}}\bigr)$}.
Equation~\eqref{eq:2.62} reduces to
\[
	f'_3-\frac{\beta r}{1+\beta r^2}f_3=0.
\]
Hence
$
	f_3=k_1\sqrt{1+\beta r^2}$,
where $k_1$ is a~real constant. In what follows, all constants $k_i$ $(i=1,2,3,\ldots)$ are assumed to be real.

Using~\eqref{eq:2.60}, equation~\eqref{eq:2.22} reduces to
\[
	f'_4-\frac{\beta r}{1+\beta r^2}f_4=0.
\]
Therefore,
\smash{$f_4=k_2\sqrt{1+\beta r^2}$}.

Equation~\eqref{eq:2.15} becomes
\[
	\frac{\epsilon\bigl(r^2 f_{10}+f_8\bigr)}{r^2\sqrt{1+\beta r^2}}
	+\frac{f_8}{r^2}+f_9=0.
\]
Thus,
\[
	f_9=-\frac{\epsilon\bigl(r^2 f_{10}+f_8\bigr)}{r^2\sqrt{1+\beta r^2}}
	-\frac{f_8}{r^2}.
\]

Using the expression for $f_9$, equation~\eqref{eq:2.25} reduces to
\[
	r \bigl(2 + \beta r^2\bigr) f_{10} - \beta r f_8
	+ \bigl(1 + \beta r^2\bigr) \bigl( r^2 f'_{10} + f'_8 \bigr) = 0.
\]
Solving this equation gives
\smash{$f_{10}=\frac{k_3\sqrt{1+\beta r^2}-f_8}{r^2}$}.

Thus, all determining equations coming from the second-order terms are satisfied. We now consider the first- and zeroth-order equations, after eliminating redundant relations. From~\eqref{eq:feq_53} and~\eqref{eq:feq_51}, we obtain
$
	f_1+f_2+r(f'_{1}+f'_{2})=0$,
and
\[
	f_1 + f_2 + r (f'_1 + f'_2) + \frac{\epsilon (f_1 + f_2)}{\sqrt{1 + \beta r^2}} = 0.
\]
These two equations imply $f_2=-f_1$.

Equation~\eqref{eq:feq_29} becomes
\begin{equation}
	\label{eq:2.74}
	\frac{2 k_3 \hbar^2 \bigl(\epsilon \sqrt{1 + \beta r^2} + 1\bigr)}{r \sqrt{1 + \beta r^2}}
	+ 2 \hbar r (f_1 + r f'_1)
	- 8 k_3 \epsilon r V_2 = 0.
\end{equation}
Solving~\eqref{eq:2.74} for $f'_1$, we get
\begin{equation}
	\label{eq:2.75}
	f'_{1}=\frac{k_3 \epsilon\bigl(4 r^2 V_2-\hbar^2\bigr)}{\hbar r^3}
	-\frac{k_3\hbar}{r^3 \sqrt{1+\beta r^2}}-\frac{f_1}{r}.
\end{equation}

From~\eqref{eq:feq_2}, we obtain
\[
	(k_1 + k_2) \bigl[ \epsilon \hbar^2 \bigl( 2 \bigl(1 + \beta r^2\bigr)^{3/2} + \epsilon \bigl(2 + 3 \beta r^2\bigr) \bigr)
	+ 4 r^4 \bigl(1 + \beta r^2\bigr)^2 V_3 \bigr] = 0.
\]
Therefore, we distinguish two cases.

{\bf I.} \smash{$V_3 =
- \frac{\hbar^2 [ 2 \epsilon (1+\beta r^2)^{3/2} + 2 + 3 \beta r^2 ]}{4 r^4 (1+\beta r^2)^2}$}.
 Equation~\eqref{eq:feq_98} reduces to
\begin{align}
	&8k_3 r^2\bigl(1+\beta r^2\bigr)^2 V_2
	-6k_3\hbar^2-10\beta k_3 \hbar^2 r^2-2\beta^2 k_3\hbar^2 r^4 \nonumber \\
	& \qquad -6k_3\epsilon\hbar^2\sqrt{1+\beta r^2}
	-6\beta k_3\epsilon\hbar^2 r^2\sqrt{1+\beta r^2}
	-2\epsilon\hbar r^2\bigl(1+\beta r^2\bigr) f_1=0.\label{eq:2.77}
\end{align}
For $k_3\neq0$, solving equation~\eqref{eq:2.77} for $V_2$, we obtain
\begin{align}
	\label{eq:2.78}
	V_2 ={}& \frac{1}{4k_3 r^2\bigl(1+\beta r^2\bigr)^2}
	\bigl[ 3k_3 \hbar^2 + 5\beta k_3\hbar^2 r^2 + \beta^2 k_3 \hbar^2 r^4 \nonumber \\
	&  + 3k_3\epsilon \hbar^2\sqrt{1+\beta r^2}
	+ 3\beta k_3 \epsilon\hbar^2 r^2\sqrt{1+\beta r^2}
	+ \epsilon\hbar r^2\bigl(1+\beta r^2\bigr) f_1 \bigr].
\end{align}
Substituting~\eqref{eq:2.78} into~\eqref{eq:2.75} gives
\begin{equation}
	\label{eq:2.79}
	f_1=-\frac{k_3\hbar}{r^2}\left(\frac{1}{\sqrt{1+\beta r^2}}
	+\frac{\epsilon}{1+\beta r^2}\right)
	+\frac{k_4}{\sqrt{1+\beta r^2}}.
\end{equation}
Finally, substituting~\eqref{eq:2.79} into~\eqref{eq:2.78} yields

{\bf A.} $k_3\neq 0$.
\[
	V_2 =
	\frac{
		\hbar \bigl[
		k_4 \epsilon r^2 \sqrt{1+\beta r^2}
		+ k_3 \hbar \bigl(
		2 + 5 \beta r^2 + \beta^2 r^4
		+ \epsilon \sqrt{1+\beta r^2} \bigl(2 + 3 \beta r^2\bigr)
		\bigr)
		\bigr]
	}{
		4 k_3 r^2 \bigl(1+\beta r^2\bigr)^2
	}.
\]
Equation~\eqref{eq:feq_8} reads
\[
	k_3(k_1-k_2) \left(\hbar + \epsilon\hbar\sqrt{1+\beta r^2} - \frac{k_4}{k_3}\epsilon r^2 \sqrt{1+\beta r^2}\right)
	- 2 k_3 \epsilon r^2 \bigl(1+\beta r^2\bigr) f_7 = 0,
\]
which leads to
\[
	f_7 = \frac{(k_1-k_2) k_3 \big(\hbar + \epsilon\hbar\sqrt{1+\beta r^2} - \frac{k_4}{k_3}\epsilon r^2 \sqrt{1+\beta r^2}\big)}{2 k_3 \epsilon r^2 \bigl(1+\beta r^2\bigr)}.
\]
To determine $V_0$ from~\eqref{eq:feq_4}, we first consider $k_1+k_2\neq0$. The complementary possibility~${k_1+k_2=0}$ is treated in Case II below. Equation~\eqref{eq:feq_4} reduces to
\begin{gather}
	4 k_3^2 \epsilon \bigl(r + \beta r^3\bigr)^3 V'_0
	+ k_3^2 \hbar^2 \bigl[
	\sqrt{1+\beta r^2} \bigl(8 + 20 \beta r^2 + 9 \beta^2 r^4\bigr)
	\notag\\
	\qquad {} + 2 \epsilon \bigl(4 + 12 \beta r^2 + 9 \beta^2 r^4 + 3 \beta^3 r^6\bigr)
	\bigr] - 3\beta k_3 k_4 \hbar r^4\sqrt{1+\beta r^2} = 0,
\end{gather}
which yields
\[
	V_0 =
	-\frac{
		\hbar \bigl[
		k_4 r^2 \sqrt{1+\beta r^2}
		- k_3 \hbar \bigl(
		\sqrt{1+\beta r^2} \bigl(4 + 3 \beta r^2\bigr)
		+ \epsilon \bigl(4 + 6 \beta r^2 + 3 \beta^2 r^4\bigr)
		\bigr)
		\bigr]
	}{
		4 k_3 \epsilon r^2 \bigl(1+\beta r^2\bigr)^2
	}
	+ k_5.
\]
Substituting this condition into the remaining zeroth-order determining equations, we obtain
\smash{$
(k_4+\beta\hbar k_3)
\bigl[2\epsilon(k_4+\beta\hbar k_3)\sqrt{1+\beta r^2}
+3\beta\hbar k_3\bigr]=0$}.
For $\beta\neq0$, this identity on the radial interval implies
$k_3 = -\frac{k_4}{\beta \hbar}$.

For $\beta=0$, the undivided compatibility condition gives $k_4=0$. We use $\beta\neq0$ in the normalization of $Y_P^{15}$ in~\eqref{eq:2.88}. In this case, all determining equations arising from the coefficients of the first- and zeroth-order terms are satisfied. Hence, the full set of determining equations is satisfied. We thus obtain three arbitrary constants. The corresponding integrals of motion for~$k_1$,~$k_2$ and $k_4$ are given, respectively, by
$Y_P^{13}$, $Y_P^{14}$ and $Y_P^{15}$ in~\eqref{eq:2.86}--\eqref{eq:2.88}.

{\bf B.} $k_3=0$.
Substituting this condition into the determining equation~\eqref{eq:feq_98} yields $f_1 = 0$. Equation~\eqref{eq:feq_8} reduces to
\begin{gather*}
	(k_1 - k_2)\bigl[\hbar^2\bigl(3 + \beta^2 r^4 + 3\epsilon\sqrt{1 + \beta r^2} + \beta r^2 \bigl(5 + 3\epsilon\sqrt{1 + \beta r^2}\bigr)\bigr) - 4\bigl(r + \beta r^3\bigr)^2 V_2 \bigr] \notag \\
	 \qquad {}- 2 \epsilon \hbar r^2 \bigl(1 + \beta r^2\bigr) f_7 = 0.
\end{gather*}
Next, we solve this equation for the potential $V_2$. This requires the factor $k_1-k_2$ to be nonzero. We therefore split the analysis into two cases, depending on whether $k_1-k_2$ vanishes.

{\bf a.} $k_1-k_2\neq 0$
\begin{align}
	V_2 ={}& \frac{\hbar^2}{4 r^2 \bigl(1 + \beta r^2\bigr)^2}
	\bigl[ 3 + \beta^2 r^4 + 3 \epsilon \sqrt{1 + \beta r^2}
	+ \beta r^2 \bigl( 5 + 3 \epsilon \sqrt{1 + \beta r^2} \bigr) \bigr] \nonumber \\
 & - \frac{\epsilon\hbar f_7}{2 (k_1-k_2)\bigl(1 + \beta r^2\bigr)}. \label{eq:2.90}
\end{align}
By~\eqref{eq:feq_10}, we obtain
\begin{equation}
	\label{eq:2.91}
	(k_1 - k_2) \epsilon  \hbar^2 \bigl(\epsilon + \sqrt{1 + \beta r^2}\bigr)
	- 2 r^2 \sqrt{1 + \beta r^2} \bigl[ \hbar f_7 + 2 (k_1 - k_2) \epsilon V_2 + \hbar r f'_7 \bigr] = 0.
\end{equation}
Substituting equation~\eqref{eq:2.90} into~\eqref{eq:2.91}, we obtain
\begin{equation}
	\label{eq:2.92}
	f_7 = \frac{(k_1 - k_2) \hbar \bigl( 1 + \beta r^2 + \epsilon \sqrt{1 + \beta r^2} \bigr)}{2 r^2 \bigl(1 + \beta r^2\bigr)^{3/2}}
	+ \frac{k_6}{\sqrt{1 + \beta r^2}}.
\end{equation}
Substituting~\eqref{eq:2.92} into~\eqref{eq:2.90} results in
\[
	V_2 =
	\frac{
		\hbar^2 \bigl[
		\epsilon \bigl(1 + \beta r^2\bigr) \bigl(2 + 3 \beta r^2\bigr)
		+ \sqrt{1 + \beta r^2} \bigl( 2 + \beta r^2 \bigl(5 + \beta r^2\bigr) \bigr)
		\bigr]
	}{
		4 r^2 \bigl(1 + \beta r^2\bigr)^{5/2}
	}
	- \frac{
		k_6 \epsilon \hbar
	}{
		2 (k_1 - k_2) \bigl(1 + \beta r^2\bigr)^{3/2}
	}.
\]
Using~\eqref{eq:feq_4}, we find that
\begin{gather*}
	(k_1 + k_2) \bigl[
	\hbar \bigl( 6 \beta k_6 r^4 \bigl(1 + \beta r^2\bigr)
	+ k_1 \hbar \bigl( \bigl(2 + 3 \beta r^2\bigr)\bigl(\epsilon + \beta \epsilon r^2\bigr)^2 \\
	\qquad
	{}+ 3 \bigl(2 + 7 \beta r^2 + 7 \beta^2 r^4 + 2 \beta^3 r^6\bigr) \bigr) \bigr) + \hbar^2 k_1 \bigl( 2 \epsilon \sqrt{1 + \beta r^2}
	\bigl(4 + 12 \beta r^2 + 9 \beta^2 r^4 + 3 \beta^3 r^6\bigr) \bigr) \\
	\qquad
	{}- k_2 \hbar^2 \bigl( \bigl(2 + 3 \beta r^2\bigr)\bigl(\epsilon + \beta \epsilon r^2\bigr)^2 \bigr) - \hbar^2 k_2 \bigl( 3 \bigl(2 + 7 \beta r^2 + 7 \beta^2 r^4 + 2 \beta^3 r^6\bigr) \\
	\qquad
	{}+ 2 \epsilon \sqrt{1 + \beta r^2} \bigl(4 + 12 \beta r^2 + 9 \beta^2 r^4 + 3 \beta^3 r^6\bigr) \bigr) + 4 (k_1 - k_2) \epsilon r^3 \bigl(1 + \beta r^2\bigr)^{7/2} V'_0
	\bigr] = 0.
\end{gather*}
Hence, the analysis splits into the following two cases.

{\bf a1.} $k_1 + k_2 \neq 0$
\begin{align*}
	V_0 ={}& \frac{\hbar^2}{4 (k_1 - k_2) \epsilon r^2 \bigl(1 + \beta r^2\bigr)^{5/2}}
	\biggl[ \frac{2 k_6}{\hbar}r^2 \bigl(1 + \beta r^2\bigr)
	+ k_1 \bigl( 3 + 5 \beta r^2 + 2 \beta^2 r^4 \bigr) \\
	&  {}+ k_1 \epsilon \sqrt{1 + \beta r^2} \bigl(4 + 6 \beta r^2 + 3 \beta^2 r^4\bigr)
	- k_2 \bigl( 3 + 5 \beta r^2 + 2 \beta^2 r^4 \bigr) \biggr] \\
	&{} - \frac{\hbar^2}{4 (k_1 - k_2) \epsilon r^2 \bigl(1 + \beta r^2\bigr)^{5/2}}
	\bigl[ k_2 \bigl(\epsilon + \beta \epsilon r^2\bigr)^2 + k_2 \epsilon \sqrt{1 + \beta r^2} \bigl(4 + 6 \beta r^2 + 3 \beta^2 r^4\bigr) \\
	&{}- k_1 \bigl(\epsilon + \beta \epsilon r^2\bigr)^2 \bigr] + k_7.
\end{align*}
The remaining zeroth-order determining equations imply
$2 k_6 - \beta \hbar (k_1 - k_2) = 0$,
which gives
\smash{$
	k_6 = \frac{\beta \hbar (k_1 - k_2)}{2}$}.

In this case, all determining equations arising from the coefficients of the first- and zeroth-order terms are satisfied. Hence, the full set of determining equations is satisfied. The corresponding integrals of motion for $k_1$ and $k_2$ are $Y_P^{13}$ and $Y_P^{14}$, given in
\eqref{eq:2.86} and~\eqref{eq:2.87}, respectively.

Notice that the Hamiltonian obtained here coincides with that in Case A, with $k_7=k_5$. Hence, this branch gives no new Hamiltonian or integral of motion.

{\bf a2.} $k_1 + k_2 = 0$.
From equation~\eqref{eq:feq_12}, we get
\begin{gather*}
	 k_1^2 \hbar^3 \bigl[ 8 + 8 \epsilon \sqrt{1+\beta r^2} + 4 \beta r^2 \bigl(7 + 6 \epsilon \sqrt{1+\beta r^2}\bigr)
	+ \beta^2 r^4 \bigl(23 + 14 \epsilon \sqrt{1+\beta r^2}\bigr)
	\\
 \qquad{}+ \beta^3 r^6 \bigl(3 + 2 \epsilon \sqrt{1+\beta r^2}\bigr) \bigr] + 4 k_1^2 \epsilon \hbar r^3 \bigl(1 + \beta r^2\bigr)^{7/2} V'_0
	 \\
	\qquad{}+ \beta k_1 k_6 \hbar^2 r^4 \bigl(1+\beta r^2\bigr) \bigl(9 + 8 \epsilon \sqrt{1+\beta r^2}\bigr)
	- 4 k_6^2 \epsilon \hbar r^4 \bigl(1+\beta r^2\bigr)^{3/2} = 0.
\end{gather*}
Here $k_1\neq0$. For $\beta\neq0$, solving this for the potential $V_0$, we obtain
\begin{align}
	V_0 ={}& \frac{\hbar^2 \sqrt{1+\beta r^2} \bigl(4+\beta r^2\bigr) + \epsilon \hbar^2 \bigl(2+\beta r^2\bigr)^2 + 4\epsilon \bigl(r+\beta r^3\bigr)^2 k_8}{4 \epsilon \bigl(r+\beta r^3\bigr)^2} \notag \\
	& {}+ \frac{k_6 \hbar \bigl(3 + 4 \epsilon \sqrt{1+\beta r^2}\bigr)}{4 k_1 \epsilon \bigl(1+\beta r^2\bigr)^{3/2}}
	- \frac{k_6^2}{2 \beta k_1^2 \bigl(1+\beta r^2\bigr)}.\label{eq:2.101}
\end{align}

In this case, all determining equations arising from the coefficients of the first- and zeroth-order terms are satisfied. Hence, the full set of determining equations is satisfied. We define~$\smash{\frac{k_6}{k_1}=\kappa_1}$,
where $\kappa_1$ is a~real constant. For $\beta=0$, the preceding undivided equation instead gives
\[
V_0=k_8+\hbar^2(1+\epsilon)/r^2+\kappa_1^2r^2/2.
\]
 The corresponding integral of motion for $k_1$ is given by~$Y_P^{16}$ in~\eqref{eq:2.103}.

{\bf b.} $k_1-k_2=0$.
Substituting this condition into the determining equation~\eqref{eq:feq_8} yields $f_7 = 0$. From~\eqref{eq:feq_12}, we find that
\begin{gather*}
	k_1 \bigl[
	-2\epsilon\hbar^2 \bigl( 3 + \beta r^2 \bigl(3 + \beta r^2\bigr)\bigl(3 + 2\beta r^2\bigr) \bigr)
	- 3 \hbar^2 \bigl(1+\beta r^2\bigr)^{3/2} \bigl(2 + 3\beta r^2\bigr) \notag \\
	\qquad {}- 2 \epsilon \bigl(r + \beta r^3\bigr)^3 (V'_0 + V'_2)
	\bigr] = 0,
\end{gather*}
which directly yields two cases depending on whether $k_1$ vanishes.

{\bf b1.} $k_1\neq0$
\[
	V_2 = \frac{\hbar^2 \bigl[ 6 \bigl(1 + \beta r^2\bigr)^{3/2} + \epsilon \bigl(2 + \beta r^2\bigr)\bigl(3 + 4 \beta r^2\bigr) \bigr]}{4 \epsilon r^2 \bigl(1 + \beta r^2\bigr)^2} + k_9 - V_0.
\]

All the determining equations derived from the coefficients of the first- and zeroth-order terms are satisfied. The potential $V_0$ remains arbitrary. Hence, we have only one arbitrary constant~$k_1$, and the corresponding integral of motion for this constant reads
$Y_P^{17}$ in~\eqref{eq:2.106}.

{\bf b2.} $k_1=0$.
Thus, all determining equations are satisfied. The potentials $V_0$ and $V_2$ remain arbitrary. The remaining coefficients give $Y_p=0$, so this coefficient subcase produces no nontrivial integral of motion.

{\bf II.} $k_1+k_2=0$.
Equation~\eqref{eq:feq_65} reduces to{\samepage
\begin{gather*}
	k_3 \bigl[
	\hbar^2 \bigl( 4 + 7 \beta r^2 + 2 \beta^2 r^4
	+ 4 \epsilon \bigl(1+\beta r^2\bigr)^{3/2} \bigr) \notag \\
	\qquad{}- 8 \bigl(r + \beta r^3\bigr)^2 V_2
	- 4 r^4 \bigl(1+\beta r^2\bigr)^2 V_3
	\bigr]
	+ 2 \epsilon \hbar r^2 \bigl(1+\beta r^2\bigr) f_1 = 0.	\label{eq:2.107}
\end{gather*}
To solve this equation for $V_3$, we need to consider two cases based on $k_3$.}

{\bf A.} $k_3\neq0$.
In this case
\begin{align}
	V_3 ={}&
	\frac{\hbar^2}{4r^4\bigl(1+\beta r^2\bigr)^2}
	\bigl[
	4 + 2\beta^2 r^4
	+ \beta r^2 \bigl(7 + 4\epsilon\sqrt{1+\beta r^2}\bigr)
	+ 4\epsilon\sqrt{1+\beta r^2}
	\bigr] \notag\\
&{}+ \frac{\epsilon \hbar f_1}{2k_3 r^2\bigl(1+\beta r^2\bigr)}
	- \frac{2V_2}{r^2}.\label{eq:2.108}
\end{align}
Equation~\eqref{eq:feq_8} reads
\begin{align}
	\label{eq:2.109}
	k_1 \bigl[
	&\hbar^2\bigl(4 + 2\beta^2 r^4 + 4\epsilon\sqrt{1+\beta r^2}
	+ \beta r^2\bigl(7 + 4\epsilon\sqrt{1+\beta r^2}\bigr)\bigr) \notag \\
	&{}- 4r^2\bigl(1+\beta r^2\bigr)^2\bigl(2V_2 + r^2 V_3\bigr)
	\bigr]
	- 2\epsilon \hbar r^2\bigl(1+\beta r^2\bigr) f_7 = 0.
\end{align}
Substituting~\eqref{eq:2.108} into~\eqref{eq:2.109} yields
\[
	k_1 f_1 + k_3 f_7 = 0 \;\;\Rightarrow\;\; f_7 = -\frac{k_1 f_1}{k_3}.
\]
Multiplying~\eqref{eq:2.74} by $r/2$, we have
\[
	k_3 \hbar^2\left(\epsilon + \frac{1}{\sqrt{1+\beta r^2}}\right)
	+ \hbar r^2\bigl(f_1 + r f'_1\bigr)
	- 4k_3 \epsilon r^2 V_2 = 0,
\]
which can be solved for $V_2$ as
\[
	V_2 =
	\frac{k_3\hbar^2 \bigl(\epsilon + \frac{1}{\sqrt{1+\beta r^2}}\bigr)
		+ \hbar r^2\bigl(f_1 + r f'_1\bigr)}
	{4k_3 \epsilon r^2}.
\]
The remaining zeroth-order determining equations give
\begin{gather*}
	k_3^2 \hbar^2 \biggl[
	7\epsilon\bigl(1+\beta r^2\bigr)^2\bigl(2+3\beta r^2\bigr)
	+ 2\sqrt{1+\beta r^2}\bigl(7+21\beta r^2+21\beta^2 r^4+5\beta^3 r^6\bigr)  \\
 \quad {}+ \frac{\epsilon r^4\bigl(1+\beta r^2\bigr)^{7/2}}{k_3\hbar} f''_1
	\biggr]
	- 4r^4\bigl(1+\beta r^2\bigr)^{7/2} f_1^2 \\
	\quad{}- 4r^4\bigl(1+\beta r^2\bigr) f_1 \bigl[
	\beta k_3 \hbar\bigl(1+2\beta r^2+\beta^2 r^4 - \epsilon\sqrt{1+\beta r^2}\bigr)
	+ r\bigl(1+\beta r^2\bigr)^{5/2} f'_1
	\bigr] \\
	\quad{}- k_3 \bigl[
	2\hbar r^3\bigl(1+\beta r^2\bigr)^2
	\bigl(2+2\beta^2 r^4+\beta r^2\bigl(4-\epsilon\sqrt{1+\beta r^2}\bigr)\bigr) f'_1
	 - 4k_3 r^3\bigl(1+\beta r^2\bigr)^{7/2} V'_0
	\bigr]= 0.
\end{gather*}
Solving this equation for $V_0$, we get
\begin{align*}
	V_0 ={}&
	\frac{1}{4 k_3 r^2 \bigl(1+\beta r^2\bigr)^2} \bigl[
	k_3 \hbar^2 \bigl( 7 + 5 \beta^2 r^4 + 7 \epsilon \sqrt{1+\beta r^2}
	+ \beta r^2 \bigl(13 + 7 \epsilon \sqrt{1+\beta r^2}\bigr) \bigr) \notag\\
	&- \epsilon \hbar r^3 \bigl(1+\beta r^2\bigr)^2 f'_1
	+ \hbar r^2\bigl(1+\beta r^2\bigr) f_1 \bigl( \epsilon \bigl(1-\beta r^2\bigr) + 4 \bigl(1+\beta r^2\bigr)^{3/2} \bigr)
	\bigr]+ \frac{r^2 f_1^2}{2 k_3^2} + k_{10}.
\end{align*}

All determining equations corresponding to the first- and zeroth-order
terms are satisfied. The constant $k_1$ does not appear in the Hamiltonian.
The corresponding integral of motion is~$Y_P^{18}$ in~\eqref{eq:2.115}. After subtracting $k_1Y_P^{18}$ from $Y_p$, the remaining terms give the integral of motion~$Y_P^{19}$ in~\eqref{eq:2.116}.

{\bf B.} $k_3=0$.
Substituting this condition into the determining equation~\eqref{eq:feq_65} yields $f_1 = 0$. Then from~\eqref{eq:feq_8}
\begin{gather*}
\begin{split}
&	 k_1 \bigl[
	\hbar^2 \bigl[ 4 + 2 \beta^2 r^4 + 4 \epsilon \sqrt{1+\beta r^2}
	+ \beta r^2 \bigl( 7 + 4 \epsilon \sqrt{1+\beta r^2} \bigr) \bigr] \\
&	 \qquad{} - 8 \bigl(r + \beta r^3\bigr)^2 V_2 - 4 r^4 \bigl(1+\beta r^2\bigr)^2 V_3
	\bigr]
	- 2 \epsilon \hbar r^2 \bigl(1+\beta r^2\bigr) f_7=0.
\end{split}
\end{gather*}
To solve this equation for the potential $V_2$, the factor $k_1$ is required to be nonzero. Thus we split the analysis into two cases, depending on whether $k_1$ vanishes.

{\bf a.} $k_1\neq0$.
Now we have
\begin{gather}
	\label{eq:2.118}
	V_2 = \frac{\hbar^2 \bigl[ 4 + 2\beta^2 r^4 + 4\epsilon \sqrt{1+\beta r^2}
		+ \beta r^2 \bigl(7 + 4\epsilon \sqrt{1+\beta r^2}\bigr) \bigr]}{8 r^2 \bigl(1+\beta r^2\bigr)^2}
	- \frac{\epsilon \hbar f_7}{4 k_1 \bigl(1+\beta r^2\bigr)} - \frac{r^2 V_3}{2}.\!\!\!
\end{gather}
Equation~\eqref{eq:feq_93} gives
\begin{equation}
	\label{eq:2.119}
	2 k_1 \epsilon \bigl(4 r^2 V_2 - \hbar^2\bigr)
	+ 2 \hbar r^2 \bigl(r f'_7 + f_7\bigr)
	- \frac{2 k_1 \hbar^2}{\sqrt{1+\beta r^2}} = 0.
\end{equation}
Solving~\eqref{eq:2.119} for $V_2$, we obtain
\[
 V_2=\frac{\hbar^2}{4r^2}
 \left(1+\frac{\epsilon}{\sqrt{1+\beta r^2}}\right)
 -\frac{\epsilon\hbar}{4k_1}(rf_7'+f_7).
\]
Substituting this expression into~\eqref{eq:2.118}, we find that
\[
 V_3=\frac{\hbar^2}{4r^4}
 \left[
 \frac{2\epsilon}{\sqrt{1+\beta r^2}}
 +\frac{2+3\beta r^2}{\bigl(1+\beta r^2\bigr)^2}
 \right]
 +\frac{\epsilon\hbar f_7'}{2k_1r}
 +\frac{\epsilon\beta\hbar f_7}{2k_1\bigl(1+\beta r^2\bigr)}.
\]
The remaining first-order determining equations are then satisfied.
Next, considering the determining equation~\eqref{eq:feq_12}
for $V_0$, we have
\begin{gather*}
\frac{\rm d}{{\rm d}r}\Bigg\{V_0
 -\frac{\hbar^2}{4r^2}
 \left[
 \frac{7+13\beta r^2+5\beta^2r^4}{\bigl(1+\beta r^2\bigr)^2}
 +\frac{7\epsilon}{\sqrt{1+\beta r^2}}
 \right]
 -\frac{\epsilon\hbar r f_7'}{4k_1}\\
\qquad{}+\frac{\hbar f_7}{4k_1}
 \left[
 \frac{\epsilon\bigl(1-\beta r^2\bigr)}{1+\beta r^2}
 +4\sqrt{1+\beta r^2}
 \right]
 -\frac{r^2f_7^2}{2k_1^2}\Bigg\}=0.
\end{gather*}
By solving this equation for $V_0$, we find that
\begin{align*}
 V_0={}&k_{10}+\frac{\hbar^2}{4r^2}
 \left[
 \frac{7+13\beta r^2+5\beta^2r^4}{\bigl(1+\beta r^2\bigr)^2}
 +\frac{7\epsilon}{\sqrt{1+\beta r^2}}
 \right]
 +\frac{\epsilon\hbar r f_7'}{4k_1}\\
&{}-\frac{\hbar f_7}{4k_1}
 \left[
 \frac{\epsilon\bigl(1-\beta r^2\bigr)}{1+\beta r^2}
 +4\sqrt{1+\beta r^2}
 \right]
 +\frac{r^2f_7^2}{2k_1^2},
\end{align*}
where $k_{10}$ is a~real integration constant. Substituting these
expressions into the remaining zeroth-order determining equations,
we find that they are identically satisfied. Hence all the determining
equations are satisfied, and $f_7(r)$ remains arbitrary. The
corresponding integral of motion is
\begin{align*}
&\sqrt{1+\beta r^2}(\vec\sigma_1-\vec\sigma_2,\vec p)
-\frac{\epsilon+\sqrt{1+\beta r^2}}{r^2}
 (\vec x,\vec\sigma_1-\vec\sigma_2)(\vec x,\vec p)\\
&\qquad
{}+\frac{{\rm i}\hbar\bigl(\epsilon+\sqrt{1+\beta r^2}\bigr)}{r^2}
 (\vec x,\vec\sigma_1-\vec\sigma_2)
+\frac{f_7}{k_1}(\vec x,\vec\sigma_1\times\vec\sigma_2).
\end{align*}
However, notice that the potentials coincide with those obtained in Case A,
with $f_1/k_3$ there replaced by $-f_7/k_1$ here. Under the same
substitution, the above integral coincides with $Y_P^{18}$ in~\eqref{eq:2.115}. Hence, this subcase gives no new Hamiltonian
or integral of motion.

{\bf b.} $k_1=0$.
Substituting these conditions into the determining equation~\eqref{eq:feq_8} leads to $f_7 = 0$. The remaining coefficients give $Y_p=0$. Hence, the full set of determining equations is satisfied, but no nontrivial integral of motion arises in this coefficient subcase.

{\bf S2.} $f_3=0$. 
Letting $f_3=0$, we obtain $f_5=0$ from~\eqref{eq:2.60}. If $f_4\neq0$, the second-order determining equations give
$
 \hbar^2V_1'+3\hbar rV_1^2-r^3V_1^3=0$,
which reproduces the spin-orbit condition already considered in S1 and the previously treated solution $V_1=\hbar/r^2$. We therefore continue with~${f_4=0}$, and hence $f_6=0$. Equation~\eqref{eq:2.8} reads
$\hbar f_9 + f_8 V_1 + f_{10} \bigl( r^2 V_1 - \hbar \bigr) = 0$,
from which we obtain
\[
	f_9 = \frac{\hbar f_{10} - r^2 f_{10} V_1 - f_8 V_1}{\hbar}.
\]
Equation~\eqref{eq:2.9} becomes
\[
	\hbar V_1 \bigl( r^2 f'_{10} + f'_8 \bigr)
	+ f_8 \bigl( r V_1^2 + \hbar V'_1 \bigr)
	+ r f_{10} \bigl( 2 \hbar V_1 + r^2 V_1^2 + \hbar r V'_1 \bigr) = 0.
\]
Now we consider the following two cases.

{\bf I.} $V_1\neq0$.
\begin{gather}
	\label{eq:2.134}
	f'_{10}=\frac{-f_8\bigl(r {V_1}^2+\hbar V'_{1}\bigr)-\hbar V_1f'_8-rf_{10}\bigl(2\hbar V_1+r^2 {V_1}^2+\hbar r V'_{1}\bigr)}{\hbar r^2 V_1}.
\end{gather}

Equation~\eqref{eq:2.12} reads
\begin{gather}
	\label{eq:2.135}
	r \bigl( 2 \hbar - r^2 V_1 \bigr) \bigl( f_{10} \bigl(\hbar - r^2 V_1\bigr) - f_8 V_1 \bigr)
	+ \hbar^2 \bigl( r^2 f'_{10} + f'_8 \bigr)
	+ \hbar r^3 f_{10} V_1 = 0.
\end{gather}
Substituting equation~\eqref{eq:2.134} into~\eqref{eq:2.135} yields
$
 \bigl(r^2f_{10}+f_8\bigr)
 \bigl(\hbar^2V_1'+3\hbar rV_1^2-r^3V_1^3\bigr)=0$.
If the second factor vanishes identically, we return to the previously considered spin-orbit family. In the remaining case, the second-order equations, including
\[
 \hbar^2\frac{\rm d}{{\rm d}r}\bigl(r^2f_{10}+f_8\bigr)
 -rV_1\bigl(2\hbar-r^2V_1\bigr)\bigl(r^2f_{10}+f_8\bigr)=0,
\]
imply
\[
	r^2 f_{10} + f_8 = 0 \ \Rightarrow\ f_{10} = - \frac{f_8}{r^2}.
\]
Thus, all determining equations arising from the coefficients of the quadratic terms are satisfied. Next, we consider the determining equations obtained from the coefficients of the first- and zeroth-order terms. Outside the previously treated case $V_1=\hbar/r^2$, substituting these conditions into the remaining determining equations leads to $f_1=0$, $f_2=0$ and $f_7=0$. Hence, all the first- and zeroth-order determining equations are satisfied, and only the zero operator remains.

{\bf II.} $V_1=0$.
From the determining equation~\eqref{eq:2.5}
$2 r f_{10} + r^2 f'_{10} + f'_{8} = 0$,
we obtain
\[
	f_{10} = \frac{k_{13}}{r^2} - \frac{f_8}{r^2}.
\]

With this choice, all determining equations arising from the coefficients of the quadratic terms are satisfied. We now consider the determining equations obtained from the coefficients of the first- and zeroth-order terms. For $k_{13}\neq0$, these give $f_2=-f_1$, $f_7=0$ and
\[
 V_2=-\frac{\hbar}{4k_{13}}(f_1+rf_1'),\qquad
 V_3=\frac{\hbar f_1'}{2k_{13}r},\qquad
 V_0=\kappa_2+\frac{r^2f_1^2}{2k_{13}^2}
 +\frac{\hbar}{4k_{13}}(rf_1'+3f_1),
\]
where $f_1(r)$ is arbitrary. This is the family obtained in Case~S1.II.A at $\beta=0$, $\epsilon=-1$, with $k_3=k_{13}$ and $k_{10}=\kappa_2$. If $k_{13}=0$, the remaining equations instead give only $Y_p=0$.

{\bf Case 2. $\boldsymbol{V_{4}=0}$, $\boldsymbol{V_{5} \neq 0}$.}
From the determining equations arising from the third-order terms, we obtain
\[
	f_3 = f_4 = f_5 = f_6 = 0, \qquad f_8 = -r^2 f_{10}, \qquad f_9 = f_{10}, \qquad f_{11} = -f_{10}.
\]
We next examine the quadratic-level equations. Substituting the above relations into~\eqref{eq:2.5}, \eqref{eq:2.6} and~\eqref{eq:2.8} yields $f_1 = f_2 = f_7 = 0$. The remaining arbitrary function $f_{10}(r)$ represents the zero operator through the relation among $P_8$, $P_9$, $P_{10}$, $P_{11}$. Thus $Y_p=0$, and the first- and zeroth-order determining equations are identically satisfied. Hence this case does not admit any nontrivial pseudo-scalar integral of the assumed form.

{\bf Case 3. $\boldsymbol{V_{4}\neq0}$, $\boldsymbol{V_{5}=0}$.}
From the determining equations corresponding to the third-order terms,
we obtain
\[
 f_5=f_6=0,\qquad f_9=f_{10},\qquad f_{11}=-f_{10}.
\]
We now consider the determining equations at the quadratic level.
Substituting these relations into~\eqref{eq:2.6} and
\eqref{eq:2.13}, and taking their sum and difference, we obtain
\begin{equation}
\label{eq:2.143}
 (1+2V_4)(f_3'+f_4')-V_4'(f_3+f_4)=0,\qquad
 (2V_4-1)(f_3'-f_4')-V_4'(f_3-f_4)=0.
\end{equation}
The remaining equations at this level give
\begin{gather}
 f_1+f_2=0,\qquad
 (2V_4-1)\bigl(f_8'+r^2f_{10}'+2rf_{10}\bigr)
 -V_4'\bigl(f_8+r^2f_{10}\bigr)=0,\nonumber\\
 \hbar V_4(f_3'+f_4')+rV_1(f_3+f_4)=0,\qquad
 \hbar V_4(f_3'-f_4')-rV_1(f_3-f_4)-4rV_4f_7=0,\nonumber\\
 \hbar V_4\bigl(f_8'+r^2f_{10}'+2rf_{10}\bigr)
 -rV_1\bigl(f_8+r^2f_{10}\bigr)+rV_4(f_1-f_2)=0.\label{eq:2.144}
\end{gather}
Equations~\eqref{eq:2.143} and~\eqref{eq:2.144} are equivalent to
all the second-order determining equations under the above
third-order conditions. We distinguish two subcases according to
whether $V_4$ is nonconstant or constant.

{\bf Subcase 1.} $V_4$ nonconstant.
From~\eqref{eq:2.144}, we first obtain $f_2=-f_1$ and
\begin{align}
\label{eq:c3.step.f7}
 f_7={}&\frac{\hbar}{4r}(f_3'-f_4')
 -\frac{V_1}{4V_4}(f_3-f_4),\\
\label{eq:c3.step.f1}
 f_1={}&\frac{V_1}{2V_4}\bigl(f_8+r^2f_{10}\bigr)
 -\frac{\hbar}{2r}\bigl(f_8'+r^2f_{10}'+2rf_{10}\bigr).
\end{align}
Multiplying the third equation in~\eqref{eq:2.144} by $1+2V_4$
and subtracting $\hbar V_4$ times the first equation in
\eqref{eq:2.143}, we obtain
\begin{equation}
\label{eq:c3.step.product}
 (f_3+f_4)\bigl[\hbar V_4V_4'+rV_1(1+2V_4)\bigr]=0.
\end{equation}
Using~\eqref{eq:c3.step.product} and the third equation in
\eqref{eq:2.144}, we distinguish the cases $f_3+f_4\neq0$ and
$f_3+f_4=0$. We first consider $f_3+f_4\neq0$.

{\bf I.} $f_3+f_4\neq0$.
The first equation in~\eqref{eq:2.143} gives
\[
 \frac{\rm d}{{\rm d}r}
 \left(\frac{1+2V_4}{(f_3+f_4)^2}\right)=0.
\]
Hence,
\begin{equation}
\label{eq:c3.step.I.V4}
 V_4=\frac{d_1(f_3+f_4)^2-1}{2},
\end{equation}
where $d_1$ is a~real integration constant. Since $V_4$ is
nonconstant, we have $d_1\neq0$. Substituting~\eqref{eq:c3.step.I.V4} into the third equation in
\eqref{eq:2.144}, we obtain
\begin{equation}
\label{eq:c3.step.I.V1}
 V_1=
 -\frac{\hbar\bigl[d_1(f_3+f_4)^2-1\bigr](f_3'+f_4')}
 {2r(f_3+f_4)}.
\end{equation}
To distinguish the possible subcases, we use the second equation in
\eqref{eq:2.143},
\[
 (2V_4-1)(f_3'-f_4')-V_4'(f_3-f_4)=0,
\]
and its derivative. Substituting~\eqref{eq:c3.step.I.V4},
\eqref{eq:c3.step.I.V1} and~\eqref{eq:c3.step.f7} into
\eqref{eq:feq_93}, and using these relations, we obtain
\begin{align}
0={}&(f_3-f_4)\biggl\{
 4(f_3+f_4)^2\bigl[d_1(f_3+f_4)^2-2\bigr]V_2
 -\frac{\hbar^2}{r}\bigl[d_1(f_3+f_4)^2-2\bigr]
 (f_3+f_4)(f_3'+f_4')\nonumber
 \\
& {}+\hbar^2\bigl[d_1(f_3+f_4)^2-2\bigr]
 \bigl[d_1(f_3+f_4)^2-1\bigr]
 (f_3+f_4)(f_3''+f_4'')
-2\hbar^2(f_3'+f_4')^2\biggr\}.\label{eq:c3.step.IA.V2_equation}
\end{align}
We therefore distinguish the cases $f_3-f_4\neq0$
and $f_3-f_4=0$.

{\bf A.} $f_3-f_4\neq0$.
The second equation in~\eqref{eq:2.143} gives
\[
 \frac{\rm d}{{\rm d}r}
 \left(\frac{2V_4-1}{(f_3-f_4)^2}\right)=0.
\]
Hence,
\begin{equation}
\label{eq:c3.step.IA.V4}
 V_4=\frac{1+d_2(f_3-f_4)^2}{2},
\end{equation}
where $d_2$ is a~real constant. Since $V_4$ is nonconstant,
$d_2\neq0$, and consequently $2V_4-1\neq0$ on the interval under
consideration. Comparing~\eqref{eq:c3.step.IA.V4} with
\eqref{eq:c3.step.I.V4}, we obtain
\begin{equation}
\label{eq:c3.step.IA.constraint}
 d_1(f_3+f_4)^2-d_2(f_3-f_4)^2=2.
\end{equation}
Differentiating this relation and using it to eliminate
$d_2(f_3-f_4)^2$, we find
\[
 f_3'-f_4'
 =\frac{d_1(f_3+f_4)(f_3'+f_4')}
 {d_1(f_3+f_4)^2-2}(f_3-f_4).
\]

Multiplying the second equation in~\eqref{eq:2.144} by $f_3-f_4$
and subtracting $f_8+r^2f_{10}$ times the second equation in
\eqref{eq:2.143}, we obtain
\[
 (f_3-f_4)\bigl(f_8'+r^2f_{10}'+2rf_{10}\bigr)
 -\bigl(f_8+r^2f_{10}\bigr)(f_3'-f_4')=0,
\]
where we have used $2V_4-1\neq0$. Therefore,
\[
 \frac{\rm d}{{\rm d}r}
 \left(\frac{f_8+r^2f_{10}}{f_3-f_4}\right)=0,
\]
and hence
\begin{equation}
\label{eq:c3.step.IA.f10}
 f_8+r^2f_{10}=d_3(f_3-f_4),\qquad
 f_{10}=\frac{d_3(f_3-f_4)-f_8}{r^2},
\end{equation}
where $d_3$ is an arbitrary real constant.

Substituting~\eqref{eq:c3.step.I.V1} and
\eqref{eq:c3.step.IA.f10} into~\eqref{eq:c3.step.f7} and
\eqref{eq:c3.step.f1}, we obtain
\begin{gather*}
 f_7=\frac{\hbar}{4r}
 \left[f_3'-f_4'
 +\frac{(f_3-f_4)(f_3'+f_4')}{f_3+f_4}\right]
 =\frac{\hbar(f_3f_3'-f_4f_4')}{2r(f_3+f_4)},\\
 f_1=-2d_3f_7
 =-\frac{\hbar d_3(f_3f_3'-f_4f_4')}{r(f_3+f_4)},
 \qquad f_2=-f_1.
\end{gather*}
Together with $f_5=f_6=0$, $f_9=f_{10}$ and $f_{11}=-f_{10}$,
these expressions satisfy all the second-order determining equations.
The function $f_8(r)$ remains arbitrary, while $f_3$ and $f_4$
are related by~\eqref{eq:c3.step.IA.constraint}.

We now consider the determining equations arising from the
coefficients of the first-order terms. Using the above relations,
equation~\eqref{eq:feq_2} reduces to
\begin{align*}
0={}&4r^2(f_3+f_4)V_3
 -\hbar^2\bigl[2d_1(f_3+f_4)^2-1\bigr]
 \left(f_3''+f_4''-\frac{f_3'+f_4'}{r}\right)\\
&{}-\frac{\hbar^2\bigl[1+5d_1(f_3+f_4)^2\bigr]}{2(f_3+f_4)}
 (f_3'+f_4')^2.
\end{align*}
Since $f_3+f_4\neq0$, this gives
\begin{align}
V_3={}&\frac{\hbar^2}{4r^2}\bigg[
 \bigl[2d_1(f_3+f_4)^2-1\bigr]
 \left(\frac{f_3''+f_4''}{f_3+f_4}
 -\frac{f_3'+f_4'}{r(f_3+f_4)}\right)\nonumber\\
 &+\frac{1+5d_1(f_3+f_4)^2}{2}
 \frac{(f_3'+f_4')^2}{(f_3+f_4)^2}\bigg].\label{eq:c3.step.IA.V3}
\end{align}
Since $f_3-f_4\neq0$ and
$d_1(f_3+f_4)^2-2=d_2(f_3-f_4)^2\neq0$, equation
\eqref{eq:c3.step.IA.V2_equation} gives
\begin{equation}
\label{eq:c3.step.IA.V2}
\begin{aligned}
V_2=\frac{\hbar^2}{4}\Bigg[&
 \frac{f_3'+f_4'}{r(f_3+f_4)}
 -\bigl[d_1(f_3+f_4)^2-1\bigr]
 \frac{f_3''+f_4''}{f_3+f_4} +\frac{2(f_3'+f_4')^2}
 {\bigl[d_1(f_3+f_4)^2-2\bigr](f_3+f_4)^2}\Bigg].
\end{aligned}
\end{equation}
Substituting~\eqref{eq:c3.step.IA.V3} and
\eqref{eq:c3.step.IA.V2} into the remaining first-order determining
equations, we find that they are all satisfied. We next consider the
zeroth-order determining equations.

Using the preceding expressions for the coefficients and potentials,
equation~\eqref{eq:feq_4} reduces~to
\[
0=4r^2(f_3+f_4)\frac{\rm d}{{\rm d}r}\Bigg[
 V_0+\frac{\hbar^2}{8}
 \frac{2d_1^2(f_3+f_4)^4-3d_1(f_3+f_4)^2+2}
 {d_1(f_3+f_4)^2-2}
 \left(\frac{f_3'+f_4'}{f_3+f_4}\right)^2\Bigg].
\]
Since $f_3+f_4\neq0$, we obtain
\[
 V_0=d_4-\frac{\hbar^2}{8}
 \frac{2d_1^2(f_3+f_4)^4-3d_1(f_3+f_4)^2+2}
 {d_1(f_3+f_4)^2-2}
 \left(\frac{f_3'+f_4'}{f_3+f_4}\right)^2,
\]
where $d_4$ is an arbitrary real constant. Substituting this expression
into~\eqref{eq:feq_11}, \eqref{eq:feq_12} and \eqref{eq:feq_95},
we find that these equations are identically satisfied. Hence all the
determining equations are satisfied. The constant $d_3$ does not appear
in the potentials. The corresponding integral of motion is
$Y_P^{29}$, given in~\eqref{eq:c3.step.IA.Y35}.
After subtracting $d_3Y_P^{29}$ from $Y_p$, the remaining terms give
$Y_P^{30}$ in~\eqref{eq:c3.step.IA.Y36}, where $f_3$ and $f_4$
satisfy~\eqref{eq:c3.step.IA.constraint}.

{\bf B.} $f_3-f_4=0$. In this case, $f_4=f_3\neq0$, and
\eqref{eq:c3.step.I.V4} and~\eqref{eq:c3.step.I.V1} become
\begin{align}
 V_4=\frac{4d_1f_3^2-1}{2},\qquad
 V_1=-\frac{\hbar\bigl(4d_1f_3^2-1\bigr)f_3'}{2rf_3}.\label{eq:c3.step.IB.V4}
\end{align}
From~\eqref{eq:c3.step.f7} and~\eqref{eq:c3.step.f1}, we obtain
\begin{equation}
\label{eq:c3.step.IB.f1f7}
 f_7=0,\qquad f_2=-f_1,\qquad
 f_1=-\frac{\hbar}{2r}\left[
 f_8'+r^2f_{10}'+2rf_{10}
 +\frac{f_3'}{f_3}\bigl(f_8+r^2f_{10}\bigr)\right].
\end{equation}
The second equation in~\eqref{eq:2.144} reduces to
\begin{equation}
\label{eq:c3.step.IB.coefficient_equation}
 \bigl(2d_1f_3^2-1\bigr)\bigl(f_8'+r^2f_{10}'+2rf_{10}\bigr)
 -2d_1f_3f_3'\bigl(f_8+r^2f_{10}\bigr)=0.
\end{equation}

Using these relations, equation~\eqref{eq:feq_2} gives
\[
 V_3=\frac{\hbar^2}{4r^2}\left[
 \bigl(8d_1f_3^2-1\bigr)\left(\frac{f_3''}{f_3}
 -\frac{f_3'}{rf_3}\right)
 +\frac{1+20d_1f_3^2}{2}\frac{(f_3')^2}{f_3^2}\right].
\]
Substituting~\eqref{eq:c3.step.IB.V4}--\eqref{eq:c3.step.IB.f1f7}
into~\eqref{eq:feq_29}, and using
\eqref{eq:c3.step.IB.coefficient_equation} and its derivative,
we obtain
\begin{align}
0={}&\bigl(f_8+r^2f_{10}\bigr)\biggl[
 4f_3^2\bigl(2d_1f_3^2-1\bigr)V_2
 -\frac{\hbar^2}{r}\bigl(2d_1f_3^2-1\bigr)f_3f_3' \nonumber\\
& {}+\hbar^2\bigl(2d_1f_3^2-1\bigr)\bigl(4d_1f_3^2-1\bigr)f_3f_3''
 -\hbar^2(f_3')^2\biggr].\label{eq:c3.step.IB.product}
\end{align}
We therefore consider the cases $f_8+r^2f_{10}=0$ and
$f_8+r^2f_{10}\neq0$ separately.

{\bf a.} $f_8+r^2f_{10}=0$.
In this case,
$
 f_{10}=-\frac{f_8}{r^2}$, $ f_1=f_2=f_7=0$.
Together with the preceding relations, these expressions satisfy all
the second- and first-order determining equations. We next consider
the zeroth-order terms. Equation~\eqref{eq:feq_4} yields
\[
 \frac{\rm d}{{\rm d}r}\left\{
 V_0+V_2+\frac{\hbar^2}{4}\left[
 \bigl(4d_1f_3^2-1\bigr)\frac{f_3''}{f_3}
 +\frac{1+8d_1f_3^2}{2}\frac{(f_3')^2}{f_3^2}
 -\frac{f_3'}{rf_3}\right]\right\}=0.
\]
Hence,
\[
 V_0=d_4-V_2-\frac{\hbar^2}{4}\left[
 \bigl(4d_1f_3^2-1\bigr)\frac{f_3''}{f_3}
 +\frac{1+8d_1f_3^2}{2}\frac{(f_3')^2}{f_3^2}
 -\frac{f_3'}{rf_3}\right],
\]
where $d_4$ is a~real constant. Substituting this expression into the
remaining zeroth-order determining equations, we find that they are
identically satisfied. Thus all the determining equations are
satisfied, and $V_2(r)$ remains arbitrary. The remaining terms in
$Y_p$ give the following integral of motion:
$Y_P^{31}$ in~\eqref{eq:c3.step.IBa.Y37}.

{\bf b.} $f_8+r^2f_{10}\neq0$.
Equation~\eqref{eq:c3.step.IB.coefficient_equation} gives
\[
 \frac{\rm d}{{\rm d}r}\left[
 \frac{2V_4-1}{\bigl(f_8+r^2f_{10}\bigr)^2}\right]=0.
\]
Therefore,
\smash{$
 V_4=\frac{1+d_2(f_8+r^2f_{10})^2}{2}$},
where $d_2$ is a~real constant. Comparing this expression with equation~\eqref{eq:c3.step.IB.V4}, we obtain
\smash{$
 4d_1f_3^2-d_2\bigl(f_8+r^2f_{10}\bigr)^2=2$}.
Since $V_4$ is nonconstant, $d_2\neq0$. Hence~${2d_1f_3^2-1\neq0}$, and
\eqref{eq:c3.step.IB.f1f7} becomes
\[
 f_1=-\frac{\hbar\bigl(4d_1f_3^2-1\bigr)f_3'\bigl(f_8+r^2f_{10}\bigr)}
 {2rf_3\bigl(2d_1f_3^2-1\bigr)},\qquad f_2=-f_1.
\]
Equation~\eqref{eq:c3.step.IB.product} now gives
\[
 V_2=\frac{\hbar^2}{4}\left[
 \frac{f_3'}{rf_3}-\bigl(4d_1f_3^2-1\bigr)\frac{f_3''}{f_3}
 +\frac{(f_3')^2}{\bigl(2d_1f_3^2-1\bigr)f_3^2}\right].
\]
With these expressions, all the second- and first-order determining
equations are satisfied. Equation~\eqref{eq:feq_4} then yields
\[
 \frac{\rm d}{{\rm d}r}\left[
 V_0+\frac{\hbar^2}{8}
 \frac{16d_1^2f_3^4-6d_1f_3^2+1}{2d_1f_3^2-1}
 \frac{(f_3')^2}{f_3^2}\right]=0,
\]
and consequently
\[
 V_0=d_4-\frac{\hbar^2}{8}
 \frac{16d_1^2f_3^4-6d_1f_3^2+1}{2d_1f_3^2-1}
 \frac{(f_3')^2}{f_3^2},
\]
where $d_4$ is a~real constant. The remaining zeroth-order determining
equations are then identically satisfied.

The potentials coincide with those obtained in Case A, with
$f_3+f_4$ there replaced by $2f_3$ here. Subtracting the integral
$Y_P^{31}$ from $Y_p$, the remaining terms give the integral of motion
\[
Y_p-Y_P^{31}
=\left.Y_P^{29}\right|_{\substack{
 f_3\mapsto f_3+\frac12 (f_8+r^2f_{10} )\\
 f_4\mapsto f_3-\frac12 (f_8+r^2f_{10} )}}.
\]
Here the substitutions are applied to $Y_P^{29}$ in
\eqref{eq:c3.step.IA.Y35}; the functions on the left belong to Case A~and those on the right belong to the present subcase. Hence this subcase gives no new
Hamiltonian family.

{\bf II.} $f_3+f_4=0$.
In this case, $f_4=-f_3$, and equations~\eqref{eq:2.143} and
\eqref{eq:2.144} reduce to
\begin{gather}
 (2V_4-1)f_3'-V_4'f_3=0,\qquad
 (2V_4-1)\bigl(f_8'+r^2f_{10}'+2rf_{10}\bigr)
 -V_4'\bigl(f_8+r^2f_{10}\bigr)=0,\nonumber\\
 f_7=\frac{\hbar f_3'}{2r}-\frac{V_1f_3}{2V_4},\!\qquad
 f_1=\frac{V_1\bigl(f_8+r^2f_{10}\bigr)}{2V_4}
 -\frac{\hbar}{2r}\bigl(f_8'+r^2f_{10}'+2rf_{10}\bigr),
 \qquad\!
 f_2=-f_1.\!\!\!\label{eq:c3.step.II.coefficients}
\end{gather}
Substituting these relations into~\eqref{eq:feq_10}, and using the
first equation in~\eqref{eq:c3.step.II.coefficients} and its derivative,
we obtain
\begin{gather}
\label{eq:c3.step.II.product}
0=f_3\Bigg[4r^2V_4^2V_3-r^2V_4V_1^2
 -\hbar r(V_4V_1'-V_1V_4')-2\hbar^2V_4^2\left(V_4''-\frac{V_4'}r\right)\Bigg].
\end{gather}
We therefore consider $f_3\neq0$ and $f_3=0$ separately.

{\bf A.} $f_3\neq0$.
The first equation in~\eqref{eq:c3.step.II.coefficients} gives
$\frac{\rm d}{{\rm d}r}\bigl(\frac{2V_4-1}{f_3^2}\bigr)=0$.
Hence,
\begin{gather}
\label{eq:c3.step.IIA.V4}
 V_4=\frac{1+d_1f_3^2}{2},
\end{gather}
where $d_1$ is a~real constant. Since $V_4$ is nonconstant,
$d_1\neq0$, and consequently $2V_4-1\neq0$. Multiplying the second
equation in~\eqref{eq:c3.step.II.coefficients} by $f_3$ and subtracting
$f_8+r^2f_{10}$ times the first one, we find
\[
 \frac{\rm d}{{\rm d}r}\left(\frac{f_8+r^2f_{10}}{f_3}\right)=0.
\]
Therefore,
\[
 f_8+r^2f_{10}=d_2f_3,\qquad
 f_{10}=\frac{d_2f_3-f_8}{r^2},
\]
where $d_2$ is an arbitrary real constant. The remaining relations in
\eqref{eq:c3.step.II.coefficients} give
\[
 f_7=\frac{\hbar f_3'}{2r}-\frac{V_1f_3}{1+d_1f_3^2},
 \qquad f_1=-d_2f_7,\qquad f_2=d_2f_7.
\]
Together with $f_4=-f_3$, $f_5=f_6=0$, $f_9=f_{10}$ and
$f_{11}=-f_{10}$, these expressions satisfy all the second-order
determining equations.

We now consider the first-order terms. Substituting
\eqref{eq:c3.step.IIA.V4} into~\eqref{eq:c3.step.II.product}, and
using $f_3\neq0$, we obtain
\begin{equation}
\label{eq:c3.step.IIA.V3}
\begin{aligned}
V_3={}&\frac{V_1^2}{2\bigl(1+d_1f_3^2\bigr)}
 +\frac{\hbar V_1'}{2r\bigl(1+d_1f_3^2\bigr)}
 -\frac{d_1\hbar f_3f_3'V_1}{r\bigl(1+d_1f_3^2\bigr)^2}+\frac{d_1\hbar^2}{2r^2}
 \left[f_3f_3''+(f_3')^2-\frac{f_3f_3'}r\right].
\end{aligned}
\end{equation}
Equation~\eqref{eq:feq_93} then gives
\begin{align}
V_2={}&\frac{\hbar}{4\bigl(1+d_1f_3^2\bigr)}
 \bigl[d_1rf_3^2V_1'+\bigl(2+3d_1f_3^2\bigr)V_1\bigr]\nonumber\\
&{}+\frac{d_1\hbar rf_3f_3'\bigl(1-d_1f_3^2\bigr)V_1}
 {4\bigl(1+d_1f_3^2\bigr)^2}
 -\frac{d_1\hbar^2}{8}f_3\left(f_3''-\frac{2f_3'}r\right).\label{eq:c3.step.IIA.V2}
\end{align}
The remaining first-order determining equations are then identically
satisfied. We next consider the zeroth-order terms. Equation
\eqref{eq:feq_95} reduces to
\[
\frac{\rm d}{{\rm d}r}\left[V_0+V_2
 -\frac{r^2V_1^2}{2\bigl(1+d_1f_3^2\bigr)^2}-\frac{\hbar}{2}\left(rV_1'+3V_1+\frac{V_1}{1+d_1f_3^2}\right)
 -\frac{d_1\hbar^2f_3f_3'}{2r}\right]=0.
\]
Integrating and substituting~\eqref{eq:c3.step.IIA.V2}, we obtain
\begin{align}
V_0={}&d_3+\frac{r^2V_1^2}{2\bigl(1+d_1f_3^2\bigr)^2}
 +\frac{\hbar\bigl(2+d_1f_3^2\bigr)}{4\bigl(1+d_1f_3^2\bigr)}(rV_1'+3V_1)\nonumber\\
&+\frac{d_1\hbar rf_3f_3'\bigl(d_1f_3^2-1\bigr)V_1}{4\bigl(1+d_1f_3^2\bigr)^2}
 +\frac{d_1\hbar^2}{8}f_3\left(f_3''+\frac{2f_3'}r\right),\label{eq:c3.step.IIA.V0}
\end{align}
where $d_3$ is a~real constant. Substituting this expression into the
remaining zeroth-order equations, we find that they are identically
satisfied. Hence all the determining equations are satisfied.
The functions $f_3(r)$ and $V_1(r)$ remain arbitrary, subject to the
conditions of the present subcase.

The constant $d_2$ does not appear in the potentials. The corresponding
integral of motion is~$Y_P^{32}$ in~\eqref{eq:c3.step.IIA.Y38}.
After subtracting $d_2Y_P^{32}$ from $Y_p$, the remaining terms give
another integral of motion,
$Y_P^{33}$ in~\eqref{eq:c3.step.IIA.Y39}.
On an interval where $V_4\neq0,\pm1/2$, the present family also
belongs to the family obtained in Case I.B.a if and only if
\[
 V_1=-\frac{\hbar V_4V_4'}{r(1+2V_4)}.
\]
Under this condition, the potentials coincide with those obtained
in Case~I.A.

{\bf B.} $f_3=0$.
We have $f_3=f_4=f_7=0$, and
\eqref{eq:c3.step.II.coefficients} becomes
\begin{gather}
 (2V_4-1)\bigl(f_8'+r^2f_{10}'+2rf_{10}\bigr)
 -V_4'\bigl(f_8+r^2f_{10}\bigr)=0,\nonumber\\
 f_1=\frac{V_1\bigl(f_8+r^2f_{10}\bigr)}{2V_4}
 -\frac{\hbar}{2r}\bigl(f_8'+r^2f_{10}'+2rf_{10}\bigr),
 \qquad f_2=-f_1.\label{eq:c3.step.IIB.coefficients}
\end{gather}
Adding~\eqref{eq:feq_29} and~\eqref{eq:feq_98}, and using the first
relation in~\eqref{eq:c3.step.IIB.coefficients} and its derivative,
we obtain
\begin{equation}
\label{eq:c3.step.IIB.product}
0=\bigl(f_8+r^2f_{10}\bigr)\left[4r^2V_4^2V_3-r^2V_4V_1^2
 -\hbar r(V_4V_1'-V_1V_4')-2\hbar^2V_4^2\left(V_4''-\frac{V_4'}r\right)\right].
\end{equation}
We therefore consider $f_8+r^2f_{10}=0$ and $f_8+r^2f_{10}\neq0$
separately.

{\bf a.} $f_8+r^2f_{10}=0$.
In this case, $f_{10}=-f_8/r^2$ and $f_1=f_2=0$. All the determining
equations are satisfied for arbitrary $V_0$, $V_1$, $V_2$, $V_3$ and nonconstant
$V_4\neq0$. However, substitution into the symmetric operator gives
$Y_p=0$. Hence this subcase yields no nontrivial integral of motion.

{\bf b.} $f_8+r^2f_{10}\neq0$.
The first equation in~\eqref{eq:c3.step.IIB.coefficients} gives
\[
 \frac{\rm d}{{\rm d}r}\left[
 \frac{2V_4-1}{\bigl(f_8+r^2f_{10}\bigr)^2}\right]=0.
\]
Hence,
\smash{$
 V_4=\frac{1+d_1(f_8+r^2f_{10})^2}{2}$},
where $d_1\neq0$ because $V_4$ is nonconstant. Consequently
$2V_4-1\neq0$, and the second relation in~\eqref{eq:c3.step.IIB.coefficients} can be written as
\[
 f_1=\left[\frac{V_1}{2V_4}
 -\frac{\hbar V_4'}{2r(2V_4-1)}\right]
 \bigl(f_8+r^2f_{10}\bigr),\qquad f_2=-f_1.
\]
All the second-order determining equations are then satisfied.
Equation~\eqref{eq:c3.step.IIB.product} gives
\[
 V_3=\frac{V_1^2}{4V_4}
 +\frac{\hbar}{4r}\left(\frac{V_1}{V_4}\right)'
 +\frac{\hbar^2}{2r^2}\left(V_4''-\frac{V_4'}r\right),
\]
and equation~\eqref{eq:feq_29} yields
\begin{equation}
\label{eq:c3.step.IIBb.V2}
\begin{aligned}
V_2={}&\frac{\hbar}{8}\left[
 (2V_4-1)r\left(\frac{V_1}{V_4}\right)'
 +(rV_4'+6V_4-1)\frac{V_1}{V_4}\right]-\frac{\hbar^2}{8}\left[
 V_4''-\frac{(V_4')^2}{2V_4-1}-\frac{2V_4'}r\right].
\end{aligned}
\end{equation}
The remaining first-order equations are identically satisfied.
For the zeroth-order terms, equation~\eqref{eq:feq_11} reduces to
\[
 \frac{\rm d}{{\rm d}r}\left[
 V_0+V_2-\frac{r^2V_1^2}{8V_4^2}
 -\frac{\hbar}{2}(rV_1'+3V_1)
 -\frac{\hbar V_1}{4V_4}
 -\frac{\hbar^2V_4'}{2r}\right]=0.
\]
Using~\eqref{eq:c3.step.IIBb.V2}, we obtain
\begin{align*}
V_0={}&d_3+\frac{r^2V_1^2}{8V_4^2}+\frac{\hbar}{8}\left[
 (2V_4+1)\left(r\left(\frac{V_1}{V_4}\right)'
 +\frac{3V_1}{V_4}\right)
 +3rV_4'\frac{V_1}{V_4}\right]\\
&{}+\frac{\hbar^2}{8}\left[
 V_4''-\frac{(V_4')^2}{2V_4-1}+\frac{2V_4'}r\right],\label{eq:c3.step.IIBb.V0}
\end{align*}
where $d_3$ is a~real constant. The remaining zeroth-order determining
equations are then identically satisfied.

Substituting the coefficients into $Y_p$, the remaining terms give
the integral of motion
\begin{equation}
\label{eq:c3.step.IIBb.integral}
Y_p=\left.Y_P^{32}\right|_{f_3\mapsto f_8+r^2f_{10}}.
\end{equation}
The potentials and this integral are obtained from
\eqref{eq:c3.step.IIA.V4}, \eqref{eq:c3.step.IIA.V3},
\eqref{eq:c3.step.IIA.V2}, \eqref{eq:c3.step.IIA.V0} and
\eqref{eq:c3.step.IIA.Y38}, respectively, by replacing $f_3$ in
those expressions by $f_8+r^2f_{10}$.
Hence this subcase gives no new Hamiltonian family or integral expression.

{\bf Subcase 2.} $V_4=\alpha_1$, $\alpha_1\neq0$.
From \eqref{eq:2.13}, we obtain
$f'_3 + 2 \alpha_1 f'_4 = 0$.
Hence
\begin{equation}
	\label{eq:2.147}
	f_3 = d_1 - 2 \alpha_1 f_4,
\end{equation}
where $d_1$ is a~constant. In what follows, all $d_i$ are constants.

Equation~\eqref{eq:2.6} gives
$
	\bigl(1 - 4 \alpha_1^2 \bigr) f'_4 = 0$.
Therefore,
$f_4 = d_2 $ or $ \alpha_1 = \pm \frac{1}{2}$.

{\bf S1.} $f_4=d_2$, $\alpha_1\neq\pm\frac12$.
Equations~\eqref{eq:2.7} and~\eqref{eq:2.17} reduce to
$
	2 \alpha_1 f_7 - d_2 V_1 = 0$, $ f_1 + f_2 = 0$.
For $d_2\neq0$, we obtain~$\smash{V_1 = \frac{2 \alpha_1 f_7}{d_2}}$, $ f_1 = -f_2$.
If $d_2=0$ and $d_1\neq0$, the second-order equations instead give~${V_1=f_7=0}$; the remaining equations yield the constant-potential Case~I.A below. If~${d_1=d_2=0}$, a~nonzero component with coefficient
$f_8+r^2f_{10}$ gives the family in II below, whose
formulas also include $d_1=0$. If that component also vanishes, \eqref{eq:2.144} gives $Y_p=0$.

Equation~\eqref{eq:2.2} becomes
$
	(d_1 + d_2 - 2 \alpha_1 d_2) f_7 = 0$.
Therefore, either
$f_7 = 0 $ or $d_2 = \frac{d_1}{2 \alpha_1 - 1}$.

We first consider the case $f_7 = 0$.

{\bf I.} $f_7=0$.
From the determining equation~\eqref{eq:2.5}, we have
$2 r f_{10} + r^2 f'_{10} + f'_8 = 0$.
Hence
$f_{10} = \frac{d_3}{r^2} - \frac{f_8}{r^2}$,
where $d_3$ is a~constant. Substituting the above relations into~\eqref{eq:2.8} gives $f_2 = 0$. Thus, all determining equations coming from the second-order terms are satisfied.

We now consider the determining equations corresponding to the first- and zeroth-order terms, using the relations obtained above. From~\eqref{eq:feq_29} and~\eqref{eq:feq_98}, we obtain
\begin{equation}
	\label{eq:2.156}
	d_3 V_2 = 0, \qquad d_3 V_3 = 0.
\end{equation}
Therefore, either
$d_3 = 0$ or $ V_2 = V_3 = 0$.

{\bf A.} $V_2=V_3=0$, $(d_3\neq0)$.
Substituting this condition into~\eqref{eq:feq_11}, we obtain $V_0=\alpha_2$. Hence all determining equations corresponding to the first- and zeroth-order terms are satisfied. This leaves three arbitrary constants. The integrals of motion are given by $Y_P^{20}$, $Y_P^{21}$, and $Y_P^{22}$ in~\eqref{eq:2.140}, \eqref{eq:2.158}, and~\eqref{eq:2.159}.

{\bf B.} $d_3=0$.
No restriction on $V_2$ or $V_3$ follows from~\eqref{eq:2.156} in this case. We therefore substitute $d_3=0$ into the remaining determining equations. On an interval where $V_2\neq0$, a~nonzero solution satisfies
\begin{equation}
\label{eq:c3.constant_even}
 d_1=(1+2\alpha_1)d_2,\qquad
 V_3=0,\qquad V_2=d_9-V_0.
\end{equation}
Here $V_0(r)$ remains arbitrary, and $V_1=V_5=0$. All determining equations are satisfied, and the integral is
$Y_p=2d_2Y_P^{26}$,
where $Y_P^{26}$ is given in~\eqref{eq:2.199}. This family is valid for every constant $\alpha_1\neq0$. If $V_2=0$ identically and the integral is nonzero, the remaining equations give $V_3=0$ and $V_0=\alpha_2$, already included in Case~I.A. Otherwise, only the zero operator remains.

{\bf II.}
$d_2=\frac{d_1}{2\alpha_1-1}$,
$\bigl(\alpha_1\neq\frac12\bigr)$.
From~\eqref{eq:2.5}, we obtain
$2rf_{10}+r^2f_{10}'+f_8'=0$.
Hence,
\smash{$f_{10}=\frac{d_4-f_8}{r^2}$},
where $d_4$ is a~real constant. Then equation~\eqref{eq:2.17} reduces to
$-2\alpha_1(f_1+f_2)=0$.
Since $\alpha_1\neq0$, we obtain~${f_1=-f_2}$. Therefore, equation~\eqref{eq:2.15} becomes
$
 d_4V_1+2\alpha_1f_2=0$.
Therefore,
\smash{$
 f_2=-\frac{d_4V_1}{2\alpha_1}$},
\smash{$
 f_1=\frac{d_4V_1}{2\alpha_1}$}.
Equation~\eqref{eq:2.7}, together with
$d_2=d_1/(2\alpha_1-1)$, gives
$
 2\alpha_1(2\alpha_1-1)f_7-d_1V_1=0$.
Thus,
\smash{$
 f_7=\frac{d_1V_1}{2\alpha_1(2\alpha_1-1)}$}.
These expressions satisfy
$
 d_1f_2+(2\alpha_1-1)d_4f_7=0$.
Together with
\[
 f_3=\frac{d_1}{1-2\alpha_1},\qquad
 f_4=\frac{d_1}{2\alpha_1-1},\qquad
 f_5=f_6=0,\qquad
 f_9=f_{10},\qquad f_{11}=-f_{10},
\]
all determining equations coming from the second-order
terms are satisfied.

We now consider the first-order determining equations.
Substituting the preceding expressions into
\eqref{eq:feq_29} and~\eqref{eq:feq_93}, respectively,
we obtain
\[
\begin{aligned}
 d_4\bigl[
 8\alpha_1V_2
 -\hbar\bigl((6\alpha_1-1)V_1
 +(2\alpha_1-1)rV_1'\bigr)
 \bigr]={}&0,\\
 d_1\bigl[
 8\alpha_1V_2
 -\hbar\bigl((6\alpha_1-1)V_1
 +(2\alpha_1-1)rV_1'\bigr)
 \bigr]={}&0.
\end{aligned}
\]
If $d_1=d_4=0$, then $f_1=\cdots=f_7=0$, and the
remaining terms involving $f_8(r)$ give the zero operator
by~\eqref{eq:pseudo_structure_identity}.
We therefore exclude this zero solution.
Since $d_1$ and $d_4$ do not vanish simultaneously,
\[
 V_2=\frac{\hbar}{8\alpha_1}
 \bigl[(6\alpha_1-1)V_1+(2\alpha_1-1)rV_1'\bigr].
\]

Using this expression in~\eqref{eq:feq_65} and
\eqref{eq:feq_10}, respectively, we obtain
\[
 d_4\bigl[4\alpha_1rV_3-rV_1^2-\hbar V_1'\bigr]=0,\qquad
 d_1\bigl[4\alpha_1rV_3-rV_1^2-\hbar V_1'\bigr]=0.
\]
Consequently,
\smash{$
 V_3=\frac{V_1^2}{4\alpha_1}
 +\frac{\hbar V_1'}{4\alpha_1r}$}.
The remaining first-order determining equations are then
identically satisfied.

We next consider the zeroth-order determining equations.
Substituting the preceding expressions into
\eqref{eq:feq_11} and~\eqref{eq:feq_95}, respectively,
gives
\[
\begin{aligned}
 d_4\frac{\rm d}{{\rm d}r}\left[
 V_0-\frac{r^2V_1^2}{8\alpha_1^2}
 -\frac{(1+2\alpha_1)\hbar}{8\alpha_1}
 (3V_1+rV_1')
 \right]={}&0,\\
 d_1\frac{\rm d}{{\rm d}r}\left[
 V_0-\frac{r^2V_1^2}{8\alpha_1^2}
 -\frac{(1+2\alpha_1)\hbar}{8\alpha_1}
 (3V_1+rV_1')
 \right]={}&0.
\end{aligned}
\]
Since $d_1$ and $d_4$ do not vanish simultaneously,
we obtain
\[
 V_0=d_6+\frac{r^2V_1^2}{8\alpha_1^2}
 +\frac{(1+2\alpha_1)\hbar}{8\alpha_1}
 (3V_1+rV_1'),
\]
where $d_6$ is a~real constant.

All determining equations are satisfied, and the potential
$V_1(r)$ remains arbitrary. These potentials give
system~(11) of
Theorem~\ref{thm:main}.
Substituting the coefficient functions into the symmetrized
operator and using~\eqref{eq:pseudo_structure_identity},
we obtain
\smash{$
 Y_p=d_1Y_P^{23}+d_4Y_P^{24}$},
where $Y_P^{23}$ and $Y_P^{24}$ are given in
\eqref{eq:2.170} and~\eqref{eq:2.171}, respectively.
The arbitrary function $f_8(r)$ cancels from the operator.
Since~$d_1$ and $d_4$ are independent constants that do not
appear in the potentials, their coefficients give the two
corresponding pseudo-scalar integrals of motion.

{\bf S2.} $\alpha_1=\frac{1}{2}$.
From the determining equation~\eqref{eq:2.17}, we obtain
$f_1+f_2=0$,
and hence $f_1=-f_2$. Equation~\eqref{eq:2.2} gives
$d_1 V_1 = 0$.
Therefore, $d_1 = 0$ or $V_1 = 0$.

{\bf I.} $d_1=0$.
Equation~\eqref{eq:2.24} gives
$ 2rf_7-2rf_4V_1+\hbar f_4'=0$.
Hence,
$
 f_7=V_1f_4-\frac{\hbar f_4'}{2r}$.
The determining equation~\eqref{eq:2.20} becomes
\[
 \hbar\bigl(f_8'+r^2f_{10}'+2rf_{10}\bigr)
 -2rV_1\bigl(f_8+r^2f_{10}\bigr)-2rf_2=0.
\]
Therefore,
\[
 f_2=\frac{\hbar}{2r}\bigl(f_8'+r^2f_{10}'+2rf_{10}\bigr)
 -V_1\bigl(f_8+r^2f_{10}\bigr),\qquad f_1=-f_2.
\]
Together with
$
 f_3=-f_4$, $ f_5=f_6=0$, $
 f_9=f_{10}$, $ f_{11}=-f_{10}$,
these expressions satisfy all the determining equations coming from
the second-order terms.

We now consider the first-order equations. From~\eqref{eq:feq_93}
and~\eqref{eq:feq_29}, respectively, we obtain
\begin{equation}
\label{eq:c3.S2I.V2_conditions}
 f_4(2V_2-\hbar V_1)=0,\qquad
 \bigl(f_8+r^2f_{10}\bigr)(2V_2-\hbar V_1)=0.
\end{equation}
If $f_4=0$ and $f_8+r^2f_{10}=0$ identically, then $Y_p=0$.
We therefore consider intervals on which
\begin{equation}
\label{eq:c3.S2I.nonzero}
 f_4^2+\bigl(f_8+r^2f_{10}\bigr)^2\neq0.
\end{equation}
Equations~\eqref{eq:c3.S2I.V2_conditions} then give
\smash{$
 V_2=\frac{\hbar V_1}{2}$}.
Substituting this expression into~\eqref{eq:feq_8} and
\eqref{eq:feq_98}, we obtain
\[
 f_4\bigl(2rV_3-rV_1^2-\hbar V_1'\bigr)=0,\qquad
 \bigl(f_8+r^2f_{10}\bigr)\bigl(2rV_3-rV_1^2-\hbar V_1'\bigr)=0.
\]
Hence,
$
 V_3=\frac{V_1^2}{2}+\frac{\hbar V_1'}{2r}$.
The remaining first-order determining equations are then identically
satisfied.

We next consider the zeroth-order equations. Using the preceding
relations, \eqref{eq:feq_95} and~\eqref{eq:feq_11} reduce to
\begin{gather}
 f_4\left[
 V_0'-rV_1^2-r^2V_1V_1'
 -\frac{\hbar}{2}(rV_1''+4V_1')
 \right]=0,\nonumber\\
 \bigl(f_8+r^2f_{10}\bigr)\left[
 V_0'-rV_1^2-r^2V_1V_1'
 -\frac{\hbar}{2}(rV_1''+4V_1')
 \right]=0.\label{eq:c3.S2I.V0_conditions}
\end{gather}
By~\eqref{eq:c3.S2I.nonzero}, we have
\[
 \frac{\rm d}{{\rm d}r}\left[
 V_0-\frac{r^2V_1^2}{2}
 -\frac{\hbar}{2}(rV_1'+3V_1)
 \right]=0.
\]
Consequently,
\[
 V_0=d_8+\frac{r^2V_1^2}{2}
 +\frac{\hbar}{2}(rV_1'+3V_1),
\]
where $d_8$ is an arbitrary real constant. Substituting these
expressions into the remaining zeroth-order determining equations,
we find that they are identically satisfied. Hence all the determining
equations are satisfied. The functions $V_1(r)$, $f_4(r)$, $f_8(r)$
and $f_{10}(r)$ remain arbitrary, subject to~\eqref{eq:c3.S2I.nonzero}. Substituting the coefficient functions into $Y_p$ and writing
$g=f_8+r^2f_{10}$, we obtain the integral of motion
$Y_P^{25}$ in~\eqref{eq:2.191}.

{\bf II.} $V_1=0$, $(d_1\neq0)$.
From the determining equation~\eqref{eq:2.22}, we obtain
$2 r f_7 + \hbar f'_4 = 0$.
Hence
$f_7 = - \frac{\hbar f'_4}{2 r}$.
Equation~\eqref{eq:2.25} gives
$2 \hbar r f_{10} - 2 r f_2 + \hbar \bigl( r^2 f'_{10} + f'_8 \bigr) = 0$.
Therefore,
\[
	f_2 = \frac{\hbar \bigl( 2 r f_{10} + r^2 f'_{10} + f'_8 \bigr)}{2 r}.
\]
Thus, all determining equations coming from the second-order terms are satisfied. We now consider the first- and zeroth-order equations, after eliminating redundant relations.

The remaining determining equations give
\[
 (d_1-2f_4)V_2=0,\qquad d_1V_3=0,\qquad
 d_1(V_0'+V_2')=0.
\]
Since $d_1\neq0$, we have $V_3=0$. On an interval where $V_2\neq0$, the first equation gives
$f_4 = \frac{d_1}{2}$.
The case in which $V_2$ vanishes identically is considered below.

The remaining first-order equation becomes
$\bigl(r^2f_{10}+f_8\bigr)V_2=0$.
Thus,
$f_{10} = -\frac{f_8}{r^2}$.
Finally, \eqref{eq:feq_4} gives
$V_2' + V_0' = 0$.
Therefore,
$V_2 = d_9 - V_0$,
where $V_0$ remains arbitrary. All first- and zeroth-order determining equations are satisfied. The integral of motion corresponding to $d_1$ is~$Y_P^{26}$ in~\eqref{eq:2.199}.
This is the $\alpha_1=1/2$ member of the family in~\eqref{eq:c3.constant_even}. If $V_2=0$ identically, then~${V_0=d_9}$ and the factors $d_1-2f_4$ and $f_8+r^2f_{10}$ are not forced to vanish. For this choice, the integral expression~\eqref{eq:2.191} for $Y_P^{25}$ also applies with $V_1=0$ and $d_8=d_9$.

{\bf S3.} $\alpha_1=-\frac12$.
From the determining equation~\eqref{eq:2.17}, we obtain
$
 f_1+f_2=0$,
and hence $f_1=-f_2$. Equation~\eqref{eq:2.12} gives
$
 2rf_{10}+r^2f_{10}'+f_8'=0$.
Therefore,
\[
 f_{10}=\frac{d_{10}-f_8}{r^2},
\]
where $d_{10}$ is a~real constant.

From~\eqref{eq:2.25}, we obtain
$
 f_2-d_{10}V_1=0$.
Thus,
$
 f_2=d_{10}V_1$, $ f_1=-d_{10}V_1$.
Using~\eqref{eq:2.147}, equation~\eqref{eq:2.2} becomes
\begin{equation}
\label{eq:c3.S3.f4_equation}
 \hbar f_4'=rV_1(d_1+2f_4).
\end{equation}
Substituting this relation into~\eqref{eq:2.24}, we obtain
$
 d_1V_1-2f_7=0$,
which gives
$f_7=\frac{d_1V_1}{2}$.
Together with $f_3=d_1+f_4$, $f_5=f_6=0$,
$f_9=f_{10}$ and $f_{11}=-f_{10}$, these expressions
satisfy all determining equations coming from the
second-order terms.

We now consider the first-order determining equations.
From~\eqref{eq:feq_29} and~\eqref{eq:feq_93}, respectively,
we obtain
\begin{equation}
\label{eq:c3.S3.V2_conditions}
\begin{aligned}
 d_{10}\bigl[\hbar(2V_1+rV_1')-2V_2\bigr]=0,\qquad
 d_1\bigl[\hbar(2V_1+rV_1')-2V_2\bigr]=0.
\end{aligned}
\end{equation}
Equation~\eqref{eq:feq_10} gives
$ f_4\bigl[\hbar V_1'+rV_1^2+2rV_3\bigr]=0$.
Using this relation and~\eqref{eq:c3.S3.V2_conditions}
in~\eqref{eq:feq_2} and~\eqref{eq:feq_65}, we obtain
\[
 d_1\bigl[\hbar V_1'+rV_1^2+2rV_3\bigr]=0,\qquad
 d_{10}\bigl[\hbar V_1'+rV_1^2+2rV_3\bigr]=0.
\]
If $d_1=d_{10}=0$ and $f_4$ vanishes identically,
the remaining arbitrary function $f_8(r)$ represents
the zero operator by~\eqref{eq:pseudo_structure_identity}.
We therefore exclude this zero solution.
Otherwise, the preceding equations imply
\[
 V_3=-\frac{V_1^2}{2}-\frac{\hbar V_1'}{2r}.
\]
All first-order determining equations are then satisfied
subject to~\eqref{eq:c3.S3.V2_conditions}.

We next consider the zeroth-order determining equations.
Substituting the preceding expressions into
\eqref{eq:feq_4}, \eqref{eq:feq_11} and~\eqref{eq:feq_95},
and using~\eqref{eq:c3.S3.V2_conditions} and their
derivatives, we obtain, for a~nonzero integral,
\[
 \frac{\rm d}{{\rm d}r}\left[
 V_0+V_2-\frac{r^2V_1^2}{2}
 -\frac{\hbar}{2}(2V_1+rV_1')
 \right]=0.
\]
Hence,
\begin{equation}
\label{eq:c3.S3.V0}
 V_0=d_{11}-V_2+\frac{r^2V_1^2}{2}
 +\frac{\hbar}{2}(2V_1+rV_1'),
\end{equation}
where $d_{11}$ is a~real constant.
The remaining zeroth-order determining equations are
identically satisfied subject to
\eqref{eq:c3.S3.V2_conditions}.

If
\begin{equation}
\label{eq:c3.S3.V2}
 V_2=\frac{\hbar}{2}(2V_1+rV_1'),
\end{equation}
then~\eqref{eq:c3.S3.V2_conditions} is identically
satisfied, and~\eqref{eq:c3.S3.V0} becomes
\smash{$
 V_0=d_{11}+\frac{r^2V_1^2}{2}$}.
Thus, all determining equations are satisfied, and
$V_1(r)$ remains arbitrary. The potentials coincide
with those of system~(14) of
Theorem~\ref{thm:main}.

The constant $d_{10}$ does not appear in the potentials.
The corresponding integral of motion is~$Y_P^{27}$ in
\eqref{eq:2.210}. After subtracting $d_{10}Y_P^{27}$
from $Y_p$, the remaining terms give $Y_P^{28}$ in
\eqref{eq:2.211}, where~$f_4$ satisfies
\eqref{eq:c3.S3.f4_equation}.
The integrals satisfy
\[
 \big[(\vec\sigma_1,\vec\sigma_2),Y_P^{28}\big]
 =2{\rm i}d_1Y_P^{27}.
\]

This Hamiltonian family is also the $\alpha_1=-1/2$
subfamily of system~(11), with
$d_6=d_{11}$.

If~\eqref{eq:c3.S3.V2} does not hold identically on the
radial interval, equations~\eqref{eq:c3.S3.V2_conditions}
instead imply
$
 d_1=d_{10}=0$.
Consequently,
$
 f_1=f_2=f_7=0$, $ f_3=f_4$,
and~\eqref{eq:c3.S3.f4_equation} reduces to
$
 \hbar f_4'=2rV_1f_4$.
The potentials are
\begin{gather*}
 V_0=d_{11}-V_2+\frac{r^2V_1^2}{2}
 +\frac{\hbar}{2}(2V_1+rV_1'),\\
 V_3=-\frac{V_1^2}{2}-\frac{\hbar V_1'}{2r},
 \qquad V_4=-\frac12,\qquad V_5=0.
\end{gather*}
Here $V_1(r)$ and $V_2(r)$ remain arbitrary, subject to
$V_2\not\equiv\hbar(2V_1+rV_1')/2$.
All determining equations are satisfied, and the
corresponding integral of motion is
\[
 f_4(\vec\sigma_1+\vec\sigma_2,\vec p)
 -\frac{{\rm i}\hbar f_4'}{2r}
 (\vec\sigma_1+\vec\sigma_2,\vec x),
\]
where $f_4$ is a~nonzero solution of the preceding
coefficient equation. This expression coincides with
$Y_P^{28}$ at $d_1=0$.

The potentials and this integral are contained in
system~(16) of Theorem~\ref{thm:main}: setting $d_1=0$,
replacing $f_3$ by $f_4$, and identifying $d_4=d_{11}$
in that system gives the above potentials and the
integral $Y_P^{31}$ in~\eqref{eq:c3.step.IBa.Y37}.

{\bf Case 4. $\boldsymbol{V_{4}\neq}0$, $\boldsymbol{V_{5} \neq 0}$.}
From the determining equations corresponding to the third-order terms, we obtain
\[
	f_3 = f_4 = f_5 = f_6 = 0, \qquad
	f_8 = -r^2 f_{10}, \qquad
	f_9 = f_{10}, \qquad
	f_{11} = -f_{10}.
\]
We next consider the equations at the quadratic level. Substituting the above relations into \eqref{eq:2.5}, \eqref{eq:2.6}, and~\eqref{eq:2.11} gives $
	f_1=f_2=f_7=0$.
The remaining arbitrary function $f_{10}(r)$ represents the zero operator, so that $Y_p=0$. Therefore all first- and zeroth-order determining equations are identically satisfied, and this case does not admit any nontrivial pseudo-scalar integral of the assumed form.

\section{Symmetry algebra}\label{sec3}
In this section, we study the algebraic structures generated by the pseudo-scalar integrals of motion for selected Hamiltonians in our classification. The system numbers below refer to the entries of Theorem~\ref{thm:main}.

Throughout this section, we include the Hamiltonian $H$, the rotational generators $J_i$ (and hence \smash{$\vec{J}^{\, 2}$}), and the spin scalar $K=(\vec{\sigma}_1,\vec{\sigma}_2)$
as trivial integrals of motion. All identities in this section are understood as formal differential-operator identities on radial intervals where the coefficients are smooth.

\subsection{Symmetry algebra of system~(1)}

We consider the first family of potentials given in Theorem~\ref{thm:main}
\[
 V_0=V_0(r), \qquad V_1=\frac{\hbar}{r^2}, \qquad V_2=\frac{\hbar^2}{4r^2}, \qquad V_3=V_4=V_5=0.
\]
With this choice the Hamiltonian takes the form
\begin{equation}
 \label{eq:alg.2}
 H=-\frac{\hbar^2}{2}\Delta+V_0(r)+\frac{\hbar}{2r^2}\bigl(\vec{\sigma}_1+\vec{\sigma}_2,\vec{L}\bigr)+\frac{\hbar^2}{4r^2}(\vec{\sigma}_1,\vec{\sigma}_2).
\end{equation}
This system admits seven pseudo-scalar integrals of motion $Y_P^1, \dots, Y_P^7$ given by~\eqref{eq:2.41}--\eqref{eq:2.47}. Among them, the integrals $Y_P^5$, $Y_P^6$ and $Y_P^7$ contain no derivatives
\[
 Y_P^5=\frac{1}{r}(\vec{x}, \vec{\sigma}_1\times\vec{\sigma}_2),\qquad
 Y_P^6=\frac{1}{r}(\vec{\sigma}_1,\vec{x}),\qquad
 Y_P^7=\frac{1}{r}(\vec{\sigma}_2-\vec{\sigma}_1,\vec{x}).
\]

\begin{Proposition}\label{prop:algebra}
The operators $K$, $Y_P^5$, $Y_P^6$ and $Y_P^7$ generate a~polynomial symmetry algebra whose non-trivial commutation relations are
\begin{gather*}
 \bigl[K,Y_P^6\bigr]=2{\rm i}Y_P^5, \qquad \bigl[K,Y_P^5\bigr]=4{\rm i} Y_P^7, \qquad \bigl[K,Y_P^7\bigr]=-4{\rm i} Y_P^5,\qquad
 \bigl[Y_P^6,Y_P^7\bigr]=0,\\
 \bigl[Y_P^6,Y_P^5\bigr]=2{\rm i} \left(K-1+\frac{\bigl(Y_P^5\bigr)^2}{2}\right),\qquad
\bigl[Y_P^7,Y_P^5\bigr]=-4{\rm i} \left(K-1+\frac{\bigl(Y_P^5\bigr)^2}{2}\right).
\end{gather*}
Moreover, the following identities hold
\begin{equation}\label{eq:alg.8}
\bigl(Y_P^6\bigr)^2=1,\qquad \bigl(Y_P^7\bigr)^2=\bigl(Y_P^5\bigr)^2,
\end{equation}
and the combination
\begin{equation}\label{eq:alg.9}
Y_P^7+2 Y_P^6=\frac{1}{r} (\vec{\sigma}_1+\vec{\sigma}_2,\vec{x})=\frac{2}{r}\bigl(\vec{S},\vec{x}\bigr)
\end{equation}
commutes with $K$, $Y_P^5$, $Y_P^6$ and $Y_P^7$.
\end{Proposition}

\begin{proof}
The relations follow from the Pauli algebra identities
\begin{gather*}
(\sigma_a)_j(\sigma_a)_k = \delta_{jk}+{\rm i}\varepsilon_{jk\ell}(\sigma_a)_\ell, \qquad a=1,2,\\
[K,(\sigma_1)_i]=2{\rm i}(\vec{\sigma}_1\times\vec{\sigma}_2)_i,\qquad
[K,(\sigma_2)_i]=-2{\rm i}(\vec{\sigma}_1\times\vec{\sigma}_2)_i,
\end{gather*}
together with the fact that the operators depend only on $\hat{x}=\vec{x}/r$. The identities in~\eqref{eq:alg.8} follow from $(\vec{\sigma}_a,\hat{x})^2=1$, and~\eqref{eq:alg.9} is verified by direct computation.
\end{proof}

{\bf Extension of the symmetry algebra.}
The algebra generated in Proposition~\ref{prop:algebra} is not maximal. The system~\eqref{eq:alg.2} admits additional scalar integrals of motion~\cite{TuncerYurdusen2025}
\smash{$
X=\bigl(\vec{\sigma}_2,\vec{L}\bigr)$}, \smash{$
Y=\bigl(\vec{\sigma}_1,\vec{L}\bigr)$},
whose commutator yields
$
[Y,X]={\rm i}\hbar Z$, \smash{$ Z=\bigl(\vec{\sigma}_1\times\vec{\sigma}_2,\vec{L}\bigr)$}.

We therefore consider the enlarged set
\[
\mathcal{G}=\{K, Y_P^5, Y_P^6, Y_P^7, X, Y, Z\}.
\]
The following proposition directly follows from the Pauli algebra together with
\[
[L_i,x_j]={\rm i}\hbar \varepsilon_{ijk}x_k,
\qquad
\bigl(\vec{\sigma}_a,\vec{L}\bigr)=\varepsilon_{ijk}(\sigma_a)_i x_j p_k.
\]
\begin{Proposition}
The operators in $\mathcal{G}$, together with the central elements $H$ and $\vec{J}^{ 2}$, generate a~finitely generated polynomial symmetry algebra.
In addition to the commutation relations in Proposition~{\rm\ref{prop:algebra}}, the
non-trivial commutation relations involving $X$, $Y$ and $Z$ are
\begin{gather*}
 [K,Y]=2{\rm i} Z,\qquad [K,X]=-2{\rm i} Z,\qquad
 [K,Z]=4{\rm i} (X-Y),\qquad
 [Y,X]={\rm i}\hbar Z,\\
 \bigl[Y,Y_P^6\bigr]=2{\rm i} Y_P^3,\qquad
 \bigl[X,Y_P^6\bigr]=-{\rm i}\hbar Y_P^5,\qquad
 \bigl[Y,Y_P^7\bigr]={\rm i}\hbar Y_P^5-2{\rm i} Y_P^3,\\
 \bigl[X,Y_P^7\bigr]=-2{\rm i} Y_P^2+{\rm i}\hbar Y_P^5,\qquad
 \bigl[Y,Y_P^5\bigr]
=
2{\rm i} Y_P^1-2{\rm i}\hbar \bigl(Y_P^6+Y_P^7\bigr),\\
 \bigl[X,Y_P^5\bigr]
=
-2{\rm i} \bigl(Y_P^1+Y_P^4\bigr)+2{\rm i}\hbar Y_P^6.
\end{gather*}
Moreover, the commutators involving $Z$ are polynomial in the generators
\begin{gather*}
[Y,Z]
=
2{\rm i}K \vec{L}^2
-{\rm i}\bigl(YX+XY\bigr)
-2{\rm i}\hbar X,\\
[X,Z]
=
-2{\rm i}K \vec{L}^2
+{\rm i}\bigl(YX+XY\bigr)
+2{\rm i}\hbar Y,
\end{gather*}
where
\[
\vec{L}^2
=
\vec{J}^{\, 2}-\hbar(X+Y)-\frac{\hbar^2}{4}(6+2K).
\]
The remaining commutators involving $Z$ are
\begin{gather*}
\bigl[Z,Y_P^5\bigr]=2{\rm i}\bigl(Y_P^3-Y_P^2\bigr),\qquad
\bigl[Z,Y_P^6\bigr]={\rm i}\bigl\{Y_P^6+Y_P^7,Y+\hbar\bigr\},\\
\bigl[Z,Y_P^7\bigr]=-{\rm i}\bigl\{Y_P^6,X+\hbar\bigr\}
-{\rm i}\bigl\{Y_P^6+Y_P^7,Y+\hbar\bigr\}.
\end{gather*}
The operators $Y_P^1, \dots, Y_P^4$ appearing on the right-hand sides are
polynomial expressions in the generators of $\mathcal G$, as shown in
\eqref{alg.poly1}--\eqref{alg.poly4} below.
\end{Proposition}

\begin{Remark}
The algebra generated by $\mathcal{G}$ is closed. Indeed, commutators produce quadratic expressions such as
\smash{$
\frac{1}{r^2}(\vec{\sigma}_1,\vec{x})(\vec{\sigma}_2,\vec{x})
=Y_P^6 \bigl(Y_P^6+Y_P^7\bigr)$},
and the remaining pseudo-scalar integrals $Y_P^1, \dots, Y_P^4$ can be expressed polynomially in terms of~$\mathcal{G}$
\begin{gather}
\label{alg.poly1}
Y_P^1 = -Y_P^6 X + \frac{{\rm i}\hbar}{2}Y_P^5,\\
Y_P^2 = -{\rm i} \bigl(Y_P^6+Y_P^7\bigr)(X+\hbar),\qquad
Y_P^3 = {\rm i} Y_P^6(Y+\hbar), \\
\label{alg.poly4}
Y_P^4 = Y_P^6 X-\bigl(Y_P^6+Y_P^7\bigr) Y-{\rm i}\hbar Y_P^5.
\end{gather}
Thus no additional generator is required for the polynomial closure of the displayed operators.
\end{Remark}

\subsection{Symmetry algebra of system (2)}

We consider the second family of potentials given in Theorem~\ref{thm:main}
\[
V_0=V_0(r), \qquad V_1=\frac{\hbar}{r^2}, \qquad V_2=V_2(r), \qquad V_3=V_4=V_5=0.
\]
With this choice, the Hamiltonian takes the form
\[
H=-\frac{\hbar^2}{2}\Delta+V_0(r)+\frac{\hbar}{2r^2}\bigl(\vec{\sigma}_1+\vec{\sigma}_2,\vec{L}\bigr)+V_2(r)(\vec{\sigma}_1,\vec{\sigma}_2).
\]

This system admits the following pseudo-scalar integrals of motion of order at most one
\begin{gather*}
Y_P^8 = \frac{1}{r}(\vec{\sigma}_1+\vec{\sigma}_2,\vec{x}),\qquad
Y_P^9 = -\frac{1}{r}\bigl[(\vec{\sigma}_1,\vec{x})\bigl(\vec{\sigma}_2,\vec{L}\bigr)+(\vec{\sigma}_2,\vec{x})\bigl(\vec{\sigma}_1,\vec{L}\bigr)\bigr],\\
Y_P^{10} = \frac{1}{r}(\vec{\sigma}_1+\vec{\sigma}_2,\vec{x})(\vec{x},\vec{p})
- r(\vec{\sigma}_1+\vec{\sigma}_2,\vec{p})
-\frac{{\rm i}\hbar}{r}(\vec{\sigma}_1+\vec{\sigma}_2,\vec{x}).
\end{gather*}

We also introduce the scalar operator
$
A:=\bigl(\vec{\sigma}_1+\vec{\sigma}_2,\vec{L}\bigr)=2\bigl(\vec{S},\vec{L}\bigr)$,
which is an integral of motion for this family of potentials (see~\cite{TuncerYurdusen2025}).

\begin{Proposition}
The operators $A$, $Y_P^8$, $Y_P^{10}$, together with the central elements
$K$ and $\vec{J}^{\, 2}$, generate a~finitely generated quadratic polynomial algebra. More precisely,
\[
[K,A]=\bigl[K,Y_P^8\bigr]=\bigl[K,Y_P^{10}\bigr]=\bigl[K,\vec{J}^{\, 2}\bigr]=0,
\]
and
\[
\bigl[\vec{J}^{\, 2},A\bigr]=\bigl[\vec{J}^{\, 2},Y_P^8\bigr]=\bigl[\vec{J}^{\, 2},Y_P^{10}\bigr]=0.
\]
The operator $Y_P^9$ does not require an additional generator, since it satisfies
\[
Y_P^9 = -Y_P^8 A + {\rm i} Y_P^{10} - \hbar Y_P^8.
\]
The remaining non-trivial commutation relations are
\begin{gather*}
\bigl[A,Y_P^8\bigr]=-2{\rm i} Y_P^{10},\qquad
\bigl[Y_P^8,Y_P^{10}\bigr]
=
-2{\rm i} A - {\rm i}\hbar\!\bigl(6+2K-\bigl(Y_P^8\bigr)^2\bigr),\\
\bigl[A,Y_P^{10}\bigr]
=
{\rm i} \bigl\{Y_P^8,\vec{L}^2\bigr\}
=2{\rm i} Y_P^8 \vec{L}^2 - 2\hbar Y_P^{10},
\end{gather*}
where
\[
\vec{L}^2=\vec{J}^{\, 2}-\hbar A-\frac{\hbar^2}{4}(6+2K).
\]
Hence the algebra closes quadratically on $A$, $Y_P^8$, $Y_P^{10}$, with $K$ and $\vec{J}^{\, 2}$ as central elements.
\end{Proposition}

\begin{proof}
Let
\[
\vec{\Sigma}=2\vec{S}=\vec{\sigma}_1+\vec{\sigma}_2,\qquad
B=Y_P^8=\frac{1}{r}\bigl(\vec{\Sigma},\vec{x}\bigr),\qquad
C=Y_P^{10}.
\]
Then $A=\bigl(\vec{\Sigma},\vec L\bigr)$, and
\[C=\frac{1}{r}\bigl(\vec{\Sigma},\vec{x}\bigr)(\vec{x},\vec p)
-r(\vec{\Sigma},\vec p)
-\frac{{\rm i}\hbar}{r}\bigl(\vec{\Sigma},\vec{x}\bigr).\]
The proof is obtained by direct use of the canonical commutation relations
\[
[x_i,p_j]={\rm i}\hbar\delta_{ij},\qquad
[L_i,x_j]={\rm i}\hbar\varepsilon_{ijk}x_k,\qquad
[L_i,p_j]={\rm i}\hbar\varepsilon_{ijk}p_k,
\]
together with the Pauli algebra. Since $K$ commutes with $\vec{\Sigma}$, $\vec{x}$, $\vec p$, and $\vec L$, we have
$ [K,A]=[K,B]=[K,C]=0$.
Moreover, $A$, $B$, and $C$ are rotational scalars with respect to the total angular momentum
$\vec J=\vec L+\frac{\hbar}{2}\vec{\Sigma}$,
and therefore
\[ \bigl[\vec{J}^{\, 2},A\bigr]=\bigl[\vec{J}^{\, 2},B\bigr]=\bigl[\vec{J}^{\, 2},C\bigr]=0.\]

The relation for $Y_P^9$ follows by expanding
\smash{$ BA=\frac{1}{r}\bigl(\vec{\Sigma},\vec{x}\bigr)\bigl(\vec{\Sigma},\vec L\bigr)$}
and using $\bigl(\vec{x},\vec L\bigr)=0$. This gives $BA=-Y_P^9+{\rm i}C-\hbar B, $
or equivalently
$ Y_P^9=-BA+{\rm i}C-\hbar B$.

The remaining commutators are obtained by the same identities
\[ [A,B]=-2{\rm i}C, \qquad [B,C]=-2{\rm i}A-{\rm i}\hbar\bigl(6+2K-B^2\bigr), \]
and
$[A,C]={\rm i}\{B,\vec{L}^2\}=2{\rm i}B\vec{L}^2-2\hbar C$.
Indeed, $\bigl[\vec{L}^2,B\bigr]=2{\rm i}\hbar C$, so the two expressions are equal.
The anticommutator form makes the anti-Hermiticity explicit.
Finally,
\[
{\vec{J}}^2=\vec{L}^2+\hbar A+\frac{\hbar^2}{4}{\vec{\Sigma}}^2,
\qquad
{\vec{\Sigma}}^2=6+2K,
\]
so that
\[ \vec{L}^2=\vec{J}^{\, 2}-\hbar A-\frac{\hbar^2}{4}(6+2K). \]
Thus $\vec{L}^2$ does not introduce a~new generator, and the algebra closes quadratically on
$A$, $B$, $C$, with $K$ and $\vec{J}^{\, 2}$ central.
\end{proof}
We conclude that the symmetry algebra generated by these operators closes on the finite set
$\big\{H, J_i, \vec{J}^{\, 2}, K, A, Y_P^8, Y_P^{10}\big\}$.

The algebra admits a~polynomial identity involving the generators
\[
A^2+\bigl(Y_P^{10}+{\rm i}\hbar Y_P^8\bigr)^2
=\bigl(6+2K-\bigl(Y_P^8\bigr)^2\bigr)
\left(\vec{J}^{\, 2}-\hbar A-\frac{\hbar^2}{4}(6+2K)\right).
\]

For $V_2=\hbar^2/\bigl(4r^2\bigr)$, this family reduces to system~(1). The displayed
algebra is then a~subalgebra of the enlarged symmetry algebra described
in the preceding subsection.

\subsection{Symmetry algebra of system (11)}

We consider the family of potentials
\begin{gather*}
V_0(r) = d_6
+\frac{r^2V_1^2}{8\alpha_1^2}
+\frac{(1+2\alpha_1)\hbar}{8\alpha_1}\bigl(3V_1+rV_1'\bigr), \\
V_2(r) = \frac{\hbar\bigl((6\alpha_1-1)V_1+(2\alpha_1-1)rV_1'\bigr)}{8\alpha_1},\qquad
V_3(r) = \frac{V_1^2}{4\alpha_1}
+\frac{\hbar V_1'}{4\alpha_1r},\\
V_4(r) = \alpha_1,\qquad V_5(r)=0.
\end{gather*}
Here $V_1(r)$ is arbitrary, and $\alpha_1$, $d_6$ are real constants with $\alpha_1\neq0$ and $1-2\alpha_1\neq0$. This system admits the pseudo-scalar integral
\[
Y_P^{23}
=
\frac{1}{1-2 \alpha_1}
\left[
(\vec{\sigma}_1-\vec{\sigma}_2,\vec p)
-\frac{V_1}{2\alpha_1}
(\vec{x},\vec{\sigma}_1\times\vec{\sigma}_2)
\right],
\]
together with
\[
Y_P^{24}
=
(\vec{\sigma}_1\times\vec{\sigma}_2,\vec p)
+\frac{V_1}{2\alpha_1}(\vec{\sigma}_1-\vec{\sigma}_2,\vec{x}).
\]
Let
$
\Delta=\vec{\sigma}_1-\vec{\sigma}_2$,
$
\Omega=\vec{\sigma}_1\times\vec{\sigma}_2$,
$
a=1-2\alpha_1$,
$
\tilde{g}(r)=-\frac{V_1(r)}{2\alpha_1}$.
Then
\[Y_P^{23}=\frac1a\bigl[(\Delta,\vec p)+\tilde{g}(r)(\vec x,\Omega)\bigr],\qquad Y_P^{24}=(\Omega,\vec p)-\tilde{g}(r)(\Delta,\vec x).\]

\begin{Proposition}
The operators $K$, $Y_P^{23}$ and $Y_P^{24}$ generate a~finitely generated polynomial symmetry algebra. Their non-trivial commutation relations are
\begin{gather}
\bigl[K,Y_P^{23}\bigr]
=
\frac{4{\rm i}}{a}Y_P^{24},\label{eq:alg14.4}\\
\bigl[K,Y_P^{24}\bigr]
=
-4{\rm i}aY_P^{23},\label{eq:alg14.5}\\
\bigl[Y_P^{23},Y_P^{24}\bigr]
=
{\rm i}a(K+1)\bigl(Y_P^{23}\bigr)^2.\label{eq:alg14.6}
\end{gather}
Moreover, the following quadratic identity holds
\begin{equation}
\label{eq:alg14.7}
\bigl(Y_P^{24}\bigr)^2=a^2\bigl(Y_P^{23}\bigr)^2.
\end{equation}
\end{Proposition}

\begin{proof}
The first two commutation relations follow from the Pauli algebra identities $[K,\Delta_i]=4{\rm i}\Omega_i$, $[K,\Omega_i]=-4{\rm i}\Delta_i.$
Indeed,
\begin{align*}
\bigl[K,Y_P^{23}\bigr]={}&
\frac1a
\bigl[
([K,\Delta],\vec p)+\tilde{g}(r)(\vec x,[K,\Omega])
\bigr]=
\frac1a
\bigl[
4{\rm i}(\Omega,\vec p)-4{\rm i}\tilde{g}(r)(\Delta,\vec x)
\bigr]\\
={}&
\frac{4{\rm i}}{a}
\bigl((\Omega,\vec p)-\tilde{g}(r)(\Delta,\vec x)\bigr)
=
\frac{4{\rm i}}{a}Y_P^{24}.
\end{align*}
Similarly,
\begin{align*}
\bigl[K,Y_P^{24}\bigr]
&=
\bigl[
([K,\Omega],\vec p)-\tilde{g}(r)([K,\Delta],\vec x)
\bigr]=
\bigl[
-4{\rm i}(\Delta,\vec p)-4{\rm i}\tilde{g}(r)(\vec x,\Omega)
\bigr]\\
&=
-4{\rm i}
\bigl((\Delta,\vec p)+\tilde{g}(r)(\vec x,\Omega)\bigr)
=
-4{\rm i}aY_P^{23}.
\end{align*}

We now derive the remaining commutator. The spin scalar $K$ satisfies $K^2+2K-3=0$,
or equivalently $(K+1)^2=4$.
Moreover, the Pauli algebra gives the following anticommutators
$\{K+1,\Delta_i\}=0$, $\{K+1,\Omega_i\}=0$.
Consequently, $\{K+1,Y_P^{23}\}=0.$
Using $\bigl[K,Y_P^{23}\bigr]=(4{\rm i}/a)Y_P^{24}$, we obtain
\[\bigl[K,Y_P^{23}\bigr]=\bigl[K+1,Y_P^{23}\bigr]=2(K+1)Y_P^{23}.\]
Thus
\[2(K+1)Y_P^{23}=\frac{4{\rm i}}{a}Y_P^{24},
\]
and hence
\begin{equation}
\label{eq:alg14.9}
Y_P^{24}=-\frac{{\rm i}}{2}a(K+1)Y_P^{23}.
\end{equation}
Therefore,
\begin{align*}
\bigl[Y_P^{23},Y_P^{24}\bigr]
={}&
\left[
Y_P^{23},
-\frac{{\rm i}}{2}a(K+1)Y_P^{23}
\right]=
-\frac{{\rm i}}{2}a
\bigl[
Y_P^{23}(K+1)Y_P^{23}
-
(K+1)\bigl(Y_P^{23}\bigr)^2
\bigr].
\end{align*}
Since $Y_P^{23}$ anticommutes with $K+1$, we have $Y_P^{23}(K+1)Y_P^{23}=-(K+1)\bigl(Y_P^{23}\bigr)^2$.
It follows that
\[\bigl[Y_P^{23},Y_P^{24}\bigr]={\rm i} a(K+1)\bigl(Y_P^{23}\bigr)^2.
\]

Finally, from~\eqref{eq:alg14.9} and $(K+1)^2=4$, together with the anticommutation relation
$\{{K+1},\allowbreak Y_P^{23}\}=0$,
we get
\[
\bigl(Y_P^{24}\bigr)^2
=
\left(-\frac{\rm i}{2}a\right)^2
(K+1)Y_P^{23}(K+1)Y_P^{23}=
-\frac{a^2}{4}
(K+1)Y_P^{23}(K+1)Y_P^{23}.
\]
Since $Y_P^{23}(K+1)=-(K+1)Y_P^{23}$, we obtain
\[
(K+1)Y_P^{23}(K+1)Y_P^{23}
=
-(K+1)^2\bigl(Y_P^{23}\bigr)^2
=
-4\bigl(Y_P^{23}\bigr)^2.
\]
Therefore $\bigl(Y_P^{24}\bigr)^2=a^2\bigl(Y_P^{23}\bigr)^2$.
\end{proof}

 Moreover,
\[
\bigl\{Y_P^{23},Y_P^{24}\bigr\}=0,
\qquad
\bigl[K,\bigl(Y_P^{23}\bigr)^2\bigr]=\bigl[Y_P^{24},\bigl(Y_P^{23}\bigr)^2\bigr]=0.
\]
Thus $\bigl(Y_P^{23}\bigr)^2$ is central in the algebra generated by the displayed operators.

Note that for $\alpha_1=-1/2$, we have $a=2$, and a~direct calculation gives
\smash{$
 \bigl(Y_P^{23}\bigr)^2=H-d_6$}.
Consequently, \eqref{eq:alg14.7} and~\eqref{eq:alg14.6} yield
\[
 \bigl(Y_P^{24}\bigr)^2=4(H-d_6),\qquad
 \bigl[Y_P^{23},Y_P^{24}\bigr]=2{\rm i}(K+1)(H-d_6).
\]
The remaining commutation relations among $K$, $Y_P^{23}$
and $Y_P^{24}$ follow by setting $a=2$ in
\eqref{eq:alg14.4} and~\eqref{eq:alg14.5}.

\subsection{Symmetry algebra of system~(13)}

We next consider the specialization $V_4=1/2$ of the family in system~(13) of Theorem~\ref{thm:main}, namely
\begin{gather}
\label{eq:alg17.1}
V_0=V_0(r),\qquad
V_1=0,\qquad
V_2=d_{9}-V_0(r),\qquad
V_3=0,\qquad
V_4=\frac12,\qquad
V_5=0.
\end{gather}
For this choice the system admits the first-order pseudo-scalar integral
\smash{$
Y_P^{26}=\frac12(\vec\sigma_1+\vec\sigma_2,\vec p)$}.
Setting \smash{$\vec\Sigma=\vec\sigma_1+\vec\sigma_2$} again yields
\smash{$
Y_P^{26}=\bigl(\vec{S},\vec p\bigr)=\frac12(\vec\Sigma,\vec p)$}.

\begin{Proposition}
The operators $H$, $K$, $Y_P^{26}$, together with $\vec{J}^{\, 2}$ and $J_3$, form a~commuting set of integrals of motion:
\begin{equation}
\label{eq:alg17.3}
\bigl[H,\vec{J}^{\, 2}\bigr]=[H,J_3]=\bigl[\vec{J}^{\, 2},J_3\bigr]=0,
\end{equation}
and
\begin{equation}
\label{eq:alg17.4}
\bigl[H,Y_P^{26}\bigr]=[H,K]=\bigl[K,Y_P^{26}\bigr]
=\bigl[\vec{J}^{\, 2},Y_P^{26}\bigr]=\bigl[J_3,Y_P^{26}\bigr]
=\bigl[\vec{J}^{\, 2},K\bigr]=[J_3,K]=0.
\end{equation}
Moreover, $Y_P^{26}$ satisfies the quadratic relation
\begin{equation}
\label{eq:alg17.5}
\bigl(Y_P^{26}\bigr)^2=\frac14(H-d_{9})(K+3).
\end{equation}
\end{Proposition}

\begin{proof}
The commutation relations follow from rotational invariance and the Pauli algebra. Since $K$ commutes with $\vec\Sigma$, we have $\bigl[K,Y_P^{26}\bigr]=0$. Moreover, $Y_P^{26}$ is a~total rotational scalar, hence \smash{$\bigl[\vec{J}^{\, 2},Y_P^{26}\bigr]=\bigl[J_3,Y_P^{26}\bigr]=0$}.

More explicitly, $K\Sigma_i=\Sigma_iK=\Sigma_i$ gives
$(1-K)Y_P^{26}=Y_P^{26}(1-K)=0$ and
\[
\bigl[V_0(r)(1-K),Y_P^{26}\bigr]
=\frac{{\rm i}\hbar V_0'(r)}{2r}(\vec\Sigma,\vec x)(1-K)=0.
\]
The kinetic terms also commute with $Y_P^{26}$, so $\bigl[H,Y_P^{26}\bigr]=0$
for arbitrary $V_0(r)$.

Using
\[
(\vec\Sigma,\vec p)^2
=
(\vec\sigma_1,\vec p)^2+(\vec\sigma_2,\vec p)^2
+2(\vec\sigma_1,\vec p)(\vec\sigma_2,\vec p),
\]
and $
(\vec\sigma_a,\vec p)^2=p^2$, $a=1,2$,
we obtain
\[
\bigl(Y_P^{26}\bigr)^2
=
\frac12p^2+\frac12(\vec\sigma_1,\vec p)(\vec\sigma_2,\vec p).
\]
For the potentials~\eqref{eq:alg17.1}, the Hamiltonian can be written as
\[
H
=
\bigl(Y_P^{26}\bigr)^2+d_{9}K+V_0(r)(1-K).
\]
On the triplet sector $K=1$, this gives
\smash{$
\bigl(Y_P^{26}\bigr)^2=H-d_{9}$}.
On the singlet sector $K=-3$, one has \smash{$\vec\Sigma=0$}, hence $Y_P^{26}=0$. These two cases combine into the operator identity~\eqref{eq:alg17.5}.
\end{proof}

The algebra generated by $H$, $K$, $Y_P^{26}$, $\vec{J}^{\, 2}$, $J_3$ is abelian. However, it contains a~non-trivial quadratic relation involving the Hamiltonian. Thus $Y_P^{26}$ acts as a~square root of the shifted Hamiltonian on the triplet sector, in the formal operator sense of~\eqref{eq:alg17.5}.

The commutation relations~\eqref{eq:alg17.3}--\eqref{eq:alg17.4} remain valid for every constant $V_4=\alpha_1$ in system~(13), whereas~\eqref{eq:alg17.5} uses the specialization $V_4=1/2$.

\section{Conclusion}\label{sec4}

In this paper, we have continued the systematic investigation of superintegrability in the interaction of two non-relativistic spin-$\frac12$ particles in three-dimensional Euclidean space $E^3$. The initial classification of superintegrable systems involving spin interactions was carried out for the case of two particles, only one of which has spin~\cite{DWY, Winternitz.c, wy3, YTW}. In a~recent work~\cite{TuncerYurdusen2025}, this framework was generalized to the case of two particles both of which have spin $1/2$, and the classification problem was solved for first-order scalar integrals of motion. The present paper extends this program to first-order pseudo-scalar integrals of motion.

Starting from the formally Hermitian, rotationally invariant Hamiltonian
of Okubo--Marshak type~\eqref{eq:intro.H}, we constructed the most general
formally Hermitian pseudo-scalar operator of differential order at most
one and imposed the commutation condition $[H,Y_p]=0$.
This led to an overdetermined system of determining equations for the radial potentials in the Hamiltonian and for the coefficient functions appearing in the pseudo-scalar operator. The determining equations were solved systematically by separating the main branches according to the vanishing or non-vanishing of the momentum-dependent potentials $V_4$ and $V_5$, and then analyzing the resulting subcases.

As in the scalar classification, gauge-induced branches play a~special role. When the potentials coincide with the gauge-induced family specified in the introduction, the corresponding branch is excluded from the classification.

The main classification results are summarized in Theorem~\ref{thm:main}. The theorem lists the admissible Hamiltonian families together with their pseudo-scalar integrals and identifies
equivalent parametrizations, overlapping subfamilies and special parameter values admitting additional integrals. Several of these systems admit more than one pseudo-scalar integral.

Some of these Hamiltonians also admit scalar integrals of motion. Comparing our results with the scalar classification in~\cite{TuncerYurdusen2025}, we obtain the following corollary.

\begin{Corollary}\label{cor:scalar_pseudoscalar}
In addition to the pseudo-scalar integrals stated in
Theorem~{\rm \ref{thm:main}}, the Hamiltonians listed there admit the following
scalar integrals, denoted as in~{\rm \cite{TuncerYurdusen2025}}:
\begin{enumerate}\itemsep=0pt
\item[$(1)$] The Hamiltonians in systems~$(1)$--$(5)$, for which
$V_4=0$ and $V_1=\hbar/r^2$, admit
\[
 Y_s^6=
 \frac{(\vec\sigma_1,\vec x)(\vec\sigma_2,\vec x)}{r^2}.
\]

\item[$(2)$] For the Hamiltonians satisfying $V_4=0$ and $V_3=0$,
in particular those in systems~$(1)$ and~$(2)$, an additional scalar integral is
$
 Y_s^7=\bigl(\vec\sigma_1+\vec\sigma_2,\vec L\bigr)$.

\item[$(3)$] For the Hamiltonians satisfying
$
 V_4=0$, $ 4V_2+2r^2V_3=\hbar V_1$,
in particular those in systems~$(1)$ and~$(5)$, the additional scalar integrals are
\[
 Y_s^8=\bigl(\vec\sigma_1-\vec\sigma_2,\vec L\bigr),\qquad
 Y_s^9=(\vec\sigma_1,\vec x)(\vec\sigma_2,\vec p)
 -(\vec\sigma_2,\vec x)(\vec\sigma_1,\vec p).
\]

\item[$(4)$] For the Hamiltonians with $V_4=1/2$, in particular those in
system~$(12)$, an additional scalar integral is
$
 Y_s^{22}=
 \varphi(r)\bigl[1-(\vec\sigma_1,\vec\sigma_2)\bigr]$,
where $\varphi(r)$ is an arbitrary real smooth nonconstant function.

\end{enumerate}
\end{Corollary}

We also investigated the symmetry algebras generated by the pseudo-scalar integrals for selected cases. These algebras are finitely generated polynomial algebras involving the pseudo-scalar integrals, the spin scalar $K=(\vec{\sigma}_1,\vec{\sigma}_2)$, and the rotational invariants. They provide additional information about the structure of the corresponding Hamiltonians.

The present work completes the local classification
of Hamiltonians of the form~\eqref{eq:intro.H} admitting additional
pseudo-scalar integrals of differential order at most one. Several natural problems remain open. One direction is to use the obtained symmetry algebras to study representation theory and exact solvability for particular potentials. Another is to extend the classification to higher-order pseudo-scalar integrals. The first-order vector case with $V_4=0$ has been studied in~\cite{TuncerYurdusenVector2026}. It would also be natural to continue the program by considering higher-order vector integrals and first- and higher-order axial-vector, tensor and pseudo-tensor integrals of motion. These problems are expected to reveal further algebraic structures and new families of quantum superintegrable systems with spin.

\appendix

\section{Determining equations}
\label{appendix-determining-equations}

In this appendix, we list the determining equations used in Section~\ref{section2}. The third-order equations were already given in the main text, since they determine the main branching of the classification. Here we present an equivalent lower-order system for the symmetrized Hamiltonian. The displayed equations include dependent relations; together with the third-order equations, they are equivalent to the vanishing of the full commutator. No coefficient function or potential has been divided out in forming this system.

We first give the determining equations coming from the second-order derivative terms
\begin{gather}
 \hbar \bigl(f_{3}^{\prime} + f_{4}^{\prime}\bigr) V_{4} +r \bigl(\hbar V_{5} + V_{1}\bigr) \bigl(f_{3} + f_{4}\bigr) +\hbar r \bigl(f_{5} + f_{6}\bigr)=0,\label{eq:2.2}
\\
 -\hbar \bigl(V_{4} + 1\bigr) \bigl(f_{10}^{\prime} + f_{11}^{\prime}\bigr) -r \bigl(f_{10} + f_{11}\bigr) V_{1} +r V_{5}\bigl[ 2 \bigl(f_{1} + f_{2}\bigr) +3 \hbar \bigl(f_{10} + f_{11}\bigr)\bigr]=0,\label{eq:2.3}
\\
 \hbar \bigl(V_{4} + 1\bigr) \bigl(f_{5}^{\prime} + f_{6}^{\prime}\bigr) +r \bigl(\hbar V_{5} + V_{1}\bigr) \bigl(f_{5} + f_{6}\bigr)=0,
\\
 \hbar V_{4}^{\prime}\bigl[ r^{2} f_{9} +f_{8}\bigr] +r^{3} V_{1}\bigl[ -f_{9} +f_{10}\bigr] +\hbar V_{4}\bigl[ -r^{2} f_{9}^{\prime} +r^{2} f_{11}^{\prime} -2 f_{8}^{\prime} -6 r f_{9} +r \bigl(f_{10} - f_{11}\bigr)\bigr]\notag\\
\qquad{} -r^{3} V_{5}\bigl[ -\hbar f_{9} +2 f_{1} +2 \hbar f_{10} +\hbar f_{11}\bigr] +\hbar\bigl[ r^{2} f_{10}^{\prime} +f_{8}^{\prime} +2 r f_{9}\bigr]=0,\label{eq:2.5}
\\
 \hbar V_{4}^{\prime}\bigl[ r^{2} f_{5} +f_{3}\bigr] +r^{3} V_{1} f_{6} -\hbar V_{4}\bigl[ r^{2} f_{5}^{\prime} +2 f_{3}^{\prime} +6 r f_{5}\bigr] +r^{3} V_{5}\bigl[ 2 f_{7} +\hbar f_{5}\bigr] \notag\\
\qquad{} -\hbar\bigl[ f_{4}^{\prime} +2 r f_{6}\bigr]=0,\label{eq:2.6}
\\
 V_{4}\bigl[ \hbar f_{3}^{\prime} -2 r f_{7}\bigr] +r V_{1} f_{4} +r V_{5}\bigl[ -2 r^{2} f_{7} +\hbar f_{3}\bigr] +\hbar r f_{6}=0,\label{eq:2.7}
\\
 -r V_{1}\bigl[ r^{2} f_{11} -f_{8}\bigr] -V_{4}\bigl[ \hbar r^{2} f_{10}^{\prime} +\hbar f_{8}^{\prime} +2 r f_{1} +3 \hbar r f_{10} +\hbar r f_{11}\bigr]\notag\\
\qquad{} +r V_{5}\bigl[ 2 \hbar r^{2} f_{11} +2 r^{2} f_{2} +\hbar r^{2} f_{10} -\hbar f_{8}\bigr] +\hbar r\bigl[ f_{11} +f_{9}\bigr]=0,\label{eq:2.8}
\\
 \hbar V_{4}\bigl[ f_{11}^{\prime} +f_{9}^{\prime}\bigr] +r V_{1}\bigl[ -f_{9} +f_{10}\bigr] -r V_{5}\bigl[ -\hbar f_{9} +2 f_{1} +2 \hbar f_{10} +\hbar f_{11}\bigr] \notag\\
\qquad{} +\hbar\bigl[ -f_{9}^{\prime} +f_{10}^{\prime}\bigr]=0,\label{eq:2.9}
\\
 \hbar V_{4} f_{5}^{\prime} +r V_{1} f_{6} +r V_{5}\bigl[ 2 f_{7} +\hbar f_{5}\bigr] +\hbar f_{6}^{\prime}=0,\label{eq:2.10}
\\
 -\hbar V_{4}\bigl[ -f_{9}^{\prime} +f_{10}^{\prime}\bigr] -r V_{1}\bigl[ f_{11} +f_{9}\bigr] +r V_{5}\bigl[ 2 f_{2} +2 \hbar f_{11} +\hbar f_{10} +\hbar f_{9}\bigr] \notag\\
\qquad{} -\hbar\bigl[ f_{11}^{\prime} +f_{9}^{\prime}\bigr]=0,\label{eq:2.11}
\\
 \hbar V_{4}^{\prime}\bigl[ r^{2} f_{9} +f_{8}\bigr] -r^{3} V_{1}\bigl[ f_{11} +f_{9}\bigr] +\hbar V_{4}\bigl[ -r^{2} f_{10}^{\prime} -r^{2} f_{9}^{\prime} -2 f_{8}^{\prime} -6 r f_{9} +r \bigl(f_{10} - f_{11}\bigr)\bigr]\notag\\
 \qquad{} +r^{3} V_{5}\bigl[ 2 f_{2} +2 \hbar f_{11} +\hbar f_{10} +\hbar f_{9}\bigr] -\hbar\bigl[ r^{2} f_{11}^{\prime} -f_{8}^{\prime} -2 r f_{9}\bigr]=0,\label{eq:2.12}
\\
 -\hbar V_{4}^{\prime}\bigl[ r^{2} f_{6} +f_{4}\bigr] -r^{3} V_{1} f_{5} +\hbar V_{4}\bigl[ r^{2} f_{6}^{\prime} +2 f_{4}^{\prime} +6 r f_{6}\bigr] -r^{3} V_{5}\bigl[ -2 f_{7} +\hbar f_{6}\bigr] \notag\\
\qquad{} +\hbar\bigl[ f_{3}^{\prime} +2 r f_{5}\bigr]=0,\label{eq:2.13}
\\
 V_{4}\bigl[ \hbar f_{4}^{\prime} +2 r f_{7}\bigr] +r V_{1} f_{3} +r V_{5}\bigl[ 2 r^{2} f_{7} +\hbar f_{4}\bigr] +\hbar r f_{5}=0,\label{eq:2.14}
\\
 r V_{1}\bigl[ r^{2} f_{10} +f_{8}\bigr] +V_{4}\bigl[ \hbar r^{2} f_{11}^{\prime} -\hbar f_{8}^{\prime} +2 r f_{2} +3 \hbar r f_{11} +\hbar r f_{10}\bigr]\notag\\
 \qquad{} -r V_{5}\bigl[ 2 \hbar r^{2} f_{10} +2 r^{2} f_{1} +\hbar r^{2} f_{11} +\hbar f_{8}\bigr] -\hbar r\bigl[ -f_{9} +f_{10}\bigr]=0,\label{eq:2.15}
\\
 \hbar V_{4} f_{6}^{\prime} +r V_{1} f_{5} +r V_{5}\bigl[ -2 f_{7} +\hbar f_{6}\bigr] +\hbar f_{5}^{\prime}=0,\label{eq:2.16}
\\
 -V_{4}\bigl[ \hbar r \bigl(f_{10}^{\prime} + f_{11}^{\prime}\bigr) +2 \bigl(f_{1} + f_{2}\bigr) +4 \hbar \bigl(f_{10} + f_{11}\bigr)\bigr] +\bigl(\hbar - r^{2} V_{1}\bigr) \bigl(f_{10} + f_{11}\bigr)\notag\\
\qquad{} +r^{2} V_{5}\bigl[ 2 \bigl(f_{1} + f_{2}\bigr) +3 \hbar \bigl(f_{10} + f_{11}\bigr)\bigr]=0,\label{eq:2.17}
\\
 -V_{4}\bigl[ \hbar \bigl(f_{3}^{\prime} - f_{4}^{\prime}\bigr) -4 r f_{7}\bigr] +r \bigl(f_{3} - f_{4}\bigr) V_{1} -r V_{5}\bigl[ -4 r^{2} f_{7} +\hbar \bigl(f_{3} - f_{4}\bigr)\bigr]\notag\\
\qquad{} +\hbar r \bigl(f_{5} - f_{6}\bigr)=0,
\\
 \hbar \bigl(V_{4} - 1\bigr) \bigl(f_{5}^{\prime} - f_{6}^{\prime}\bigr) -r \bigl(f_{5} - f_{6}\bigr) V_{1} +r V_{5}\bigl[ 4 f_{7} +\hbar \bigl(f_{5} - f_{6}\bigr)\bigr]=0,
\\
 -r V_{1}\bigl[ r^{2} \bigl(f_{10} - f_{11}\bigr) +2 f_{8}\bigr]\notag\\
\qquad{} +V_{4}\bigl[ \hbar r^{2} \bigl(f_{10}^{\prime} - f_{11}^{\prime}\bigr) +2 \hbar f_{8}^{\prime} +2 \hbar r \bigl(f_{10} - f_{11}\bigr) +2 r \bigl(f_{1} - f_{2}\bigr)\bigr]\notag\\
\qquad{} +r V_{5}\bigl[ 2 r^{2} \bigl(f_{1} - f_{2}\bigr) +\hbar r^{2} \bigl(f_{10} - f_{11}\bigr) +2 \hbar f_{8}\bigr] +\hbar r\bigl[ -2 f_{9} +\bigl(f_{10} - f_{11}\bigr)\bigr]=0,\label{eq:2.20}
\\
 \hbar \bigl(V_{4} - 1\bigr)\bigl[ -2 f_{9}^{\prime} +\bigl(f_{10}^{\prime} - f_{11}^{\prime}\bigr)\bigr] -r V_{1}\bigl[ -2 f_{9} +\bigl(f_{10} - f_{11}\bigr)\bigr]\notag\\
\qquad{} +r V_{5}\bigl[ -2 \hbar f_{9} +2 \bigl(f_{1} - f_{2}\bigr) +\hbar \bigl(f_{10} - f_{11}\bigr)\bigr]=0,
\\
 \hbar \bigl(r V_{5} + V_{4}^{\prime}\bigr)\bigl[ r^{2} f_{5} +f_{3}\bigr] +r V_{1}\bigl[ r^{2} f_{6} +f_{4}\bigr] -V_{4}\bigl[ \hbar r^{2} f_{5}^{\prime} +\hbar f_{3}^{\prime} +2 r f_{7} +6 \hbar r f_{5}\bigr]\notag\\
 \qquad{} -\hbar\bigl[ f_{4}^{\prime} +r f_{6}\bigr]=0,\label{eq:2.22}
\\
 -V_{4}\bigl[ -\hbar r^{2} f_{9}^{\prime} +\hbar r^{2} \bigl(f_{10}^{\prime} + f_{11}^{\prime}\bigr) -\hbar f_{8}^{\prime} -6 \hbar r f_{9} +2 r f_{1} +4 \hbar r f_{10}\bigr]\notag\\
\qquad{} -r V_{1}\bigl[ -r^{2} f_{9} +r^{2} \bigl(f_{10} + f_{11}\bigr) -f_{8}\bigr] -\hbar V_{4}^{\prime}\bigl[ r^{2} f_{9} +f_{8}\bigr]\notag\\
\qquad{} +r V_{5}\bigl[ -\hbar r^{2} f_{9} +2 r^{2} \bigl(f_{1} + f_{2}\bigr) +3 \hbar r^{2} \bigl(f_{10} + f_{11}\bigr) -\hbar f_{8}\bigr]\notag\\
\qquad{} -\hbar\bigl[ r^{2} f_{10}^{\prime} +f_{8}^{\prime} -r f_{11} +r f_{9}\bigr]=0,
\\
 \hbar \bigl(r V_{5} + V_{4}^{\prime}\bigr)\bigl[ r^{2} f_{6} +f_{4}\bigr] +r V_{1}\bigl[ r^{2} f_{5} +f_{3}\bigr] -V_{4}\bigl[ \hbar r^{2} f_{6}^{\prime} +\hbar f_{4}^{\prime} -2 r f_{7} +6 \hbar r f_{6}\bigr]\notag\\
\qquad{} -\hbar\bigl[ f_{3}^{\prime} +r f_{5}\bigr]=0,\label{eq:2.24}
\\
 V_{4}\bigl[ \hbar r^{2} \bigl(f_{10}^{\prime} + f_{11}^{\prime}\bigr) +\hbar r^{2} f_{9}^{\prime} +\hbar f_{8}^{\prime} +2 r f_{2} +4 \hbar r f_{11} +6 \hbar r f_{9}\bigr]\notag\\
\qquad{} +r V_{1}\bigl[ r^{2} \bigl(f_{10} + f_{11}\bigr) +r^{2} f_{9} +f_{8}\bigr] -\hbar V_{4}^{\prime}\bigl[ r^{2} f_{9} +f_{8}\bigr]\notag\\
\qquad{} -r V_{5}\bigl[ 2 r^{2} \bigl(f_{1} + f_{2}\bigr) +3 \hbar r^{2} \bigl(f_{10} + f_{11}\bigr) +\hbar r^{2} f_{9} +\hbar f_{8}\bigr]\notag\\
\qquad{} -\hbar\bigl[ -r^{2} f_{11}^{\prime} +f_{8}^{\prime} +r f_{10} +r f_{9}\bigr]=0.\label{eq:2.25}
\end{gather}

Now the determining equations coming from the first- and zeroth-order derivative terms are listed below
\begin{gather}
 \hbar^{2} V_{4}\bigl[ r^{3} \bigl(f_{5}^{\prime\prime} + f_{6}^{\prime\prime}\bigr) -r \bigl(f_{3}^{\prime\prime} + f_{4}^{\prime\prime}\bigr) +7 r^{2} \bigl(f_{5}^{\prime} + f_{6}^{\prime}\bigr) +\bigl(f_{3}^{\prime} + f_{4}^{\prime}\bigr)\bigr]\notag\\
 \qquad{}
 +\hbar^{2} V_{4}^{\prime}\bigl[ r^{3} \bigl(f_{5}^{\prime} + f_{6}^{\prime}\bigr) -r \bigl(f_{3}^{\prime} + f_{4}^{\prime}\bigr) +5 r^{2} \bigl(f_{5} + f_{6}\bigr) +\bigl(f_{3} + f_{4}\bigr)\bigr] +4 r^{3} \bigl(f_{3} + f_{4}\bigr) V_{3}\notag\\
 \qquad{}
 +\hbar r^{2} \bigl(\hbar V_{5}^{\prime} + V_{1}^{\prime}\bigr)\bigl[ r^{2} \bigl(f_{5} + f_{6}\bigr) +\bigl(f_{3} + f_{4}\bigr)\bigr] -\hbar^{2} r V_{4}^{\prime\prime}\bigl[ -r^{2} \bigl(f_{5} + f_{6}\bigr) +\bigl(f_{3} + f_{4}\bigr)\bigr]\notag\\
 \qquad{}
 +4 \hbar^{2} r^{2} V_{5}\bigl[ r^{2} \bigl(f_{5}^{\prime} + f_{6}^{\prime}\bigr) +\bigl(f_{3}^{\prime} + f_{4}^{\prime}\bigr) +6 r \bigl(f_{5} + f_{6}\bigr)\bigr]=0,\label{eq:feq_2}
\\
 -\hbar^{2} \bigl(2 V_{4} + 1\bigr)\bigl[ -r^{4} \bigl(f_{5}^{\prime\prime\prime} + f_{6}^{\prime\prime\prime}\bigr) -r^{2} \bigl(f_{3}^{\prime\prime\prime} + f_{4}^{\prime\prime\prime}\bigr) -10 r^{3} \bigl(f_{5}^{\prime\prime} + f_{6}^{\prime\prime}\bigr) -2 r \bigl(f_{3}^{\prime\prime} + f_{4}^{\prime\prime}\bigr) \notag\\
\qquad{}
 -20 r^{2} \bigl(f_{5}^{\prime} + f_{6}^{\prime}\bigr) +2 \bigl(f_{3}^{\prime} + f_{4}^{\prime}\bigr)\bigr] +2 \hbar^{2} V_{4}^{\prime}\bigl[ r^{4} \bigl(f_{5}^{\prime\prime} + f_{6}^{\prime\prime}\bigr) +r^{2} \bigl(f_{3}^{\prime\prime} + f_{4}^{\prime\prime}\bigr) \notag\\
\qquad{}
 +7 r^{3} \bigl(f_{5}^{\prime} + f_{6}^{\prime}\bigr) +r \bigl(f_{3}^{\prime} + f_{4}^{\prime}\bigr) +8 r^{2} \bigl(f_{5} + f_{6}\bigr) +2 \bigl(f_{3} + f_{4}\bigr)\bigr]\notag\\
 \qquad{}
 +2 r^{2} \bigl(- \hbar^{2} V_{4}^{\prime\prime\prime} + 2 r^{2} V_{3}^{\prime} + 2 V_{0}^{\prime} + 2 V_{2}^{\prime}\bigr)\bigl[ r^{2} \bigl(f_{5} + f_{6}\bigr) +\bigl(f_{3} + f_{4}\bigr)\bigr]\notag\\
\qquad{}
 +8 r^{3} V_{3}\bigl[ r^{2} \bigl(f_{5} + f_{6}\bigr) +2 \bigl(f_{3} + f_{4}\bigr)\bigr] -4 \hbar^{2} r \bigl(f_{3} + f_{4}\bigr) V_{4}^{\prime\prime}\notag\\
 \qquad{}
 +2 \hbar r^{2} \bigl(\hbar V_{5} + V_{1}\bigr)\bigl[ r^{2} \bigl(f_{5}^{\prime} + f_{6}^{\prime}\bigr) +\bigl(f_{3}^{\prime} + f_{4}^{\prime}\bigr) +4 r \bigl(f_{5} + f_{6}\bigr)\bigr]=0,\label{eq:feq_4}
\\
 -4 r \bigl(f_{3} - f_{4}\bigr) V_{2} -\hbar V_{4}\bigl[ \hbar r^{3} \bigl(f_{5}^{\prime\prime} - f_{6}^{\prime\prime}\bigr) +\hbar r \bigl(f_{3}^{\prime\prime} + f_{4}^{\prime\prime}\bigr) -8 \hbar r^{2} f_{6}^{\prime} +12 \hbar r^{2} f_{5}^{\prime}\notag\\
 \qquad{}
 +4 r^{2} f_{7}^{\prime} +4 \hbar f_{3}^{\prime} -6 \hbar r f_{6} +12 r f_{7} +30 \hbar r f_{5}\bigr] +r V_{1}\bigl[ \hbar r^{3} f_{6}^{\prime} +\hbar r f_{4}^{\prime} -2 r^{2} f_{7} \notag\\
 \qquad{}
 +6 \hbar r^{2} f_{6} +\hbar \bigl(f_{3} + f_{4}\bigr)\bigr]+\hbar V_{4}^{\prime}\bigl[ \hbar r^{3} f_{6}^{\prime} -\hbar r f_{4}^{\prime} -2 r^{2} f_{7} -\hbar r^{2} f_{5} +5 \hbar r^{2} f_{6} \notag\\
 \qquad{}
 +4 \hbar f_{3}\bigr] +4 r^{3} V_{3} f_{4} +\hbar r^{2} \bigl(\hbar V_{5}^{\prime} + V_{1}^{\prime}\bigr)\bigl[ r^{2} \bigl(f_{5} + f_{6}\bigr) +\bigl(f_{3} + f_{4}\bigr)\bigr] \notag\\
 \qquad{}
 +\hbar^{2} r V_{4}^{\prime\prime}\bigl[ r^{2} \bigl(f_{5} + f_{6}\bigr) +\bigl(f_{3} - f_{4}\bigr)\bigr] -\hbar r V_{5}\bigl[ -4 \hbar r^{3} f_{5}^{\prime} -\hbar r^{3} f_{6}^{\prime} -4 \hbar r f_{3}^{\prime} \notag\\
 \qquad{}
 -\hbar r f_{4}^{\prime} -24 \hbar r^{2} f_{5} -6 \hbar r^{2} f_{6} -6 r^{2} f_{7} -3 \hbar f_{4} +\hbar f_{3}\bigr] \notag\\
 \qquad{}
 -\hbar\bigl[ \hbar r f_{4}^{\prime\prime} +\hbar r^{2} f_{6}^{\prime} +3 \hbar f_{4}^{\prime} -2 r f_{7} +6 \hbar r f_{6}\bigr]=0,\label{eq:feq_8}
\\
 -\hbar^{2} V_{4}\bigl[ -3 r^{3} f_{5}^{\prime\prime} -r^{3} f_{6}^{\prime\prime} -3 r f_{3}^{\prime\prime} +r f_{4}^{\prime\prime} -15 r^{2} f_{5}^{\prime} -7 r^{2} f_{6}^{\prime} -f_{4}^{\prime} +3 f_{3}^{\prime}\bigr]\notag\\
 \qquad{}
 +\hbar V_{4}^{\prime}\bigl[ 2 \hbar r^{3} f_{5}^{\prime} +\hbar r^{3} f_{6}^{\prime} -\hbar r f_{4}^{\prime} +2 \hbar r f_{3}^{\prime} -2 r^{2} f_{7} +4 \hbar r^{2} f_{5} +6 \hbar r^{2} f_{6} +\hbar \bigl(f_{3} + f_{4}\bigr)\bigr]\notag\\
\qquad{}
 +4 r^{3} \bigl(f_{5} - f_{6}\bigr) V_{2} +r^{2} V_{1}\bigl[ \hbar r^{2} f_{6}^{\prime} +\hbar f_{4}^{\prime} -2 r f_{7} -\hbar r f_{5} +9 \hbar r f_{6}\bigr] +4 r^{3} V_{3} f_{4}\notag\\
 \qquad{}
 +\hbar r^{2} V_{1}^{\prime}\bigl[ r^{2} \bigl(f_{5} + f_{6}\bigr) +\bigl(f_{3} + f_{4}\bigr)\bigr] -\hbar^{2} r V_{4}^{\prime\prime}\bigl[ r^{2} \bigl(f_{5} - f_{6}\bigr) +\bigl(f_{3} + f_{4}\bigr)\bigr]\notag\\
 \qquad{}
 +\hbar r^{2} V_{5}\bigl[ -\hbar r^{2} f_{6}^{\prime} +6 \hbar r^{2} f_{5}^{\prime} -\hbar f_{4}^{\prime} +6 \hbar f_{3}^{\prime} -9 \hbar r f_{6} +14 r f_{7} +41 \hbar r f_{5}\bigr]\notag\\
 \qquad{}
 +\hbar^{2} r^{2} V_{5}^{\prime}\bigl[ -r^{2} f_{6} +3 r^{2} f_{5} -f_{4} +3 f_{3}\bigr]\notag\\
 \qquad{}
 -\hbar\bigl[ -2 \hbar r^{3} f_{6}^{\prime\prime} -\hbar r f_{4}^{\prime\prime} -11 \hbar r^{2} f_{6}^{\prime} +2 r^{2} f_{7}^{\prime} +\hbar f_{4}^{\prime}\bigr]=0,\label{eq:feq_10}
\\
 8 r^{3} \bigl(f_{1} - f_{2}\bigr) V_{2} +2 \hbar^{2} V_{4}^{\prime}\bigl[ -\hbar r^{4} f_{9}^{\prime\prime} -\hbar r^{2} f_{8}^{\prime\prime} -7 \hbar r^{3} f_{9}^{\prime} +\hbar r^{3} \bigl(f_{10}^{\prime} - f_{11}^{\prime}\bigr) -\hbar r f_{8}^{\prime}\notag\\
\qquad{}
 -8 \hbar r^{2} f_{9} -r^{2} \bigl(f_{1} - f_{2}\bigr) +\hbar r^{2} \bigl(f_{10} - f_{11}\bigr) -2 \hbar f_{8}\bigr] +\hbar^{3} \bigl(2 V_{4} - 1\bigr)\bigl[ -r^{4} f_{9}^{\prime\prime\prime} \notag\\
\qquad{}
 -r^{2} f_{8}^{\prime\prime\prime} -10 r^{3} f_{9}^{\prime\prime} +r^{3} \bigl(f_{10}^{\prime\prime} - f_{11}^{\prime\prime}\bigr) -2 r f_{8}^{\prime\prime} -20 r^{2} f_{9}^{\prime} +4 r^{2} \bigl(f_{10}^{\prime} - f_{11}^{\prime}\bigr) +2 f_{8}^{\prime}\bigr]\notag\\
 \qquad{}
 -4 \hbar r^{3} V_{3}\bigl[ 2 r^{2} f_{9} +r^{2} \bigl(f_{10} - f_{11}\bigr) +4 f_{8}\bigr] +2 \hbar^{2} r^{2} V_{5}\bigl[ -\hbar r^{2} f_{9}^{\prime} -\hbar f_{8}^{\prime} -4 \hbar r f_{9} \notag\\
 \qquad{}
 +3 r \bigl(f_{1} - f_{2}\bigr) +\hbar r \bigl(f_{10} - f_{11}\bigr)\bigr] -2 \hbar r^{2} \bigl(- \hbar^{2} V_{4}^{\prime\prime\prime} + 2 r^{2} V_{3}^{\prime} - 2 V_{0}^{\prime} + 2 V_{2}^{\prime}\bigr)\bigl[ r^{2} f_{9} +f_{8}\bigr]\notag\\
 \qquad{}
 -2 \hbar r^{2} V_{1}\bigl[ -\hbar r^{2} f_{9}^{\prime} -\hbar f_{8}^{\prime} -4 \hbar r f_{9} +\hbar r \bigl(f_{10} - f_{11}\bigr) +r \bigl(f_{1} - f_{2}\bigr)\bigr]\notag\\
 \qquad{}
 +2 \hbar^{3} r V_{4}^{\prime\prime}\bigl[ r^{2} \bigl(f_{10} - f_{11}\bigr) +2 f_{8}\bigr]=0,\label{eq:feq_11}
\\
 -\hbar V_{1}\bigl[ -3 \hbar r^{2} f_{6}^{\prime} +\hbar r^{2} f_{5}^{\prime} -3 \hbar f_{4}^{\prime} +\hbar f_{3}^{\prime} -12 \hbar r f_{6} +4 \hbar r f_{5} +8 r f_{7}\bigr]\notag\\
\qquad{}
 +\hbar^{2} V_{5}\bigl[ 2 \hbar r^{3} \bigl(f_{5}^{\prime\prime} - f_{6}^{\prime\prime}\bigr) +2 \hbar r \bigl(f_{3}^{\prime\prime} - f_{4}^{\prime\prime}\bigr) -27 \hbar r^{2} f_{6}^{\prime} +29 \hbar r^{2} f_{5}^{\prime} -11 \hbar f_{4}^{\prime} +13 \hbar f_{3}^{\prime} \notag\\
 \qquad{}
 -72 \hbar r f_{6} +16 r f_{7} +80 \hbar r f_{5}\bigr] +4 \hbar^{3} V_{5}^{\prime}\bigl[ r^{3} \bigl(f_{5}^{\prime} - f_{6}^{\prime}\bigr) +r \bigl(f_{3}^{\prime} - f_{4}^{\prime}\bigr)
+6 r^{2} \bigl(f_{5} - f_{6}\bigr) \notag\\
 \qquad{}
 +2 \bigl(f_{3} - f_{4}\bigr)\bigr] +4 V_{2}\bigl[ \hbar r^{2} \bigl(f_{5}^{\prime} - f_{6}^{\prime}\bigr) +\hbar \bigl(f_{3}^{\prime} - f_{4}^{\prime}\bigr) +4 \hbar r \bigl(f_{5} - f_{6}\bigr) +4 r f_{7}\bigr]\notag\\
 \qquad{}
 +4 \hbar V_{2}^{\prime}\bigl[ -r^{2} f_{6} +2 r^{2} f_{5} -f_{4} +2 f_{3}\bigr] +2 \hbar \bigl(- \hbar^{2} V_{4}^{\prime\prime\prime} + 2 r^{2} V_{3}^{\prime}\bigr)\bigl[ r^{2} f_{5} +f_{3}\bigr]\notag\\
 \qquad{}
 +\frac{\hbar^{2} V_{4}^{\prime}}{r^{2}}\bigl[ 2 \hbar r^{4} f_{5}^{\prime\prime} +2 \hbar r^{2} f_{3}^{\prime\prime} +13 \hbar r^{3} f_{5}^{\prime} +4 r^{3} f_{7}^{\prime} +\hbar r^{3} f_{6}^{\prime} +\hbar r \bigl(f_{3}^{\prime} + f_{4}^{\prime}\bigr) +14 \hbar r^{2} f_{5} \notag\\
 \qquad{}
 +2 \hbar r^{2} f_{6} +4 r^{2} f_{7} +4 \hbar f_{3}\bigr] -\frac{2 \hbar^{2} V_{4}}{r^{2}}\bigl[ -\hbar r^{4} f_{5}^{\prime\prime\prime} -\hbar r^{2} f_{3}^{\prime\prime\prime} -10 \hbar r^{3} f_{5}^{\prime\prime} -2 r^{3} f_{7}^{\prime\prime} \notag\\
 \qquad{}
 -2 \hbar r f_{3}^{\prime\prime} -20 \hbar r^{2} f_{5}^{\prime} -8 r^{2} f_{7}^{\prime} +2 \hbar f_{3}^{\prime}\bigr] +4 \hbar V_{0}^{\prime}\bigl[ r^{2} f_{6} +f_{4}\bigr]\notag\\
 \qquad{}
 -\frac{2 \hbar^{3} V_{4}^{\prime\prime}}{r}\bigl[ r^{2} \bigl(f_{5} - f_{6}\bigr) +2 f_{3}\bigr] +2 \hbar^{3} r V_{5}^{\prime\prime}\bigl[ r^{2} \bigl(f_{5} - f_{6}\bigr) +\bigl(f_{3} - f_{4}\bigr)\bigr]\notag\\
 \qquad{}
 +4 \hbar r V_{3}\bigl[ 2 r^{2} f_{5} +3 f_{3} +f_{4}\bigr] -\frac{\hbar^{2}}{r^{2}}\bigl[ -\hbar r^{4} f_{6}^{\prime\prime\prime} -\hbar r^{2} f_{4}^{\prime\prime\prime} -10 \hbar r^{3} f_{6}^{\prime\prime} \notag\\
\qquad{}
 +2 r^{3} f_{7}^{\prime\prime} -2 \hbar r f_{4}^{\prime\prime} -20 \hbar r^{2} f_{6}^{\prime} +8 r^{2} f_{7}^{\prime} +2 \hbar f_{4}^{\prime}\bigr]=0,\label{eq:feq_12}
\\
 -2 \hbar V_{4}\bigl[ -\hbar r^{2} f_{9}^{\prime} +\hbar r^{2} \bigl(f_{10}^{\prime} - f_{11}^{\prime}\bigr) +r^{2} \bigl(f_{1}^{\prime} - f_{2}^{\prime}\bigr) +\hbar f_{8}^{\prime} -6 \hbar r f_{9} +2 r \bigl(f_{1} - f_{2}\bigr) \notag\\
\qquad{}
 +4 \hbar r \bigl(f_{10} - f_{11}\bigr)\bigr] +2 \bigl(2 \hbar^{2} r^{2} V_{5}^{\prime} - \hbar^{2} V_{4}^{\prime} - \hbar r V_{1} + 4 r V_{2}\bigr)\bigl[ r^{2} f_{9} +f_{8}\bigr]\notag\\
 \qquad{}
 -2 \hbar^{2} r V_{5}\bigl[ -2 r^{3} f_{9}^{\prime} -2 r f_{8}^{\prime} -15 r^{2} f_{9} +4 r^{2} \bigl(f_{10} - f_{11}\bigr) -3 f_{8}\bigr]\notag\\
 \qquad{}
 +\hbar r\bigl[ r \bigl(f_{1}^{\prime} - f_{2}^{\prime}\bigr) +\bigl(f_{1} - f_{2}\bigr)\bigr]=0,\label{eq:feq_29}
\\
 2 r^{3} \bigl(f_{1} + f_{2}\bigr) V_{1} +2 r^{3} \bigl(- \hbar^{2} V_{4}^{\prime\prime} + 2 r^{2} V_{3}\bigr) \bigl(f_{10} + f_{11}\bigr) -2 \hbar r^{3} V_{5}\bigl[ -\bigl(f_{1} + f_{2}\bigr) \notag\\
\qquad{}
 +12 \hbar \bigl(f_{10} + f_{11}\bigr)\bigr] -2 \hbar r^{2} V_{4}\bigl[ \hbar r \bigl(f_{10}^{\prime\prime} + f_{11}^{\prime\prime}\bigr) -\bigl(f_{1}^{\prime} + f_{2}^{\prime}\bigr) +6 \hbar \bigl(f_{10}^{\prime} + f_{11}^{\prime}\bigr)\bigr]\notag\\
\qquad{}
 -2 \hbar^{2} r^{2} V_{4}^{\prime}\bigl[ r \bigl(f_{10}^{\prime} + f_{11}^{\prime}\bigr) +4 \bigl(f_{10} + f_{11}\bigr)\bigr] +2 \hbar r^{2} \bigl(f_{1}^{\prime} + f_{2}^{\prime}\bigr)=0,\label{eq:feq_51}
\\
 2 r V_{4}\bigl[ \hbar r \bigl(f_{10}^{\prime} + f_{11}^{\prime}\bigr) +r \bigl(f_{1}^{\prime} + f_{2}^{\prime}\bigr) +2 \bigl(f_{1} + f_{2}\bigr) +5 \hbar \bigl(f_{10} + f_{11}\bigr)\bigr]\notag\\
\qquad{}
 -8 \hbar r^{3} \bigl(f_{10} + f_{11}\bigr) V_{5} +r\bigl[ r \bigl(f_{1}^{\prime} + f_{2}^{\prime}\bigr) +\bigl(f_{1} + f_{2}\bigr)\bigr]=0,\label{eq:feq_53}
\\
 -\hbar V_{4}\bigl[ 2 \hbar r^{3} f_{11}^{\prime\prime} +\hbar r^{3} f_{9}^{\prime\prime} -\hbar r f_{8}^{\prime\prime} +15 \hbar r^{2} f_{11}^{\prime} +2 r^{2} f_{2}^{\prime} +9 \hbar r^{2} f_{9}^{\prime} +\hbar r^{2} f_{10}^{\prime} -\hbar f_{8}^{\prime} \notag\\
 \qquad{}
 +12 \hbar r f_{9} +18 \hbar r f_{11} +2 \hbar r f_{10} +8 r f_{2}\bigr] -\hbar^{2} V_{4}^{\prime}\bigl[ 2 r^{3} f_{11}^{\prime} +r^{3} f_{9}^{\prime} -r f_{8}^{\prime} +5 r^{2} f_{9} \notag\\
 \qquad{}
 +8 r^{2} f_{11} -f_{8}\bigr] -4 r V_{2}\bigl[ r^{2} \bigl(f_{10} - f_{11}\bigr) +2 f_{8}\bigr] +4 r^{3} V_{3}\bigl[ r^{2} f_{11} -f_{8}\bigr] \notag\\
 \qquad{}
 +r V_{1}\bigl[ 2 r^{2} f_{1} +\hbar r^{2} \bigl(f_{10} - f_{11}\bigr) +2 \hbar f_{8}\bigr] +\hbar r^{2} \bigl(- \hbar V_{5}^{\prime} + V_{1}^{\prime}\bigr)\bigl[ r^{2} f_{9} +f_{8}\bigr] \notag\\
 \qquad{}
 -\hbar^{2} r V_{4}^{\prime\prime}\bigl[ 2 r^{2} f_{11} +r^{2} f_{9} -f_{8}\bigr] -\hbar r V_{5}\bigl[ -2 \hbar r^{3} f_{9}^{\prime} -2 \hbar r f_{8}^{\prime} -12 \hbar r^{2} f_{9} \notag\\
 \qquad{}
 -2 r^{2} f_{2} -3 \hbar r^{2} f_{11} +11 \hbar r^{2} f_{10} +6 \hbar f_{8}\bigr] -2 \hbar r f_{1}=0,\label{eq:feq_65}
\\
 -4 r V_{2}\bigl[ 2 r^{2} f_{9} +r^{2} \bigl(f_{10} - f_{11}\bigr) +4 f_{8}\bigr] -4 r^{3} V_{3}\bigl[ r^{2} f_{10} +f_{8}\bigr] +r V_{1}\bigl[ -2 r^{2} f_{2} +2 \hbar r^{2} f_{9} \notag\\
\qquad{}
 +\hbar r^{2} \bigl(f_{10} - f_{11}\bigr) +4 \hbar f_{8}\bigr] +\hbar V_{4}\bigl[ -\hbar r^{3} f_{9}^{\prime\prime} +2 \hbar r^{3} f_{10}^{\prime\prime} +\hbar r f_{8}^{\prime\prime} -11 \hbar r^{2} f_{9}^{\prime} -2 r^{2} f_{2}^{\prime} \notag\\
\qquad{}
 -\hbar r^{2} f_{11}^{\prime} +17 \hbar r^{2} f_{10}^{\prime} +4 r^{2} f_{1}^{\prime} +3 \hbar f_{8}^{\prime} -24 \hbar r f_{9} -4 r f_{2} -6 \hbar r f_{11} +12 r f_{1} \notag\\
\qquad{}
 +26 \hbar r f_{10}\bigr] +\hbar^{2} V_{4}^{\prime}\bigl[ -r^{3} f_{9}^{\prime} +2 r^{3} f_{10}^{\prime} +r f_{8}^{\prime} -3 r^{2} f_{9} +8 r^{2} f_{10} +3 f_{8}\bigr]\notag\\
\qquad{}
 +\hbar r^{2} \bigl(- 5 \hbar V_{5}^{\prime} + V_{1}^{\prime}\bigr)\bigl[ r^{2} f_{9} +f_{8}\bigr] +\hbar^{2} r V_{4}^{\prime\prime}\bigl[ -r^{2} f_{9} +2 r^{2} f_{10} +f_{8}\bigr]\notag\\
\qquad{}
 +\hbar r V_{5}\bigl[ -2 \hbar r^{3} f_{9}^{\prime} -2 \hbar r f_{8}^{\prime} -18 \hbar r^{2} f_{9} -2 r^{2} f_{1} +3 \hbar r^{2} f_{11} +5 \hbar r^{2} f_{10} -12 \hbar f_{8}\bigr]\notag\\
\qquad{}
 -\hbar r\bigl[ r \bigl(f_{1}^{\prime} - f_{2}^{\prime}\bigr) -3 f_{2} +f_{1}\bigr]=0,
\\
 \hbar V_{4}\bigl[ -4 r^{2} f_{7}^{\prime} -\hbar r^{2} \bigl(f_{5}^{\prime} - f_{6}^{\prime}\bigr) +\hbar \bigl(f_{3}^{\prime} - f_{4}^{\prime}\bigr) -6 \hbar r \bigl(f_{5} - f_{6}\bigr) -8 r f_{7}\bigr]\notag\\
\qquad{}
 -\bigl(2 \hbar^{2} r^{2} V_{5}^{\prime} - \hbar^{2} V_{4}^{\prime} - \hbar r V_{1} + 4 r V_{2}\bigr)\bigl[ r^{2} \bigl(f_{5} - f_{6}\bigr) +\bigl(f_{3} - f_{4}\bigr)\bigr]\notag\\
\qquad{}
 -\hbar^{2} r V_{5}\bigl[ 2 r^{3} \bigl(f_{5}^{\prime} - f_{6}^{\prime}\bigr) +2 r \bigl(f_{3}^{\prime} - f_{4}^{\prime}\bigr) +15 r^{2} \bigl(f_{5} - f_{6}\bigr) +3 \bigl(f_{3} - f_{4}\bigr)\bigr]\notag\\
\qquad{}
 +2 \hbar r\bigl[ r f_{7}^{\prime} +f_{7}\bigr]=0,\label{eq:feq_93}
\\
 \hbar V_{1}\bigl[ \hbar r^{2} \bigl(f_{5}^{\prime} - f_{6}^{\prime}\bigr) +\hbar \bigl(f_{3}^{\prime} - f_{4}^{\prime}\bigr) +4 \hbar r \bigl(f_{5} - f_{6}\bigr) +4 r f_{7}\bigr] -2 \hbar^{2} \bigl(2 V_{4} - 1\bigr)\bigl[ r f_{7}^{\prime\prime} +4 f_{7}^{\prime}\bigr]\notag\\
 \qquad{}
 -\hbar^{2} V_{5}\bigl[ 2 \hbar r^{3} \bigl(f_{5}^{\prime\prime} - f_{6}^{\prime\prime}\bigr) +2 \hbar r \bigl(f_{3}^{\prime\prime} - f_{4}^{\prime\prime}\bigr) +27 \hbar r^{2} \bigl(f_{5}^{\prime} - f_{6}^{\prime}\bigr) +11 \hbar \bigl(f_{3}^{\prime} - f_{4}^{\prime}\bigr) +4 r f_{7} \notag\\
 \qquad{}
 +72 \hbar r \bigl(f_{5} - f_{6}\bigr)\bigr] +\frac{\hbar^{2} V_{4}^{\prime}}{r}\bigl[ -4 r^{2} f_{7}^{\prime} +\hbar r^{2} \bigl(f_{5}^{\prime} - f_{6}^{\prime}\bigr) +\hbar \bigl(f_{3}^{\prime} - f_{4}^{\prime}\bigr) -8 r f_{7} \notag\\
\qquad{}
 +2 \hbar r \bigl(f_{5} - f_{6}\bigr)\bigr] -4 \hbar^{3} V_{5}^{\prime}\bigl[ r^{3} \bigl(f_{5}^{\prime} - f_{6}^{\prime}\bigr) +r \bigl(f_{3}^{\prime} - f_{4}^{\prime}\bigr) +6 r^{2} \bigl(f_{5} - f_{6}\bigr) \notag\\
\qquad{}
 +2 \bigl(f_{3} - f_{4}\bigr)\bigr] -4 \hbar V_{2}\bigl[ r^{2} \bigl(f_{5}^{\prime} - f_{6}^{\prime}\bigr) +\bigl(f_{3}^{\prime} - f_{4}^{\prime}\bigr) +4 r \bigl(f_{5} - f_{6}\bigr)\bigr]\notag\\
\qquad{}
 -2 \hbar \bigl(\hbar^{2} r V_{5}^{\prime\prime} + 2 V_{2}^{\prime}\bigr)\bigl[ r^{2} \bigl(f_{5} - f_{6}\bigr) +\bigl(f_{3} - f_{4}\bigr)\bigr]\notag\\
\qquad{} +4 \hbar r \bigl(f_{3} - f_{4}\bigr) V_{3} +2 \hbar^{3} r \bigl(f_{5} - f_{6}\bigr) V_{4}^{\prime\prime}=0,
\\
 -\hbar V_{1}\bigl[ \hbar r^{2} \bigl(f_{5}^{\prime} - f_{6}^{\prime}\bigr) +\hbar \bigl(f_{3}^{\prime} - f_{4}^{\prime}\bigr) +4 \hbar r \bigl(f_{5} - f_{6}\bigr) +4 r f_{7}\bigr]\notag\\
\qquad{}
 +\hbar^{2} V_{5}\bigl[ \hbar r^{2} \bigl(f_{5}^{\prime} - f_{6}^{\prime}\bigr) +\hbar \bigl(f_{3}^{\prime} - f_{4}^{\prime}\bigr) +12 r f_{7} +4 \hbar r \bigl(f_{5} - f_{6}\bigr)\bigr]\notag\\
\qquad{}
 +\frac{\hbar^{2} V_{4}^{\prime}}{r^{2}}\bigl[ \hbar r^{4} \bigl(f_{5}^{\prime\prime} - f_{6}^{\prime\prime}\bigr) +\hbar r^{2} \bigl(f_{3}^{\prime\prime} - f_{4}^{\prime\prime}\bigr) +7 \hbar r^{3} \bigl(f_{5}^{\prime} - f_{6}^{\prime}\bigr) +\hbar r \bigl(f_{3}^{\prime} - f_{4}^{\prime}\bigr) -4 r^{2} f_{7} \notag\\
\qquad{}
 +8 \hbar r^{2} \bigl(f_{5} - f_{6}\bigr) +2 \hbar \bigl(f_{3} - f_{4}\bigr)\bigr] -\frac{\hbar^{3} \bigl(2 V_{4} - 1\bigr)}{2 r^{2}}\bigl[ -r^{4} \bigl(f_{5}^{\prime\prime\prime} - f_{6}^{\prime\prime\prime}\bigr) -r^{2} \bigl(f_{3}^{\prime\prime\prime} - f_{4}^{\prime\prime\prime}\bigr)\notag\\
\qquad{}
 -10 r^{3} \bigl(f_{5}^{\prime\prime} - f_{6}^{\prime\prime}\bigr) -2 r \bigl(f_{3}^{\prime\prime} - f_{4}^{\prime\prime}\bigr) -20 r^{2} \bigl(f_{5}^{\prime} - f_{6}^{\prime}\bigr) +2 \bigl(f_{3}^{\prime} - f_{4}^{\prime}\bigr)\bigr]\notag\\
 \qquad{}
 +\hbar \bigl(- \hbar^{2} V_{4}^{\prime\prime\prime} + 2 r^{2} V_{3}^{\prime} - 2 V_{0}^{\prime} + 2 V_{2}^{\prime}\bigr)\bigl[ r^{2} \bigl(f_{5} - f_{6}\bigr) +\bigl(f_{3} - f_{4}\bigr)\bigr] +16 r V_{2} f_{7}\notag\\
\qquad{}
 +4 \hbar r V_{3}\bigl[ r^{2} \bigl(f_{5} - f_{6}\bigr) +2 \bigl(f_{3} - f_{4}\bigr)\bigr] -\frac{2 \hbar^{3} \bigl(f_{3} - f_{4}\bigr) V_{4}^{\prime\prime}}{r}=0,\label{eq:feq_95}
\\
 \hbar V_{4}\bigl[ -\hbar r^{3} f_{9}^{\prime\prime} +2 \hbar r^{3} f_{10}^{\prime\prime} +\hbar r f_{8}^{\prime\prime} -9 \hbar r^{2} f_{9}^{\prime} +15 \hbar r^{2} f_{10}^{\prime} +2 r^{2} f_{1}^{\prime} +\hbar r^{2} f_{11}^{\prime} +\hbar f_{8}^{\prime} \notag\\
 \qquad{}
 -12 \hbar r f_{9} +18 \hbar r f_{10} +2 \hbar r f_{11} +8 r f_{1}\bigr] +\hbar^{2} V_{4}^{\prime}\bigl[ -r^{3} f_{9}^{\prime} +2 r^{3} f_{10}^{\prime} +r f_{8}^{\prime} -5 r^{2} f_{9} \notag\\
 \qquad{}
 +8 r^{2} f_{10} +f_{8}\bigr] -4 r V_{2}\bigl[ r^{2} \bigl(f_{10} - f_{11}\bigr) +2 f_{8}\bigr] -4 r^{3} V_{3}\bigl[ r^{2} f_{10} +f_{8}\bigr]\notag\\
\qquad{}
 +r V_{1}\bigl[ -2 r^{2} f_{2} +\hbar r^{2} \bigl(f_{10} - f_{11}\bigr) +2 \hbar f_{8}\bigr] +\hbar r^{2} \bigl(- \hbar V_{5}^{\prime} + V_{1}^{\prime}\bigr)\bigl[ r^{2} f_{9} +f_{8}\bigr] \notag\\
 \qquad{}
 +\hbar^{2} r V_{4}^{\prime\prime}\bigl[ -r^{2} f_{9} +2 r^{2} f_{10} +f_{8}\bigr] -\hbar r V_{5}\bigl[ -2 \hbar r^{3} f_{9}^{\prime} -2 \hbar r f_{8}^{\prime} -11 \hbar r^{2} f_{11} -12 \hbar r^{2} f_{9} \notag\\
 \qquad{}
 +2 r^{2} f_{1} +3 \hbar r^{2} f_{10} +6 \hbar f_{8}\bigr] +2 \hbar r f_{2}=0.\label{eq:feq_98}
\end{gather}

\subsection*{Acknowledgements}

This work was financially supported by T\"UB\.{I}TAK under the 1001 Program (Project No. 123F161). This work forms part of the PhD thesis of the first author, carried out at Hacettepe University Graduate School of Science and Engineering. The authors would like to express their sincere gratitude to the referees for their careful and constructive evaluation of the manuscript. Their insightful comments and valuable suggestions have significantly contributed to improving the clarity, quality, and overall presentation of the paper.

\newpage

\pdfbookmark[1]{References}{ref}
\LastPageEnding

\end{document}